\PassOptionsToPackage{dvipsnames,table}{xcolor}
\documentclass[nonacm,acmsmall,screen,natbib=false]{acmart}

\usepackage{etoolbox}
\usepackage{simplebnf}
\usepackage{scalefnt}
\usepackage{mathtools}
\usepackage{amsfonts}
\usepackage{amsthm}
\usepackage{dashbox}
\usepackage[inline]{enumitem}
\usepackage{multirow}
\usepackage{multicol}
\usepackage{newfloat}

\usepackage[draft]{fixme}

\usepackage{xparse}

\usepackage{scalerel}
\usepackage{bussproofs}
\usepackage{mathpartir}

\usepackage{stmaryrd}
\usepackage{listings}
\usepackage{diagbox}
\usepackage{bold-extra}

\ExplSyntaxOn
\prg_new_conditional:Npnn \int_if_zero:n #1 { p , T , F , TF }
  {
    \if_int_compare:w \int_eval:w #1 = \c_zero_int
      \prg_return_true:
    \else:
      \prg_return_false:
    \fi:
  }
\ExplSyntaxOff

\usepackage{nicematrixcopy}

\usepackage{wrapstuff}

\hypersetup{
  hidelinks,
  colorlinks=true,
  allcolors=black,
  pdfstartview=Fit,
  breaklinks=true
}

\usepackage[noabbrev]{cleveref}
\crefname{axiom}{axiom}{axioms}

\usepackage{tikz}
\usetikzlibrary{calc,shapes,automata,positioning,scopes,math,decorations.pathmorphing,decorations.pathreplacing,decorations.markings,arrows,shadings,%
external,%
tikzmark,%
}

\usepackage[skins,raster]{tcolorbox}
\tcbsetforeverylayer{shield externalize}
\tcbuselibrary{listings}

\usepackage{array,tabularx}

\newtheorem{definition}{Definition}
\newtheorem{lemma}{Lemma}
\newtheorem{theorem}{Theorem}
\newtheorem{corollary}{Corollary}

\usepackage{verbatim}

\excludecomment{fastcomp}

\usepackage{cmds}

\RequirePackage[
datamodel=acmdatamodel,
style=acmauthoryear, 
sorting=ynt,
url=false,
maxcitenames=2,
maxbibnames=12
]{biblatex}
\setcopyright{acmlicensed}
\acmDOI{10.1145/3704887}
\acmYear{2025}
\acmJournal{PACMPL}
\acmVolume{9}
\acmNumber{POPL}
\acmArticle{51}
\acmMonth{1}
\received{2024-07-11}
\received[accepted]{2024-11-07}

\begin{document}

\title{SNIP: Speculative Execution and Non-Interference Preservation for Compiler Transformations}

\author{Sören van der Wall}
\orcid{0009-0009-4781-8583}
\affiliation{%
  \institution{TU Braunschweig}
  \city{Braunschweig}
  \country{Germany}
}
\email{s.van-der-wall@tu-bs.de}

\author{Roland Meyer}
\orcid{0000-0001-8495-671X}
\affiliation{%
  \institution{TU Braunschweig}
  \city{Braunschweig}
  \country{Germany}
}
\email{roland.meyer@tu-bs.de}



\begin{abstract}
    We address the problem of preserving non-interference across compiler transformations \emph{under speculative semantics}.
    We develop a proof method that ensures the preservation uniformly across all source programs.
    The basis of our proof method is a new form of simulation relation.
    It operates over directives that model the attacker's control over the micro-architectural state,
    and it accounts for the fact that the compiler transformation may change the influence of the micro-architectural state on the execution (and hence the directives).
    Using our proof method, we show the correctness of dead code elimination.
    When we tried to prove register allocation correct,
    we identified a previously unknown weakness that introduces violations to non-interference.
    We have confirmed the weakness for a mainstream compiler on code from the \texttt{libsodium} cryptographic library.
    To reclaim security once more, we develop a novel static analysis that operates on a product of source program and register-allocated program.
    Using the analysis, we present an automated fix to existing register allocation implementations.
    We prove the correctness of the fixed register allocations with our proof method.
\end{abstract}

\maketitle

\section{Introduction}
\label{section:intro}

Cryptographic implementations must satisfy two conflicting requirements:
They must compute highly performant to be of practical use and be absolutely secure for use in critical systems.
Implementations optimize performance with knowledge about underlying micro-architectural hardware features such as memory access patterns that improve cache usage.
Security, however, is threatened by side-channel attacks that exploit precisely these hardware features to leak sensitive information \parencite{brumleyRemoteTimingAttacks2005}.
To mitigate side-channel attacks, leakage of sensitive data needs to be eliminated.
This confronts the programmer with two challenges:
First, semantics of source-level languages do not model leaks produced by side-channels \parencite{vuReconcilingOptimizationSecure2021}.
And second, even if the source-level code is secure, incautious implementation of compiler optimizations can insert new leakage, rendering efforts to secure the source program useless \parencite{bartheFormalVerificationConstanttime2019,simonWhatYouGet2018}.
Developers address these problems with coding guidelines such as \emph{constant time programming} and disabling compiler optimizations.
But following the guidelines is non-trivial and overlooked mistakes corrupt the guarantee for security \parencite{alfardanLuckyThirteenBreaking2013}.
At the same time, disabling compiler optimizations is dissatisfactory.
Formal methods have shown to help with the challenges:
The first challenge is overcome with novel leakage semantics that model side-channel leakage within the programming language's semantics \parencites{molnarProgramCounterSecurity2006,bartheSystemlevelNoninterferenceConstanttime2014}.
For the second, novel proof methods for compilers under leakage semantics provide a guarantee that side-channel security of the source program carries over to the executable \parencites{bartheSecureCompilationConstantResource2021,bartheSecureCompilationSideChannel2018,bartheFormalVerificationConstanttime2019,bartheStructuredLeakageApplications2021}.
Sadly, the recent discovery of Spectre attacks \parencite{kocherSpectreAttacksExploiting2019,canellaSystematicEvaluationTransient2019} again presents a hardware feature that leakage semantics fall short of:
\emph{Speculative execution} produces side-channel leakages not captured by leakage semantics.
This means both challenges were open again, and the verification community was quick to address the first:
The development of speculative execution semantics has already taken place \parencite{guarnieriHardwareSoftwareContractsSecure2021,cauligiConstanttimeFoundationsNew2020},
and formal tools find speculative side-channel leakages
or even prove their absence (cf. \Cref{section:related}).
Provably correct compilation under speculative execution semantics, however, remains an unsolved challenge that we address in this paper.
It is the challenge that we address in this paper.

\subsection{Background}
Before we detail our contributions,
we position our paper in the field of \emph{formally verified cryptography} and provide background on \emph{Spectre~attacks} and \emph{mitigations}.
We outline \emph{speculative execution semantics}, \emph{non-interference} as the property that guarantees a program's side-channel security even under speculation,
%
and \emph{provably correct compilation} for leakage semantics without speculation.

\begin{code}[t]
    \begin{codestackbox}{0.35}
        {%
        }
        {%
            \ttfamily%
            \begin{tabular}{rl}
                0x48 & \colvar{stk} \\
                0x40 & \colvar{buf}
            \end{tabular}
            \tcbsubtitle[after skip=0pt, before skip=2pt]{Heap}
            \begin{NiceTabular}{rl!{\quad}l}
                0xE0 &      &\Block[borders={left}]{4-1}{sensitive \\ data} \\
                .... &&\\
                0xC8 & \colvar{sec} &\\
                0xC0 &&
            \end{NiceTabular}
        }
chacha20(..., uint8* sec, uint8 bytes) {
  uint8 buf[8];
  |\label{line:spill}\codeframe{startspill}{endspill}{mauve}\tikzmark{startspill}|stk[0] = bytes;|\tikzmark{endspill}|
  |\label{line:for}|for (i = 0; i < 8; i++)
    |\label{line:for_body}|{ buf[i] = sec[i]; }
  ...
  |\label{line:fill}\codeframe{startfill}{endfill}{mauve}\tikzmark{startfill}|bytes = stk[0];|\tikzmark{endfill}|
  |\label{line:leak}|if (bytes < 64) {...}
  ...
}
\end{codestackbox}
\caption{\label{code:specv1}
    Spectre-\spht{}.
    Registers are \colreg{orange} and memory variables \colvar{teal}.
    The stack contents are shifted by the stack pointer to appear constant.
    \textcolor{mauve}{Framed} instructions were inserted by register allocation.
}
\end{code}
\subsubsection*{Formally verified cryptography}
The field of formally verified cryptography aims to provide cryptographic implementations
that are secured not only by trust in the developer but in a machine-checkable proof of correctness and security.
In order to achieve this goal, implementations are carefully crafted
and three main areas of research pursue different subgoals \cite{barbosaSoKComputerAidedCryptography2021}:
\begin{enumerate*}[label=\bfseries (\roman*)]
    \item \label{item:cryptodesign} \emph{Cryptographic protocol design} aims to provide proofs that the cryptographic protocol in itself does not reveal secrets to adversarial protocol participants, among other properties.
    \item \label{item:functionalcorrect} \emph{Correct and performant implementation} of the protocol aims to prove the implementation functionally correct.
    \item \label{item:implsec} \emph{Implementation security} investigates the compilation of implementations and the execution of binaries on real hardware in order to prove the absence of attacks
        that stem from the discrepancy between idealized program semantics and actual hardware semantics.
\end{enumerate*}
This paper belongs to \labelcref{item:implsec}:
We assume that source programs correctly implement formally verified protocols, i.e.\ steps \labelcref{item:cryptodesign,item:functionalcorrect} are completed.
We investigate whether compiler transformations preserve side-channel security.
To that end, our formal semantics models side-channel leakages and speculative execution, the micro-architectural hardware components that enable the recently discovered Spectre attacks.

\subsubsection*{Spectre attacks and mitigations} Spectre attacks observe side-channel leakages that are produced during speculative execution.
Speculative execution allows the processor to speculatively execute instructions from the pipeline even though they still have unevaluated instruction parameters.
When the processor detects a misspeculation, i.e.\ it assumed incorrect values for the unevaluated parameters, it rolls back execution to erase its effect.
Rollbacks are invisible to typical source-level semantics, but parts of the micro-architecture such as the cache-state are not reverted.
This creates side-channel leakage observable to an attacker during the speculative execution of instructions.
Spectre attacks target this in the following way: 
\begin{enumerate*}[label=\bfseries(\roman*)]
    \item \label{item:train} Train some micro-architectural component to speculatively execute a code fragment that
    \item \label{item:loadsec} under misspeculation brings sensitive data into a processor's register,
    which \item \label{item:leak} is leaked through side-channels.
\end{enumerate*}
The prominent example is Spectre-\textsf{PHT} \parencite{kocherSpectreAttacksExploiting2019},
whose source of speculation is the processor's \emph{branch prediction} unit (Prediction History Table),
and the side-channel leakage happens via the cache \cite{yaromFLUSHRELOADHigh2014,liuLastLevelCacheSideChannel2015} or the program counter \cite{molnarProgramCounterSecurity2006}.
\Cref{code:specv1} demonstrates the attack:%
\footnote{The attack on this code is unlikely to execute on actual hardware because the specific speculation patterns would be hard to train. We chose it because it also demonstrates a new vulnerability in register allocation that we present in this paper.}
The code is intended to load an 8-byte chunk from \texttt{\colvar{sec}},
which points into a stream of sensitive data,
and to store it into a stack-local buffer \texttt{\colvar{buf}} in order to later perform computation on it.
\labelcref{item:train} The attacker might train the branch predictor in a way
that it speculates the \texttt{for}-conditional in \Cref{line:for} and executes \Cref{line:for_body} an additional time even though \texttt{\colreg{i} = 8}.
\Cref{line:for_body} then stores sensitive data from \texttt{\&\colvar{sec} + 8}, say $\val$, to \texttt{\&\colvar{buf} + 8} which aliases with $\text{\texttt{\&\colvar{buf} + 8 = \&\colvar{stk}}}$.
The speculative execution might continue with \Cref{line:fill}, where \labelcref{item:loadsec}~$\val$~is loaded into \texttt{\colreg{bytes}}.
\labelcref{item:leak} The register is then used in \Cref{line:leak},
where the branching condition is leaked, disclosing to the attacker whether $\val < 64$.

The de facto approach to avoid Spectre attacks are hardware and software \emph{mitigations}.
In hardware, a simple mitigation is to disable specific speculation sources using control registers.
This penalizes performance as it disables the optimization for the whole program, even when other parts of the program do not operate on sensitive data.
Software mitigations have received more attention, especially for branch prediction (Spectre-PHT, cf.\ \Cref{section:related}, \textsf{\textit{Tools}}),
because disabling branch prediction has severe impact on performance \cite[Evaluation]{vassenaAutomaticallyEliminatingSpeculative2021}.
Spectre-PHT has two known software mitigations:
Speculation fence insertion and speculative load hardening \cite{zhangUltimateSLHTaking2023,carruthSpeculativeLoadHardening}.
Speculation fences $\IWsfence$ instruct the processor to stop speculation and wait until all instruction's unevaluated parameters are resolved
before either continuing computation in case of a correct prediction or rolling back in case of a misprediction.
This prevents instructions following $\IWsfence$ to be executed speculatively altogether.%
\footnote{This is idealized: In \texttt{x86}, for example, the instruction is realized with a memory fence \texttt{LFENCE}, which only executes after all loaded parameters to instructions are resolved - stopping the so far known speculation sources.}
The mitigation is applicable to all known kinds of speculation sources.
Speculative load hardening $\IWslh{\rega}$ is a mitigation unique to branch prediction.
Executing $\IWslh{\rega}$ wipes the contents of register $\rega$ in case of a branch misprediction, but does not stop speculative execution.
In case of correct prediction or non-speculative execution, it leaves the register contents unchanged.
In the binary, this semantics is achieved by tying the contents of $\rega$ to a previous branching condition $\mathtt{cond}$ via a data dependency (in the sense of $\IWasgn{\rega}{\inlineife{\mathtt{cond}}{\rega}{0}}$).%
\footnote{In \texttt{x86}, a \texttt{cmov} instruction is used which does not introduce control-flow branching, so branch prediction will not speculate.}
This forces the processor to evaluate $\mathtt{cond}$ before assigning a value to $\rega$.
The processor is not guaranteed to stop speculation immediately upon learning the correct value for $\mathtt{cond}$, but the value in $\rega$ is now safe to be leaked.
For other speculation sources (\Cref{section:related}, \textsf{\textit{Speculation Sources}}), speculative load hardening does not work because no similar data dependency is known.
The Spectre attack from \Cref{code:specv1} is mitigated by inserting either $\IWsfence$ or $\IWslh{\makereg{bytes}}$ between \Cref{line:fill} and \labelcref{line:leak}.

\subsubsection*{Speculative execution semantics and Non-interference}

Speculative execution semantics extend leakage semantics by speculation.
The achievement of leakage semantics is to incorporate a model leakage observable to the attacker in the semantics.
The observable side-channel leakage depends on the leakage model.
Common is the constant-time model which exposes the addresses of memory accesses and the program counter as leakage to the attacker \parencite{bartheSystemlevelNoninterferenceConstanttime2014,guarnieriHardwareSoftwareContractsSecure2021}.
Transitions in leakage semantics (without speculation) are of the form $\statec \xrarrow{\leak}{} \stated$.
They prompt a transition from $\statec$ to $\stated$ while capturing attacker-visible effects on the micro-architectural state in the \emph{leakage observation} $\leak$.
Side-channel security can now be formulated as a property on the program's leakage semantics.
For that, a relation declares initial states as indistinguishable to the attacker when they differ only w.r.t.\ sensitive data unknown to the attacker.
The property is non-interference, which requires that the executions from indistinguishable initial states produce equal leakages.
Non-interference guarantees side-channel security: leakages cannot depend on sensitive data in any way.
Without speculation, the running example \Cref{code:specv1} satisfies non-interference under the constant-time leakage model:
The control flow is not dependent on the secret $\makevar{sec}$ and the addresses of memory accesses (\Cref{line:spill,line:for_body,line:fill}) are independent as well.

The extension to speculative execution semantics came with a new challenge: Non-determinism.
Whether the processor mispredicts and when it detects misprediction is highly hardware dependent and potentially even under the attacker's influence.
As a result, there is not a single execution but instead a set of possible executions, each with a different sequence of leakages.
A transition in speculative execution semantics is of the form~$\statec \ntrans{\leak}{\direct} \stated$.
Again, $\leak$ is the attacker-observable leakage.
What is new is the \emph{directive}~$\direct$ that models the attacker's control over speculation \parencite{cauligiConstanttimeFoundationsNew2020,bartheHighAssuranceCryptographySpectre2021}.
The directives determine the program's speculation behavior.
They provide an abstraction of the micro-architecture that the attacker can use to steer the execution whenever it depends on the micro-architectural state.
In our example, the attacker steers speculation with the following sequence of directives in order to lead execution to the leakage of sensitive data:
\[
    \underset{\text{\Cref{line:spill}}}{\vphantom{)}\Dstep}
    \quad\lseq\quad\underset{\text{\Cref{line:for,line:for_body}}}{{(\Dif\;\lseq\;\Dstep)}^8}
    \quad\lseq\quad \underset{\text{\Cref{line:for}}}{\vphantom{)}\Dspec}
    \quad\lseq\quad \underset{\text{\Cref{line:for_body}}}{\vphantom{)}\Dstore{\texttt{\colvar{stk}}}{0}}
    \quad\lseq\quad \underset{\text{\Cref{line:for,line:fill,line:leak}}}{\vphantom{)}\Dif\;\lseq\;\Dstep\;\lseq\;\Dstep} .
\]
The intuition is the following.
The first instruction is a memory access which cannot be influenced by the attacker, denoted by the directive~$\Dstep$.
Then, the attacker steers execution so that the correct branch is taken 8~times: Directive $\Dif$ executes the correct branch and~$\Dstep$ executes the memory assignment inside the loop.
The attacker then chooses to begin a misspeculation with directive $\Dspec$ which enters the loop once more.
This leads to an unsafe memory access during the additional loop iteration,
where we let the attacker choose actual memory location with $\Dstore{\texttt{\colvar{stk}}}{0}$.
The remaining sequence leads the execution to the leaking instruction.

In order to phrase non-interference on speculative execution semantics,
the idea is to compare executions where the sequence of directives along the executions are equal.
Then, the attacker trained the hardware in the same way and can be sure that observed differences in leakage are due to sensitive data.
We define our speculative execution semantics (\Cref{section:language}) and non-interference property (\Cref{section:properties}) in this spirit.

\subsubsection*{Secure compilation}

\begin{figure}
    \tikzset{
        prewhite/.style = {
            preaction={draw,line width=3.5pt,white}
        },
        samepoint/.style args={#1}{
            draw,line width=0.7pt,#1,
            preaction={draw,line width=3pt,white,
                preaction={draw,line width=4.4pt,#1}
            }
        },
        equality/.style = { double, double distance=1.8pt },
        simrel/.style = { line width=1pt, postaction={decorate},
            decoration = { markings, mark=at position 0.4 with {
                    \node[fill=white] {$\simrelr$};
            }},
        },
        leaktf/.style = { line width=1pt, postaction={decorate},
            decoration = { markings, mark=at position 0.35 with {
                    \node[fill=white] {$\dtf$};
            }},
        },
        ntranscube/.style = {->,line width=1pt,dashed,dash pattern=on 3pt off 3pt},
        given/.style = {},
        new1/.style = { mauve },
        new2/.style = { NavyBlue },
        new3/.style = { YellowGreen },
        new4/.style = { BurntOrange },
        result/.style = { lightgray },
        ignoreheight/.style = { text depth=0, text height=1ex }
    }
    \begin{center}
        \hfill
        \scalebox{0.8}[0.8]{
            \begin{tikzpicture}[scale=0.3]
                \tikzmath{\disthorizontal = 3; \distvertical = 2.5; }

                \node (std1) {$\stated_1$};
                \node[above right=of std1] (std2) {$\stated_2$};
                \node[left=\disthorizontal of std1] (stc1) {$\statec_1$};
                \node[left=\disthorizontal of std2] (stc2) {$\statec_2$};
                \node[below=\distvertical of std1] (stf1) {$\statef_1$};
                \node[below=\distvertical of std2,new1] (stf2) {$\statef_2$};
                \node[below=\distvertical of stc1] (ste1) {$\statee_1$};
                \node[below=\distvertical of stc2] (ste2) {$\statee_2$};

                \draw[samepoint] (std1) -- (std2);
                \draw[samepoint] (stc1) -- (stc2);
                \draw[samepoint=result] (stf1) -- (stf2);
                \draw[samepoint=result] (ste1) -- (ste2);

                \draw[simrel] (stc2) -- (std2);
                \draw[simrel] (stf1) -- (ste1);
                \draw[simrel,new1] (ste2) -- (stf2);

                \draw[ntranscube] (stc1) -- node[right,pos=0.7] (v1) {$\tracek$} (ste1);
                \draw[ntranscube] (stc2) -- node[right,pos=0.7] (v2) {$\tracek$} (ste2);

                \draw[prewhite,simrel] (std1) -- (stc1);

                \draw[ntranscube,new1] (std2) -- node[left,pos=0.7,ignoreheight] (u2) {$\tracel$} (stf2);
                \draw[ntranscube,prewhite] (std1) -- node[left,pos=0.7,ignoreheight] (u1) {$\tracel$} (stf1);
            \end{tikzpicture}
        }
        \hfill
        \scalebox{0.8}[0.8]{
        \begin{tikzpicture}
            \tikzmath{\disthorizontal = 3; \distvertical = 2.5; }

            \node (std1) {$\nstated_1$};
            \node[above right= of std1] (std2) {$\nstated_2$};
            \node[left=\disthorizontal of std1] (stc1) {$\nstatec_1$};
            \node[left=\disthorizontal of std2] (stc2) {$\nstatec_2$};
            \node[below=\distvertical of std1] (stf1) {$\nstatef_1$};
            \node[below=\distvertical of std2,new1] (stf2) {$\nstatef_2$};
            \node[below=\distvertical of stc1] (ste1) {$\nstatee_1$};
            \node[below=\distvertical of stc2] (ste2) {$\nstatee_2$};

            \draw[samepoint] (std1) -- (std2);
            \draw[samepoint] (stc1) -- (stc2);
            \draw[samepoint=result] (stf1) -- (stf2);
            \draw[samepoint=result] (ste1) -- (ste2);

            \draw[simrel] (stc2) -- (std2);
            \draw[simrel] (stf1) -- (ste1);
            \draw[simrel,new1] (ste2) -- (stf2);

            \draw[ntranscube] (stc1) -- node[right,pos=0.7] (v1) {$\dtpair{\tracek}{\dtracee}$} (ste1);
            \draw[ntranscube] (stc2) -- node[right,pos=0.7] (v2) {$\dtpair{\tracek}{\dtracee}$} (ste2);

            \draw[prewhite,simrel] (std1) -- (stc1);

            \draw[ntranscube,new1] (std2) -- node[left,pos=0.7,ignoreheight] (u2) {$\dtpair{\tracel}{\dtraced}$} (stf2);
            \draw[leaktf,new1] (v2) -- (u2);
            \draw[ntranscube,prewhite] (std1) -- node[left,pos=0.7,ignoreheight] (u1) {$\dtpair{\tracel}{\dtraced}$} (stf1);
            \draw[leaktf,prewhite] (u1) -- (v1);
        \end{tikzpicture}}%
        \hfill\vphantom{a}
    \end{center}
    \caption{\label{fig:cubes} Constant-time cubes. Left: leakage transforming; Right: directive transforming ($\dtf$).}
\end{figure}

The goal of secure compilation is to prove that a compiler pass preserves non-interference from source to target program.
The current methods for leakage semantics without speculation draw from classical methods for compiler correctness
which utilize \emph{simulation} in order to argue that the target program's executions
can be found in the source program's semantics \parencite{mccarthyCorrectnessCompilerArithmetic1967,leroyFormallyVerifiedCompiler2009,}.
A traditional simulation is a relation $\simrel$ between the target program's states and the source program's states.
Whenever a target state $\stated$ is simulated by a source state $\statec$, $\stated \simrel \statec$,
and has a transition $\stated \xrarrow{o}{} \statef$,
where $o$ is an observable environment interaction,
then $\statec$ has to have a next transition $\statec \xrarrow{o}{} \statee$
so that $\statef \simrel \statee$.
Simulation ensures that the target program's execution produces the same observable environment interactions as the source program.
For leakage semantics, a notion of simulation needs more:
Compilers aim to preserve observable environment interactions such as system-calls,
but they regularly modify side-channel leakage which creates a difference in leakage between source and target program.
\emph{Leakage transformation} \parencite{costanzoEndtoendVerificationInformationflow2016,bartheStructuredLeakageApplications2021,bartheSecureCompilationSideChannel2018} solves this issue:
Given the leakages along a source program's execution,
the simulation also provides a way to transform the observable leakage
into the observable leakage of the corresponding target program's execution.
In order to preserve non-interference, a simulation with leakage transformation needs to satisfy the \emph{constant-time cube diagram}.
The cube diagram can be seen in \Cref{fig:cubes} on the left.
It looks at two pairs of states related by simulation $\stated_1 \simrel \statec_1$ and $\stated_2 \simrel \statec_2$.
Then, if $\statec_1$'s next transitions leak a sequence~$\tracek$ and so do $\statec_2$'s next transitions (black)
then the next transitions' leakage from $\stated_1$ and $\stated_2$ must coincide as well (\textcolor{mauve}{purple}).
The target leakage $\tracel$ does \emph{not} need to be equal to the source leakage $\tracek$.
\Cref{table:compilerworks} (left) lists compilers employing simulations with leakage transformation (no speculative semantics) that satisfy the constant-time cube diagram.

In the speculative execution setting,
compilers so far try to avoid Spectre attacks by running compiler passes that insert the mitigations discussed above.
These passes are among the last passes in the compiler chain in order to avoid the removal of mitigations by other passes.
They aim to eliminate speculative leakage by inserting mitigations conservatively, which entails significant performance overhead.
Even worse: Efforts to improve performance were flawed, again leading to insecure executables \parencite{patrignaniExorcisingSpectresSecure2021}.
Passes that insert mitigations are desirable because they free the developer from having to think about speculation:
Before the compiler runs the mitigation pass,
the semantics can be considered speculation-free.
However, recent research suggests that in order to obtain minimal performance overhead,
the developer needs an interface to control inserted mitigations \parencite{shivakumarTypingHighSpeedCryptography2023}.
This means a new proof method for compilation is needed that works when both source and target program operate under speculative execution semantics.

\begin{tablecites}[table:compilerworks]{l||X[r]|X[r]||X[r]|X[r]}{
        Left: Compilers that preserve side-channel security under leakage semantics;
        \textsf{Compcert} \cite{leroyFormallyVerifiedCompiler2009}, \textsf{Jasmin} \cite{bartheStructuredLeakageApplications2021}.
        Right: Proof methods for compilers with speculative execution semantics; \textsf{Ex. Spectres} \cite{patrignaniExorcisingSpectresSecure2021}.
        \textcolor{lightgreen}{Green} parameters are more expressive.
    }
    {\CodeBefore
        \tikz \path [right color=lightgreen,left color=base] ([shift={(0.5,0)}]7-|4) rectangle (8-|5) ;
    \Body}
            & Compcert 
            & Jasmin 
            & Ex. Spectres
            & This Paper
            \\ \hline
Property            & $\propnip$    & $\propnip$    & $\propstsp$        & \cellcolor{lightgreen} $\propsnip$   \\
Simulation          & LT-Sim        & LT-Sim        & LO-Sim             & \cellcolor{lightgreen} DT-Sim        \\
Speculation         & ---           & ---           & SW, M, TO          & \cellcolor{lightgreen} US, ST        \\
Non-Det             & No            & No            & No                 & \cellcolor{lightgreen} Yes           \\
Spec Source         & ---           & ---           & PHT                & PHT                                  \\
Memory Safety       & S             & S             & U                  & S                                    \\
Passes              & Full          & Full          & SLH \& Fence Ins   & DC \& RA \\
\end{tablecites}

To the best of our knowledge, \citet{patrignaniExorcisingSpectresSecure2021} is the only work so far to \makebox{\emph{(dis-)prove}} compiler correctness under speculative execution semantics.
The authors target specifically compiler passes that insert mitigations and discovered the aforementioned flaws in fence insertion and speculative load hardening.
Being tailored towards mitigations, they employ assumptions on the setting that do not hold in general and that we overcome in our development.
We detail the differences to our work in \Cref{table:compilerworks} (right):
\begin{enumerate*}[font=\bfseries]
    \item[Property:]
        The first difference lies in the property ensuring side-channel security.
        While speculative non-interference preservation ($\propsnip$) is the goal,
        their proof method preserves speculative taint safety ($\propstsp$).
        Taint safety is a safety property that soundly approximates non-interference.
        While $\propstsp$ is not an approximation of $\propsnip$, the method is appropriate for analyzing mitigation passes.
        Our method is designed to prove $\propsnip$, instead.
    \item[Speculation and Non-Det:]
        A bigger difference is the speculative execution semantics:
        They assume a speculation window (\textsf{SW})
        that limits the number of steps speculatively executed after a misprediction before a rollback occurs.
        Speculation windows are a restriction of the speculative execution semantics.
        While the restriction is reasonably chosen with respect to current hardware,
        it presents an under-approximation of the speculative execution semantics.
        They further assume that the semantics always mispredicts branches (\textsf{M}) to maximize speculative execution.
        Mispredict semantics are no further restriction of the semantics,
        as maximizing speculative execution also maximizes the side-channel leakages produced.
        Together, these assumptions form a \emph{deterministic} restriction of the full speculative execution semantics.
        The focus on compiler mitigations also led them to the assumption that the \emph{source language is speculation-free},
        meaning the speculative semantics are target program only (\textsf{TO}).
        In this paper, we deal with full, unbounded speculative execution semantics (\textsf{US}) and the induced non-determinism in both source and target semantics~(\textsf{ST}).
    \item[Memory:]
        Our work is presented for structured memory (\textsf{S}) and the assumption that source programs are memory safe when executed under speculation-free semantics.
        Memory safety is a common assumption for compilers, as unsafe memory accesses are usually considered undefined behavior in source semantics.
        Our proof method also works with unstructured memory~(\textsf{U}),
        but we also present a static analysis whose presentation immensely benefits from structured memory.
        This led us to present all of our work with structured memory as the concepts behind the proof method stay the same.
    \item[Simulation:]
        Mitigations insert speculation barriers and do not change the code otherwise.
        As the source program in \citet{patrignaniExorcisingSpectresSecure2021} does not speculate,
        the leakages of the source program will still be fully present and unchanged in the target program.
        Because the target program is executed with speculative execution semantics,
        there can, however, be additional leakages present in the target program.
        This leads their work to consider leave-out simulations (\textsf{LO-Sim}),
        where source leakages are equal to target leakages with additional speculative leakages.
        For general compiler transformations and unbounded speculation, we introduce the more general directive transforming simulations (\textsf{DT-Sim}).
    \item[Passes:]
        Their work targets the compiler passes that insert the software mitigations against spectre from above:
        Speculative load hardening and fence insertion.
        Our work targets two general-purpose compiler transformations:
        Dead code elimination~(\textsf{DC}) and transformations from the register allocation phase~(\textsf{RA}).
\end{enumerate*}

\subsection{Contributions}
In this paper, we present \emph{snippy simulations},
a novel proof method for preservation of non-interference under speculative execution semantics.
The main challenge to overcome with speculation is the non-determinism in directive semantics.
First, the definition of simulations becomes more involved:
Deterministic semantics have the advantage that simulations are always bi-simulations \parencite{milnerAlgebraicDefinitionSimulation1971}.
A simulation for deterministic semantics synchronizes the (singular) execution of the source program and the execution of the target program.
For non-deterministic semantics such as speculative execution semantics,
each of the target program's executions must be synchronized with a source program's execution.
Second, similar to how compilers do not preserve leakage,
they also modify where the attacker can steer the computation:
A sequence of directives to steer execution on the source program may be unfit to steer any execution on the target program (the compiler may change instructions and with them the available directives to steer execution also change).
We address this issue by introducing the new concept of \emph{directive transformations}.
Directive transformations match every executable sequence of directives in the target program
with a sequence of directives executable in the source program.
We then embed directive transformations into a new constant-time cube for speculative execution semantics (\Cref{fig:cubes}, right).
It is our contribution to make the constant-time cube applicable for speculative semantics.

We demonstrate our proof method on two compiler transformations: Dead code elimination and register allocation.
This is the first time that these compiler passes have been formally analyzed under speculative execution semantics and to our surprise,
we found a serious vulnerability in the transformations performed during the register allocation phase.
The register allocation phase is located in the compiler chain where a hardware-independent intermediate representation is replaced by a concrete ISA.
It transforms virtual registers into hardware registers and has to spill excess registers to the stack:
In \Cref{code:specv1}, the framed instructions constitute a spill of the register~$\makereg{bytes}$.
The program before register allocation (without \Cref{line:spill,line:fill}) has no side-channel leakage of sensitive data under speculative execution semantics.
The program after register allocation (with \Cref{line:spill,line:fill}) is vulnerable to the Spectre attack presented above.
This vulnerability is not unique to a singular register allocator, but more generally stems from the spilling transformation performed in this phase.
In order to fix the transformations performed,
we present a novel static analysis on a product of source program (before register allocation) and target program (after register allocation)
that finds problematic speculative leakages introduced by spilling transformations.
We then fix the problematic transformations by inserting as few mitigations as possible.
The fix is automated and applies to every existing register allocator.
We then show that the fixed transformations are secure by once more applying our proof method.

In short, we address the problem of non-interference preservation for compiler passes under speculative execution semantics.
We make the following contributions:
\begin{enumerate}[label=$\blacktriangleright$]
    \item We develop a proof method for non-interference preservation \emph{under speculative execution semantics} based on simulation relations.
        Technically, we address non-determinism from speculation with \emph{directive transformations}.
    \item We \emph{demonstrate} our proof method on dead code elimination.
    \item We show that \emph{register allocation does not preserve non-interference} under speculative execution semantics.
        We confirm this for all register allocators of the \texttt{LLVM} compiler on code from the widely used \texttt{libsodium} cryptographic library.
    \item We propose a static analysis that \emph{finds and automatically fixes} the vulnerabilities introduced by any register allocator.
        We apply our proof method to show that the fixed transformation preserves speculative non-interference.
\end{enumerate}

\subsubsection*{Outline}
\Cref{section:language} introduces our formulation of speculative execution semantics with leakages and directives.
\Cref{section:properties} defines speculative non-interference preservation ($\propsnip$).
Our proof method is presented in \Cref{section:simulation}, and we apply it in \Cref{section:deadcode} to prove that dead code elimination preserves non-interference.
We then analyze the vulnerability we found in the register allocation transformations and present our fix in \Cref{section:regalloc}.
We finish with related works in \Cref{section:related} and discuss future prospects in \Cref{section:conclusion}.

\section{Language Model}
\label{section:language}

We introduce our programming language and its speculative execution semantics.
A program is a mapping $\prog : \pcsof{\prog} \to \instrs$ from program counters to instructions.
The initial program counter is $\entry \in \pcsof{\prog}$.
Instructions $\inst \in \instrs$ are of the following form.
We denote registers by $\rega, \regb, \regc, \regd \in \regs$
and memory variables by $\vara, \varb \in \vars$.
The subscripts $\pcsuc \in \pcsof{\prog}$ are the successors of $\inst$.
We may also call $\pcsuc$ a successor of $\pc$ instead, if it is a successor of $\progof{\pc}$.
\begin{align*}
    \inst \in \instrs \;& \Coloneqq \;
    \Iexit
    \bnfalt \Inop{\pcsuc}
    \bnfalt \Iasgn{\rega}{\regb \op \regc}{\pcsuc}
    \bnfalt \Iload{\rega}{\vara}{\regb}{\pcsuc}
    \bnfalt \Istore{\vara}{\regb}{\regc}{\pcsuc}
    \bnfalt \Iif{\regb}{\pcsuc_{\btrue}}{\pcsuc_{\bfalse}}
    \bnfalt \Isfence{\pcsuc}
    \bnfalt \Islh{\rega}{\pcsuc}
\end{align*}
The instructions are return (or exit), no-op, assignment, load, store, conditional branching,
and the software mitigations for Spectre, speculation fences and speculative load hardening.

\subsubsection*{Semantics}
We introduce two semantics:
Speculation-free $\statec \trans{\leak}{\direct} \stated$ (\Cref{rules:spec-free})
and speculative $\nstatec \ntrans{\leak}{\direct} \nstated$ (\Cref{rules:spec}).
The transitions are labelled by leakage $\leak \in \leaks$
and directives $\direct \in \directs$.
Our leakages stem from the constant-time leakage model which leaks the addresses of memory accesses as well as branching conditions.
Directives resolve non-determinism for the speculative semantics, i.e.\ when speculation starts and ends or where unsafe memory accesses (out-of-bounds) actually access memory.
Directives are considered under the attacker's control.

A speculation-free state is a tuple $(\pc, \rasgn, \masgn) \in \states$
that tracks the program counter $\pc \in \pcs$,
register contents $\rasgn : \regs \to \vals$,
and memory $\masgn : \mems \to \vals$.
The semantics is given in \Cref{rules:spec-free} and is fairly standard.
Memory is structured and without dynamic allocation.
Each variable $\vara$ has static size and for an offset address $\adr \in \adrs \subseteq \vals$,
we write $\adr \in \sizeof{\vara}$ to indicate that $\adr$ lies within $\vara$'s size.
The memory is $\mems = \setcond{(\vara, \adr)}{\adr \in \sizeof{\vara}}$.
Following the leakage model, loads (\labelcref{rule:load,rule:load-unsafe}) and stores (\labelcref{rule:store,rule:store-unsafe}) leak the accessed address used via $\Eload{\adr}$ and $\Estore{\adr}$,
and branching (\labelcref{rule:branch}) leaks its condition with $\Eif{\bvalue}$.
The directives $\Dload{\varb}{\adrp}$ (\labelcref{rule:load-unsafe}) and $\Dstore{\varb}{\adrp}$ (\labelcref{rule:store-unsafe}) let the attacker control the address for unsafe memory accesses.
%
%

\begin{figure}
    \begin{ruleframes}
        \begin{ruleframe}[label=rules:spec-free]{Speculation-free Semantics}
            \definerule{nop}
            {\progof{\pc} = \Inop{\pcsuc}}
            {(\pc, \rasgn, \masgn) \trans{\Enone}{\Dstep} (\pcsuc, \rasgn, \masgn)}
            {\label{rule:nop}}

            \definerule{asgn}
            {
                \progof{\pc} = \Iasgn{\rega}{\regb \op \regc}{\pcsuc}
                \\
                \val = \rasgnof{\regb} \op \rasgnof{\regc}
            }
            {(\pc, \rasgn, \masgn) \trans{\Enone}{\Dstep} (\pcsuc, \subst{\rasgn}{\rega}{\val}, \masgn)}
            {\label{rule:asgn}}

            \definerule{branch}
            {
                \progof{\pc} = \Iif{\regb}{\pcsuc_{\btrue}}{\pcsuc_{\bfalse}}
                \\
                \bvalue = (0 \sameas \rasgnof{\regb})
            }
            {(\pc, \rasgn, \masgn) \trans{\Eif{\bvalue}}{\Dif} (\pcsuc_{\bvalue}, \rasgn, \masgn)}
            {\label{rule:branch}}

            \definerule{load}
            {
                \progof{\pc} = \Iload{\rega}{\vara}{\regb}{\pcsuc}
                \\
                \adr = \rasgnof{\regb} \in \vsizeof{\vara}
                \\
                \val = \masgnof{\vara}{\adr}
            }
            {(\pc, \rasgn, \masgn) \trans{\Eload{\adr}}{\Dstep} (\pcsuc, \subst{\rasgn}{\rega}{\val}, \masgn)}
            {\label{rule:load}}

            \definerule{store}
            {
                \progof{\pc} = \Istore{\vara}{\regb}{\regc}{\pcsuc}
                \\
                \adr = \rasgnof{\regb} \in \vsizeof{\vara}
                \\
                \val = \rasgnof{\regc}
            }
            {(\pc, \rasgn, \masgn) \trans{\Estore{\adr}}{\Dstep} (\pcsuc, \rasgn, \subst{\masgn}{(\vara,\adr)}{\val})}
            {\label{rule:store}}

            \definerule{load-unsafe}
            {
                \progof{\pc} = \Iload{\rega}{\vara}{\regb}{\pcsuc}
                \\
                \adr = \rasgnof{\regb} \notin \vsizeof{\vara}
                \\
                \val = \masgnof{\varb}{\adrp}
            }
            {
                (\pc, \rasgn, \masgn) \trans{\Eload{\adr}}{\Dload{\varb}{\adrp}}
                (\pcsuc, \subst{\rasgn}{\rega}{\val}, \masgn)
            }
            {\label{rule:load-unsafe}}

            \definerule{store-unsafe}
            {
                \progof{\pc} = \Istore{\vara}{\regb}{\regc}{\pcsuc}
                \\
                \adr = \rasgnof{\regb} \notin \vsizeof{\vara}
                \\
                \val = \rasgnof{\regc}
            }
            {
                (\pc, \rasgn, \masgn) \trans{\Estore{\adr}}{\Dstore{\varb}{\adrp}}
                (\pcsuc, \rasgn, \subst{\masgn}{(\varb, \adrp)}{\val})
            }
            {\label{rule:store-unsafe}}
        \end{ruleframe}
        \begin{ruleframe}[label=rules:spec]{Speculating Semantics}
            \definerule{step}
            {
                \text{$\instof{\statec}$ speculation insensitive}
                \\
                \statec \trans{\leak}{\direct} \stated
            }
            {\nstatec\lseq\statec \ntrans{\leak}{\direct} \nstatec\lseq\stated}
            {\label{rule:nstep}}

            \definerule{spec}
            {
                \progof{\pc} = \Iif{\rega}{\pcsuc_{\btrue}}{\pcsuc_{\bfalse}}
                \\
                \bvalue = (0 \sameas \rasgnof{\rega})
            }
            {
                \nstatec \lseq (\pc, \rasgn, \masgn) \ntrans{\Eif{\lnot\bvalue}}{\Dspec}
                \nstatec\lseq(\pc, \rasgn, \masgn)\lseq(\pcsuc_{\lnot\bvalue}, \rasgn, \masgn)
            }
            {\label{rule:nspec}}
            \\
            \definerule{rollback}
            {
                \sizeof{\nstatec} \geq 1
            }
            {\nstatec \lseq \statec \ntrans{\Erbless}{\Drb} \nstatec}
            {\label{rule:nrb}}

            \definerule{sfence}
            {
                \progof{\pc} = \Isfence{\pcsuc}
            }
            {
                (\pc, \rasgn, \masgn) \ntrans{\Enone}{\Dstep} (\pcsuc, \rasgn, \masgn)
            }
            {\label{rule:nsfence}}

            \definerule{slh}
            {
                \progof{\pc} = \Islh{\rega}{\pcsuc}
                \\
                \val = \inlineife{\sizeof{\nstatec} > 0}{0}{\rasgnof{\rega}}
            }
            {
                \nstatec\lseq(\pc, \rasgn, \masgn) \ntrans{\Enone}{\Dstep} \nstatec\lseq(\pcsuc, \subst{\rasgn}{\rega}{\val}, \masgn)
            }
            {\label{rule:nslh}}
        \end{ruleframe}
    \end{ruleframes}
\end{figure}

The source of speculation are $\IWif{\regb}$ instructions triggering branch-prediction~(\spht{}).
Our semantics models speculation only for misspeculated branches.
The reason for this is that a correctly predicted branch will later commit and the resulting architectural state and the observable leakages will coincide with an execution that did not speculate in the first place.
With no difference in correct speculation and speculation-free execution there is no need to model correctly speculated branches separately, and we will use the terms speculation and misspeculation interchangeably.
Speculation thus starts with a branch misprediction,
and later ends with a rollback to the state before speculation.%
\footnote{This means our semantics allows for another speculation immediately after rollback.
This could be avoided with an additional flag to store whether a state has already been mispredicted.}
A speculating state tracks all active mispredictions in a stack of states $\nstatec, \nstated, \ldots \in \nstates = \states^*$.
The semantics $\nstatec \ntrans{\leak}{\direct} \nstated$ is provided in \Cref{rules:spec}.
With $\instof{\statec}$, we access $\statec$'s instruction $\progof{\pc}$, when $\statec = (\pc, \rasgn, \masgn)$.
The mitigation instructions $\IWsfence$ and $\IWslh{\rega}$ are \emph{speculation sensitive},
as their semantics depend on whether the current state is speculating.
Their semantics is according to our explanation in \Cref{section:intro}:
A speculation fence $\IWsfence$ disallows speculation, so \labelcref{rule:nsfence} only executes in states currently not speculating.
\labelcref{rule:nslh} performs speculative load hardening $\IWslh{\rega}$, which wipes a register only if the state is currently speculating.
The remaining instructions are \emph{speculation insensitive}.
\labelcref{rule:nstep} executes them on the currently speculating state, i.e.\ the top-most state in the stack of states.
The directives that determine whether a misprediction happens or not are $\Dspec$ and $\Dif$.
$\Dspec$ demands misprediction performed by \labelcref{rule:nspec}.
A copy of the current state is pushed on top of the current state and the program counter is set to the incorrect branch.
Otherwise, \labelcref{rule:branch} executes on directive $\Dif$ for a correct branching.
\labelcref{rule:nrb} rolls back execution to before the last misprediction.
It can be triggered with a $\Drb$ directive in any state that is currently speculating.
There is no bound on the length of a speculation.

We write $ \nstatec \ntranss{\tracel}{\dtraced} \nstated$ for finite executions
and $\nstatec \ntransi{\itrace}{\idtrace}$ for diverging executions.
We use $\tracel$ and $\dtraced$ for both finite and infinite sequences,
i.e.\ $\tracel \in \leaks^* \cup \leaks^{\infty}$ and $\dtraced \in \directs^* \cup \directs^{\infty}$.
We call any $(\entry, \rasgn, \masgn)$ initial and $(\pc, \rasgn, \masgn)$ with $\progof{\pc} = \Iexit$ final.
%
%
The behavior of a program consists of the directives and events along any execution from an initial state.
The speculation-free semantics is deterministic, so its behavior is a single execution;
the speculative behavior is non-deterministic and its behavior forms a set of executions.
\begin{align*}
    \behdefof{\prog}{\statec}
    \; & \defeq \;
    \begin{cases}
        \tbeh{\dtpair{\tracel}{\dtraced}} & \statec \transs{\tracel}{\dtraced} \stated, \textnormal{$\stated$ final} \\
        \dbeh{\dtpair{\itrace}{\idtrace}} & \statec \transi{\itrace}{\idtrace}
    \end{cases} &
    \nbehdefof{\prog}{\nstatec}
    \; & \defeq{} \;
    \begin{aligned}
        &\setcond{\tbeh{\dtpair{\tracel}{\dtraced}}}
        {\nstatec \ntranss{\tracel}{\dtraced} \nstated, \textnormal{$\nstated$ final}} \\
        {} \cup {} &
        \setcond{\dbeh{\dtpair{\itrace}{\idtrace}}}{\nstatec \ntransi{\itrace}{\idtrace}}
    \end{aligned}
\end{align*}

We call $\prog$ safe if no memory access is unsafe,
i.e.\ for every initial state $\statec$
no directives $\Dload{\vara}{\adr}, \Dstore{\vara}{\adr}$
occur in the speculation-free behavior $\behof{\prog}{\statec}$.
For the remaining paper we assume safe programs.
Note that this does not mean that speculating memory accesses need to be safe.
As seen in \Cref{code:specv1}, Spectre Attacks utilize the fact that safe programs are not safe under speculative semantics.

\begin{example}
\begin{wrapstuff}[type=code,width=0.3\linewidth]
    \begin{codebox}{}
\begin{lstlisting}
 a = (b < buf_size)
 br (a)|$\NextInst{\mkern1mu\makepc{3},\;\makepc{4}}$|
     buf[b] = secret
 bytes = stk[0]
 br (bytes)|$\NextInst{\mkern1mu\makepc{6},\;\makepc{6}}$|
 ret
\end{lstlisting}
\end{codebox}%
\caption{\label{code:simplerv1}Simplified \Cref{code:specv1}}
\end{wrapstuff}
    \Cref{code:simplerv1} contains a simplified version of \Cref{code:specv1}.
    The secret is already in a register $\rasgnof{\makereg{secret}} = \val$ and to be stored to \makevar{buf} at offset \makereg{b}.
    Assume the offset is out of bounds, $\rasgnof{\makereg{b}} = 8 \notin \sizeof{\makevar{buf}}$.
    From a state $\statec = (\makepc{2}, \rasgn, \masgn)$, where the first instruction was already executed, i.e.\ $\rasgnof{\makereg{a}} = \bfalse \neq 0$,
    the following transitions are available:
    First, a speculation is started with \labelcref{rule:nspec} and $\statec$ is copied with program counter set to the incorrect branch \makepc{3}.
    Next, \labelcref{rule:store-unsafe} executes on directive $\Dstore{\makevar{stk}}{0}$, $\masgnp = {\subst{\masgn}{(\makevar{stk}, 0)}{\val}}$.
    The \labelcref{rule:load} then brings the secret to a register, $\rasgnp = \subst{\rasgn}{\makereg{bytes}}{\val}$.
    Finally, \labelcref{rule:branch} leaks whether the secret is 0,

    \noindent
    \begin{minipage}{\linewidth}
    \begin{align*}
        \statec &\ntrans{\Eif{\bfalse}}{\Dspec} \statec\lseq(\makepc{3},\rasgn, \masgn)
        \ntrans{\Estore{8}}{\Dstore{\makevar{stk}}{0}} \statec\lseq(\makepc{4}, \rasgn, \masgnp)
                \ntrans{\Eload{0}}{\Dstep} \statec\lseq(\makepc{5}, \rasgnp, \masgnp)
                \ntrans{\Eif{\val = 0}}{\Dif} \statec\lseq(\makepc{6}, \rasgnp,\masgnp)\,.
    \end{align*}
    \end{minipage}
\end{example}

Speculative semantics exhibit two important properties:
First, due to the constant-time leakage model,
speculative semantics reveal the program counter to the attacker:
The program counter can be deduced from the leakage of conditionals in \cref{rule:branch,rule:nspec}.
Second, directives resolve all non-determinism introduced by speculation.
To express the first property,
we write $\nstatec \defsamepoint \nstated$ to mean that $\nstatec$ and $\nstated$ are at the same program point.
For speculation-free states, $\statec \atpc \pc$ means that
$\statec$ is at program counter $\pc$, $\statec = (\pc, \rasgn, \masgn)$.
Then, $\statec \samepoint \stated$ means that $\statec$ and $\stated$ share the program counter, $\statec \atpc \pc \atpc \stated$.
For speculating states, we write $\nstatec \samepoint \nstated$ if each pair of configurations in their speculation stack is at the same program point.
Formally, $\varepsilon \samepoint \varepsilon$,
and $\nstatec\lseq\statec \samepoint \nstated\lseq\stated$ if $\nstatec \samepoint \nstated$ and $\statec \samepoint \stated$.
The following two lemmas express the properties.

\begin{lemma}[Program-Counter-Leakage]\label{lem:programpoint-by-leakage}
    If two same-point states $\nstatec_1 \samepoint \nstatec_2$ execute with the same directives and leakages, $\nstatec_1 \ntranss{\tracel}{\dtraced} \nstated_1$, $\nstatec_2 \ntranss{\tracel}{\dtraced} \nstated_2$,
    then the resulting states are also same-point, $\nstated_1 \samepoint \nstated_2$.
\end{lemma}

\begin{lemma}[Directive-Determinism]\label{lemma:directive-det}
    For all $\nstatec$ and $\direct$
    there exist at most one $\nstated$ and $\leak$
    with $\nstatec \ntrans{\leak}{\direct} \nstated$.
\end{lemma}%
\subsubsection*{Notation}
Similar to how we access the current instruction with $\instof{\statec}$,
we write $\ofnota{f}{\statec}$ instead of $\ofnota{f}{\pc}$ when $\statec = (\pc, \rasgn, \masgn)$ for any $f : \pcs \to A$.
Further, we extend not only the $\samepoint$-relation to speculating states,
but any relation $R \subseteq \states \times \states$ is extended to a relation on $\nstates$ in the obvious way:
$\varepsilon \mathrel{R} \varepsilon$ and $\nstatec\lseq\statec \mathrel{R} \nstated\lseq\stated$, if $\nstatec \mathrel{R} \nstated$ and $\statec \mathrel{R} \stated$.

\section{Non-intereference Properties}
\label{section:properties}

We define non-interference and non-interference preservation for our speculative semantics.
We require the initial state's memory to be partitioned into \emph{public} and \emph{sensitive} data
through a security level assignment $\lows : \vars \to \secset$
to a lattice $\secset = (\set{\low, \high}, \low \leq \high)$.
Sensitive data~($\high$) is considered unknown to the attacker,
and we say that initial states are indistinguishable to the attacker,
$(\entry,\rasgn,\masgn)\seceq (\entry,\rasgn,\masgnp)$,
when the memory coincides on all variables considered public. 
That is, for all $\vara \in \vars$ with $\lowsof{\vara} = \low$, $\masgnofflat{\vara} = \masgnpofflat{\vara}$.

Our formulation of speculative non-interference ($\propsni$, \Cref{def:sni}) requires indistinguishable initial states $\statec_1\seceq\statec_2$ to produce equal behavior.
That means that for both initial states
\begin{enumerate*}[label=\textbf{(\roman*)}]
    \item the sequences of executable directives are the same, and
    \item for each sequence of executable directives, the observable leakage is the same.
\end{enumerate*}
Let us explain the necessity for the first condition.
As long as the second condition is satisfied, any two executions $\statec_1 \ntrans{\tracel}{\dtraced} \nstated_1$ and $\statec_2 \ntrans{\tracel}{\dtraced} \nstated_2$
will stay in the same program point $\nstated_1 \samepoint \nstated_2$ (\Cref{lem:programpoint-by-leakage}) and synchronously execute the same instructions.
If at $\nstated_1$ the set of executable directives is different to those in $\nstated_2$,
then the instruction has to be a memory access.
All other instructions have the same set of executable directives, independent of register and memory contents.
This means that one state executes an unsafe memory accesses (\labelcref{rule:load-unsafe} or \labelcref{rule:store-unsafe})
while the other executes a safe memory accesses (\labelcref{rule:load} or \labelcref{rule:store}).
However, unsafe and safe memory accesses both leak the address used.
A difference in executable directive thus amounts to different leakage.
\begin{definition}[$\propsni$]\label{def:sni}
    A program is speculatively non-interferent, $\prog \semmodels \propsni$,
    if all indistinguishable initial states $\nstatec_1 \seceq \nstatec_2$
    have the same behavior $\nbehof{\prog}{\nstatec_1} = \nbehof{\prog}{\nstatec_2}$.
\end{definition}

Our goal is to prove preservation of non-interference for compiler transformations.
We model compiler transformations $\translp$ that map a source program $\prog$ to the transformed target program~$\tgtp$.
Transformations may modify the structure of initial states from source program $\prog$ to target program~$\tgtp$.
For example, a pass that realizes the architecture's calling convention relocates function parameters to specific registers.
We require each pass to come with a relation $\simrel$ on initial states that identifies the initial states of $\tgtp$ with the initial states of $\prog$.
In order to define preservation of non-interference, the relation has to respect $\lows$ in the following sense:%
\footnote{One could also have a second security assignment on the target program, but for simplicity we assume they are the same.}
\begin{definition}\label{def:loweqrespecting}
    A relation ${\simrel} \subseteq \nstatesof{\tgtp} \times \nstatesof{\prog}$ respects $\lows$
    if every initial $\stated \in \nstatesof{\tgtp}$
    is mapped to an initial $\statec \in \nstatesof{\prog}$ with $\stated \simrel \statec$,
    and for all pairs of initial states $\stated_1 \simrel \statec_1$ and $\stated_2 \simrel \statec_2$:
    $\stated_1 \seceq \stated_2$ if and only if $\statec_1 \seceq \statec_2$.
\end{definition}
Speculative non-interference preservation for a transformation $\translp$ asks whether for all source programs $\prog$,
$\prog \semmodels \propsni$ entails $\tgtp \semmodels \propsni$.
However, defining preservation in this way leads to potentially surprising outcomes.
Even if the source program $\prog$ fails to be $\propsni$,
it can have some indistinguishable initial states which produce equal leakage.
One would expect that a speculative non-interference preserving compiler transformation preserves this equal leakage to the target program $\tgtp$.
But the above definition gives no such guarantee: If $\prog$ fails to be $\propsni$, there are no guarantees for $\tgtp$ at all.
To counteract that, our definition of speculative non-interference preservation is more precise \parencite{patrignaniExorcisingSpectresSecure2021}.
It requires preservation of equal leakage for every pair of source program's and target program's initial states individually.
In particular, this definition entails that if $\prog \semmodels \propsni$ then also $\tgtp \semmodels \propsni$.
\begin{definition}[$\propsnip$]\label{def:propsnip}
    A program translation $\translp$ with $\lows$-respecting mapping $\simrel$
    is $\propsni$-preserving, $\translp \semmodels \propsnip$,
    if all initial states $\stated_1 \seceq \stated_2$ of $\tgtp$
    with initial source states $\stated_1 \simrel \statec_1$ and $\stated_2 \simrel \statec_2$
    of equal behavior $\nbehof{\prog}{\statec_1} = \nbehof{\prog}{\statec_2}$
    also have equal target behavior,
    $\nbehof{\translpof{\prog}}{\stated_{1}} = \nbehof{\translpof{\prog}}{\stated_2}$.
\end{definition}

\section{Proving Speculative Non-Interference Preservation}
\label{section:simulation}

We present our proof method for speculative non-interference preservation.
We introduce \emph{snippy simulations} which ensure that a code transformation preserves speculative non-interference:
\begin{theorem}\label{thm:snippysound}
    If for all $\prog$ there is a snippy simulation $(\simrel, \dtffamily)$
    between $\tgtp$ and $\prog$, then $\translp \vDash \propsnip$.
\end{theorem}
In order to reach that goal,
we first define simulations that transform directives to cope with the fact that compilers do not preserve executable sequences of directives.
We then introduce the constraints a snippy simulation needs to additionally satisfy and finally prove \Cref{thm:snippysound}.
This reduces proving a transformation $\propsnip$ to proving that it has a snippy simulation for each program~$\prog$.
In \Cref{section:deadcode}, we show how to craft a snippy simulation
that is parametric in $\prog$, reducing proof effort to a once-and-for-all proof.

\subsection{Simulation with Directive Transformation}

%
The new feature in our work is \emph{directive transformation}.
Conceptually, a simulation between the target program $\tgtp$ and the source program $\prog$ shall replay any execution of $\tgtp$ in $\prog$.
A directive sequence $\dtraced$ selects a single execution in $\tgtp$ (\Cref{lemma:directive-det}).
Our simulation wants to select a corresponding execution in $\prog$.
However, the directives $\dtraced$ may not be executable in $\prog$,
or it might select an inappropriate execution.
Instead, a different sequence of directives may be necessary on the source program,
since transformations $\translp$ are not designed to preserve them.
\begin{code}[t]
    \begin{doublecodebox}{}
        \ttfamily
        \tcbsubtitle[height=0pt,after skip=0pt]{\raggedleft Source $\prog$\phantom{aa}}
        \begin{lstlisting}
if (i < buf_size)
    a = buf[i];
a = 0;
ret;
    \end{lstlisting}
\tcblower%
    \ttfamily
    \tcbsubtitle[height=0pt,after skip=0pt]{\raggedleft Target $\tgtp$}
    \begin{coderight}
if (i < buf_size)
    nop;
a = 0;
ret;
    \end{coderight}
\end{doublecodebox}
\caption{\label{code:dce}Example code transformation from dead code elimination.}
\end{code}
\begin{example}\label{example:dtf}
    We accompany our formal development with the example transformation in \Cref{code:dce},
    where an unnecessary $\IWload{\makereg{a}}{\makevar{buf}}{\makereg{i}}$ instruction is replaced by a $\IWnop$.
    Consider an initial target state $\stated = (\makepc{a},\rasgn,\masgn)$
    where $\rasgnof{\makereg{i}} \notin \sizeof{\makevar{buf}}$,
    and the source state $\statec = (\makepc{1}, \rasgn, \masgn)$.
    The directives $\Dspec\lseq\Dstep\lseq\Dstep$ are executable from $\stated$.
    But the same sequence cannot be executed from $\statec$:
    An unsafe load necessitates a directive $\Dload{\colvar{\vara}}{\adrp}$ for any $\colvar{\vara}$ and $\adrp$.
    Thus, a transformed sequence of directives $\Dspec\lseq\Dload{\colvar{\vara}}{\adrp}\lseq\Dstep$ is executed.
\end{example}

A simulation with directive transformation ($\dtsim$) is a relation on states $\nstated \simrel \nstatec$
where a target state $\nstated \in \nstatesof{\tgtp}$ is related to a source state $\nstatec \in \nstatesof{\prog}$.
The directive transformation is a family of relations
${\dtfat{\nstatec}{\nstated}} \subseteq \directs^* \times \directs^*$.
We characterize $\dtsim$ in \Cref{rules:simulation}.
Consider any states $\nstated \in \nstatesof{\tgtp}$ and $\nstatec \in \nstatesof{\prog}$ with $\nstated \simrel \nstatec$.
To express that $\nstated$ is simulated by $\nstatec$ means the following:
Either both states are final (\labelcref{rule:final}),
or we have to explore all sequences of executable directives $\nstated \ntranss{\tracel}{\dtraced} \nstatef$ in $\tgtp$ up to some bound (\labelcref{rule:tgt}).
For each explored sequence of directives we apply the directive transformation $\dtfat{\nstatec}{\nstated}$.
Then, we need to replay the execution with a sequence of executable directives $\nstatec \ntranss{\tracek}{\dtracee} \nstatee$ in $\prog$ (\labelcref{rule:src}),
so that $\nstatef \simrel \nstatee$.
Formally, we write $\proofnode{\sderivetgtsim{\simrel}{\dtf}{\nstated}{\nstatec}{\dtraced}}$
to express that we are exploring executions in~$\tgtp$,
have already seen a sequence of directives $\dtraced$ and arrived at target state $\nstated$.
We can now either bound the exploration with \labelcref{rule:direct-tf},
or continue exploration via \labelcref{rule:tgt}.
With \labelcref{rule:direct-tf} we look up a directive transformation for the explored $\dtraced$
and swap to $\proofnode{\sderivesrcsim{\simrel}{\dtf}{\nstatef}{\nstatec}{\dtracee}}$.
This states that we are seeking to replay the explored sequence with its transformation $\dtracee$ from $\nstatec$.
If $\dtracee$ is executable~$\nstatec \ntranss{\tracek}{\dtracee} \nstatee$ in~$\prog$,
\labelcref{rule:coind} checks that the states reached from exploration in $\tgtp$ and replay in $\prog$ are again related,~$\nstatef \simrel \nstatee$.
The notions of $\guarded{\simrel}$ and $\unguarded{\simrel}$ make sure that both
exploration in $\tgtp$ and replaying in $\prog$ take at least one execution step.
The guarded version $\guarded{\simrel}$ requires at least one application of \labelcref{rule:tgt} or \labelcref{rule:src} to become the unguarded version $\unguarded{\simrel}$.
Only then, \labelcref{rule:direct-tf,rule:coind} become applicable.
We write $\eitherguarded{\simrel}$ for any of $\guarded{\simrel}$ or $\unguarded{\simrel}$.
\begin{definition}[$\dtsim$]\label{def:simulation}
    A simulation with directive transformation $(\simrel, \dtffamily)$ consists
    of a relation ${\simrel} \subseteq \nstatesof{\tgtp} \times \nstatesof{\prog}$
    and a family $\dtffamily = (\dtfat{\nstatec}{\nstated})_{(\nstated, \nstatec) \in \simrel}$
    so that for all initial $\stated \in \nstatesof{\tgtp}$,
    there is an initial $\statec \in \nstatesof{\prog}$ with $\stated \simrel \statec$,
    and for all $\nstated \simrel \nstatec$,
    $\proofnode{\sderivetgtsim{\guarded{\simrel}}{\dtfat{\nstatec}{\nstated}}{\nstated}{\nstatec}{\varepsilon}}$ can be proven in \Cref{rules:simulation}.%
    \footnote{Our way to define simulations is inspired by recent work to unify stuttering \cite{choStutteringFree2023}.}
\end{definition}
\begin{figure}
    \begin{ruleframes}
        \begin{ruleframe}[label=rules:simulation]{Characterizing Simulations}
            \definerule{tgt}
            {\forall\nstated \ntrans{\leak}{\direct} \nstatef.\,\sderivetgtsim{\unguarded{\simrel}}{\dtf}{\nstatef}{\nstatec}{\dtraced \lseq \direct}{}
                \\
            \text{$\nstated$ not final}}
            {\sderivetgtsim{\eitherguarded{\simrel}}{\dtf}{\nstated}{\nstatec}{\dtraced}}
            {\label{rule:tgt}}

            \definerule{direct-tf}
            {\sderivesrcsim{\guarded{\simrel}}{\dtf}{\nstated}{\nstatec}{\dtracee} \\ \dtracee\dtf\dtraced}
            {\sderivetgtsim{\unguarded{\simrel}}{\dtf}{\nstated}{\nstatec}{\dtraced}}
            {\label{rule:direct-tf}}

            \definerule{src}
            {\exists\nstatec \ntrans{\leak}{\direct} \nstatee.\,\sderivesrcsim{\unguarded{\simrel}}{\dtf}{\nstated}{\nstatee}{\dtracee}{}}
            {\sderivesrcsim{\eitherguarded{\simrel}}{\dtf}{\nstated}{\nstatec}{\direct \lseq \dtracee}}
            {\label{rule:src}}

            \definerule{coind}
            {\nstated \simrel \nstatec}
            {\sderivesrcsim{\unguarded{\simrel}}{\dtf}{\nstated}{\nstatec}{\varepsilon}}
            {\label{rule:coind}}

            \definerule{final}
            {\text{$\nstated, \nstatec$ final}}
            {\sderivetgtsim{\guarded{\simrel}}{\dtf}{\nstated}{\nstatec}{\varepsilon}}
            {\label{rule:final}}
        \end{ruleframe}
        \begin{ruleframe}[label=rules:synctrans]{Simulation Intervals}
            \definerule{sync}
            {
                \nstated \ntranss{\tracel}{\dtraced} \nstatef \\
                \nstated \simrel \nstatec \\
                \nstatef \simrel \nstatee \\
                \nstatec \ntranss{\tracek}{\dtracee} \nstatee
                \\\\
                \proofnode{{\nstated}\tgtsim{\nstatec}}
                \rproofedges{\dtraced}
                \proofnode{{\nstatef}\tgtsim{\nstatec}}
                \rproofedges{\dtracee \dtfspaced \dtraced}
                \proofnode{{\nstatef}\srcsim{\nstatec}}
                \rproofedges{\dtracee}
                \proofnode{{\nstatef}\srcsim{\nstatee}}
            }
            {
                \nstatec \simtranss{\tracek}{\dtracee} \nstatee\\
                (\nstatec, \nstated) \synctrans{\tracek}{\dtracee}{\tracel}{\dtraced} (\nstatee, \nstatef) \\
                \nstated \simtranst{\tracel}{\dtraced} \nstatef
            }
            {\label{rule:sync}}
        \end{ruleframe}
    \end{ruleframes}
\end{figure}

\begin{example}\label{example:simulation}
    Consider once again the transformation in \Cref{code:dce} and the initial states $\stated$ and $\statec$ from \Cref{example:dtf}.
    Further, let $\statee = (\makepc{4},\rasgnp,\masgn)$ and $\statef = (\makepc{d},\rasgnp,\masgn)$ with $\rasgnp = \subst{\rasgn}{\makereg{a}}{0}$.
    We want to prove that $\stated \simrel \statec$ is justified,
    i.e.\ we need to construct ${\dtfat{\statec}{\stated}}$ so that $\proofnode{\sderivetgtsim{\guarded{\simrel}}{\dtfat{\statec}{\stated}}{\stated}{\statec}{\varepsilon}}$ is derivable.
    We drop the subscript and just write $\dtf$.
    Exploration via \labelcref{rule:tgt} yields (among others) two sequences of directives executable from $\stated$ in $\tgtp$:
    $\Dspec\lseq\Dstep\lseq\Dstep$ as in \Cref{example:dtf} and~$\Dif$.
    The corresponding execution takes us to $\stated \ntranss{\Eif{\bfalse}\lseq\Enone\lseq\Enone}{\Dspec\lseq\Dstep\lseq\Dstep} \stated\lseq\statef$
    and $\stated \ntrans{\Eif{\bfalse}}{\Dif} \statef$, respectively.
    After exploration with \labelcref{rule:tgt}, we are thus left to prove $\proofnode{\sderivetgtsim{\unguarded{\simrel}}{\dtf}{\stated\lseq\statef}{\statec}{\Dspec\lseq\Dstep\lseq\Dstep}}$
    and $\proofnode{\sderivetgtsim{\unguarded{\simrel}}{\dtf}{\statef}{\statec}{\Dif}}$.
    For the first sequence, we transform the directives as in \Cref{example:dtf}:
    $\Dspec\lseq\Dload{\makevar{sec}}{0}\lseq\Dstep \dtf \Dspec\lseq\Dstep\lseq\Dstep$.
    For the other case, we do not need a transformation, so $\Dif \dtf \Dif$.
    With \labelcref{rule:direct-tf}, we are left with deriving $\proofnode{\sderivesrcsim{\guarded{\simrel}}{\dtf}{\stated\lseq\statef}{\statec}{\Dspec\lseq\Dload{\makevar{sec}}{0}\lseq\Dstep}}$
    and $\proofnode{\sderivesrcsim{\guarded{\simrel}}{\dtf}{\statef}{\statec}{\Dif}}$.
    Indeed, $\prog$ can replay the directives with $\statec \ntranss{\Eif{\bfalse}\lseq\Eload{\adr}\lseq\Enone}{\Dspec\lseq\Dload{\makevar{sec}}{0}\lseq\Dstep} \statec\lseq\statee$ using \labelcref{rule:src},
    where $\adr = \rasgnof{\makereg{i}}$
    and $\statec \ntrans{\Eif{\bfalse}}{\Dif} \statee$, respectively.
    To now utilize \labelcref{rule:coind} we need $\statef \simrel \statee$ and $\stated\lseq\statef \simrel \statec\lseq\statee$.
    Each would again have to be justified independently.
    Justifying the first is easy with \labelcref{rule:final},
    while the other needs another application of \labelcref{rule:tgt}, \labelcref{rule:direct-tf}, and \labelcref{rule:src}
    and can then utilize $\stated \simrel \statec$.
\end{example}

\subsubsection*{Simulation Intervals}
Our goal is to formulate snippy simulations as a constraint on simulations with directive transformation.
We define it with simulation intervals.
A simulation interval for states $\nstated \simrel \nstatec$ is a pair of an explored sequence of directives from $\nstated$ in $\tgtp$
and the corresponding replay from $\nstatec$ in $\prog$.
Formally, we define simulation intervals through a synchronized product,
whose transitions are the simulation intervals.
The states of the synchronized product are of shape $(\nstatec,\nstated)$ so that $\nstated \simrel \nstatec$.
Its transitions are of the form $(\nstatec, \nstated) \synctrans{\tracek}{\dtracee}{\tracel}{\dtraced} (\nstatee, \nstatef)$,
where $\dtracee$ is the directive transformation of an explored $\dtraced$ in $\tgtp$.
In order to formally define the transition relation, consider a proof tree that justifies $\nstated \simrel \nstatec$,
i.e.\ that derives $\proofnode{\sderivetgtsim{\guarded{\simrel}}{\dtfat{\nstatec}{\nstated}}{\nstated}{\nstatec}{\varepsilon}}$.
We use the notation 
\[
\proofnode{{\nstated}\tgtsim{\nstatec}}
\rproofedges{\dtraced}
\proofnode{{\nstatef}\tgtsim{\nstatec}}
\rproofedges{\dtracee \dtfspaced \dtraced}
\proofnode{{\nstatef}\srcsim{\nstatec}}
\rproofedges{\dtracee}
\proofnode{{\nstatef}\srcsim{\nstatee}}
\]
to state that the proof tree contains the nodes
$\proofnode{{\sderivetgtsim{\unguarded{\simrel}}{\dtfat{\nstated}{\nstatec}}{\nstatef}{\nstatec}{\dtraced}}}$,
$\proofnode{{\sderivesrcsim{\guarded{\simrel}}{\dtfat{\nstated}{\nstatec}}{\nstatef}{\nstatec}{\dtracee}}}$,
and $\proofnode{{\sderivesrcsim{\unguarded{\simrel}}{\dtfat{\nstated}{\nstatec}}{\nstatef}{\nstatee}{\varepsilon}}}$ on one path.
In particular, this means that $\nstated \ntranss{\tracel}{\dtraced} \nstatef$ and $\nstatec \ntranss{\tracek}{\dtracee} \nstatee$ for appropriate $\tracel$, $\tracek$,
and $\dtracee \dtfat{\nstatec}{\nstated} \dtraced$.
We define a synchronized transition relation that executes both in a single step.

\begin{definition}\label{def:siminterval}
    Given a simulation $(\simrel, \dtffamily)$,
    its simulation interval transition is defined by \Cref{rule:sync}.
    \begin{align*}
        {\synctrans{\,}{}{}{\,}} &\quad\subseteq\quad {\simrel} \times \directs^* \times \leaks^* \times \directs^* \times \leaks^* \times {\simrel}
    \end{align*}
\end{definition}

\labelcref{rule:sync} further defines the transition relations $\simtranss{}{}$ and $\simtranst{}{}$ as the projection of simulation intervals to source and target program.
Transitive closures of the transition relations are defined as usual.
We say that a simulation is lock-step if simulation intervals are single step: ${\simtranst{}{}}, {\simtranss{}{}} \subseteq {\ntrans{}{}}$.
\begin{example}\label{example:siminterval}
    The simulation intervals resulting from \Cref{example:simulation} for $\stated \simrel \statec$ are:
    \begin{align*}
        (\statec, \stated) & \synctransdisp{\Eif{\bfalse}\lseq\Eload{\adr}\lseq\Enone}{\Dspec\lseq\Dload{\makevar{secret}}{0}\lseq\Dstep}{\Eif{\bfalse}\lseq\Enone\lseq\Enone}{\Dspec\lseq\Dstep\lseq\Dstep} (\statec\lseq\statee, \stated\lseq\statef)
                           &
        (\statec, \stated) & \synctransdisp{\Eif{\bfalse}}{\Dif}{\Eif{\bfalse}}{\Dif} (\statee, \statef)
    \end{align*}
\end{example}

The following lemma states that our formulation of simulations is sound.
That is, we find all of $\tgtp$'s behavior in the projection of the simulation interval transition relation $\simtranst{}{}$.
This lets us perform (co-)induction on $\nbehof{\tgtp} \nstatec$ with $\simtranst{}{}$ rather than $\ntrans{}{}$, which we will utilize in our proof of \Cref{thm:snippysound}.
The same is not true for the source program's behavior and $\simtranss{}{}$.
\begin{lemma}\label{lem:behaviorininterval}
    If $\nstated$ occurs in $\simrel$,
    $\nbehof{\translpof{\prog}}{\nstated} = \setcond{\tbeh{\dtpair{\tracel}{\dtraced}}}{\nstated \simtransts{\tracel}{\dtraced} \nstateh, \textrm{$\nstateh$ final}} \cup
    \setcond{\dbeh{\dtpair{\itrace}{\idtrace}}}{\nstated \simtransti{\itrace}{\idtrace}}$.
\end{lemma}

\subsection{Snippy simulations}

\begin{wrapstuff}[type=figure,width=0.4\linewidth]
        \tikzset{
                prewhite/.style = {
                    preaction={draw,line width=3.5pt,white}
                },
                samepoint/.style args={#1}{
                    draw,line width=0.7pt,#1,
                    preaction={draw,line width=3pt,white,
                        preaction={draw,line width=4.4pt,#1}
                    }
                },
                equality/.style = { double, double distance=1.8pt },
                simrel/.style = { line width=1pt, postaction={decorate},
                    decoration = { markings, mark=at position 0.4 with {
                            \node[fill=white] {$\simrelr$};
                    }},
                },
                leaktf/.style = { line width=1pt, postaction={decorate},
                    decoration = { markings, mark=at position 0.35 with {
                            \node[fill=white] {$\dtf$};
                    }},
                },
                ntranscube/.style = {->,line width=1pt,dashed,dash pattern=on 3pt off 3pt},
                given/.style = {},
                new1/.style = { mauve },
                new2/.style = { NavyBlue },
                new3/.style = { YellowGreen },
                new4/.style = { BurntOrange },
                result/.style = { lightgray },
                ignoreheight/.style = { text depth=0, text height=1ex }
            }
        \scalebox{0.9}[0.9]{
            \begin{tikzpicture}[scale=1]
                \tikzmath{\disthorizontal = 3; \distvertical = 2.5; }

                \node (std1) {$\nstated_1$};
                \node[above right=of std1] (std2) {$\nstated_2$};
                \node[left=\disthorizontal of std1] (stc1) {$\nstatec_1$};
                \node[left=\disthorizontal of std2] (stc2) {$\nstatec_2$};
                \node[below=\distvertical of std1] (stf1) {$\nstatef_1$};
                \node[below=\distvertical of std2,new1] (stf2) {$\nstatef_2$};
                \node[below=\distvertical of stc1] (ste1) {$\nstatee_1$};
                \node[below=\distvertical of stc2] (ste2) {$\nstatee_2$};

                \draw[samepoint] (std1) -- (std2);
                \draw[samepoint] (stc1) -- (stc2);
                \draw[samepoint=result] (stf1) -- (stf2);
                \draw[samepoint=result] (ste1) -- (ste2);

                \draw[simrel] (stc2) -- (std2);
                \draw[simrel] (stf1) -- (ste1);
                \draw[simrel,new1] (ste2) -- (stf2);

                \draw[ntranscube] (stc1) -- node[right,pos=0.7] (v1) {$\dtpair{\tracek}{\dtracee}$} (ste1);
                \draw[ntranscube] (stc2) -- node[right,pos=0.7] (v2) {$\dtpair{\tracek}{\dtracee}$} (ste2);

                \draw[prewhite,simrel] (std1) -- (stc1);

                \draw[ntranscube,new1] (std2) -- node[left,pos=0.7,ignoreheight] (u2) {$\dtpair{\tracel}{\dtraced}$} (stf2);
                \draw[leaktf,new1] (v2) -- (u2);
                \draw[ntranscube,prewhite] (std1) -- node[left,pos=0.7,ignoreheight] (u1) {$\dtpair{\tracel}{\dtraced}$} (stf1);
                \draw[leaktf,prewhite] (u1) -- (v1);
            \end{tikzpicture}
        }
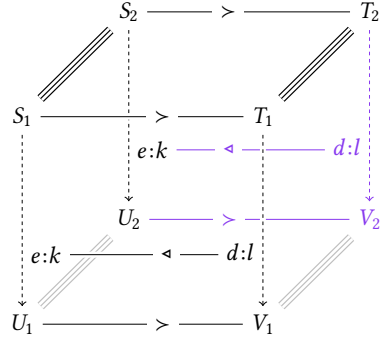
    \caption{\label{fig:snippysim} Whenever the black conditions are met,
        a snippy simulation $\simrel$ also explores the \textcolor{mauve}{purple} execution
        and simulates it by the source execution.
        \textcolor{lightgray}{Gray} conditions follow from the semantics.
    }
\end{wrapstuff}

So far, simulations are very liberal:
Simulation merely require that a sequence of directives in $\tgtp$ can be transformed via $\dtffamily$ into a sequence of directives in $\prog$.
The length and contained directives can change when applying $\dtffamily$ and there are no restrictions on how the leakage changes when applying~$\dtffamily$.
In this section, we establish \emph{snippy simulations},
a constant-time cube constraint \parencite{bartheSecureCompilationSideChannel2018} on simulations for speculative semantics
that entails $\propsnip$ when satisfied.

The intuition for snippy simulations can be explained as follows.
In order to prove that a transformation $\translp$ satisfies $\propsnip$,
we are given a source program $\prog$ and the target program $\tgtp$,
as well as four initial states:
Two target initial states $\stated_1 \samepoint \stated_2$
and two simulating source states $\stated_1 \simrel \statec_1$ and $\stated_2 \simrel \statec_2$,
so that $\nbehof{\prog}{\statec_1} = \nbehof{\prog}{\statec_2}$.
The goal is to prove that the equality of behavior carries over to the target program.
Given a simulation $(\simrel, \dtffamily)$ we so far know how to replay any sequence of directives $\dtraced$ from $\stated_1$ transformed on $\statec_1$ (\Cref{lem:behaviorininterval}).
Consider a simulation interval $(\statec_1, \stated_1) \synctrans{\tracek}{\dtracee}{\tracel}{\dtraced} (\nstatee_1, \nstatef_1)$.
Due to same behavior of $\statec_1$ and $\statec_2$,
the source directives can also be executed from $\statec_2$,
$\statec_2 \ntrans{\tracek}{\dtracee} \nstatee_2$.
Snippy simulations now state that, in such a situation,
the simulation interval for $\statec_2$ and $\stated_2$
also contains $(\statec_2, \stated_2) \synctrans{\tracek}{\dtracee}{\tracel}{\dtraced} (\nstatee_2, \nstatef_2)$.
That means, $\stated_2$ can also execute $\dtraced$ and produce the same leakage.
And further, $(\simrel, \dtffamily)$ does also explore $\dtraced$ from $\stated_2$, not a longer or shorter sequence.
\Cref{fig:snippysim} demonstrates the constraint in the general case, where states need not be initial.
The simulation interval of $\nstatec_1$ and $\nstated_1$ and the ability for another source state $\nstatec_2$ to mimic the behavior are in black.
The constraint is in purple: $(\simrel, \dtffamily)$ has to also provide the same simulation interval for any other state $\nstated_2 \simrel \nstatec_2$ at the same program point as $\nstated_1$.

With snippy simulations defined, we conclude the section with the proof of \Cref{thm:snippysound}.

\begin{definition}\label{def:snippy}
    A snippy simulation $(\simrel, \dtffamily)$ is $\lows$-respecting and satisfies the diagram in \Cref{fig:snippysim}.
    That is, for all $\nstatec_1 \samepoint \nstatec_2$ and $\nstated_1 \samepoint \nstated_2$
    with $\nstated_1 \simrel \nstatec_1$, $\nstated_2 \simrel \nstatec_2$,
    and $(\nstatec_1, \nstated_1) \synctrans{\tracek}{\dtracee}{\tracel}{\dtraced} (\nstatee_1, \nstatef_1)$,
    \begin{align*}
        \nstatec_2 &\ntranss{\tracek}{\dtracee} \nstatee_2
                   &&\text{implies the existence of $\nstatef_2$ with} 
                   &
        (\nstatec_2, \nstated_2) &\synctrans{\tracek}{\dtracee}{\tracel}{\dtraced} (\nstatee_2, \nstatef_2) \,.
    \end{align*}
\end{definition}

\begin{proof}[Proof of \Cref{thm:snippysound}]
    Consider a program $\prog$ and a snippy simulation $(\simrel, \dtffamily)$.
    We need to prove the following:
    For all initial $\stated_1 \seceq \stated_2$ with $\stated_1 \simrel \statec_1$ and $\stated_2 \simrel \statec_2$:
    When $\nbehof{\tgtp}{\nstatec_1} = \nbehof{\tgtp}{\statec_2}$,
    then also $\nbehof{\tgtp}{\stated_1} = \nbehof{\tgtp}{\stated_2}$.
    We claim a stronger statement:
    Whenever $\nstated_1 \simrel \nstatec_1$, $\nstated_2 \simrel \nstatec_2$,
    $\nstatec_1 \samepoint \nstatec_2$, and $\nstated_1 \samepoint \nstated_2$,
    and $\nbehof{\prog}{\nstatec_1} = \nbehof{\prog}{\nstatec_2}$:
    Then $\nbehof{\tgtp}{\nstated_1} \subseteq \nbehof{\tgtp}{\nstated_2}$ holds.

    This is sufficient:
    Consider initial target states $\stated_1 \seceq \stated_2$ as well as source states $\stated_1 \simrel \statec_1$ and $\stated_2 \simrel \statec_2$
    with $\nbehof{\prog}{\statec_1} = \nbehof{\prog}{\statec_2}$.
    Initial states are all at the same program point, so the requirements of the claim are satisfied
    and $\nbehof{\tgtp}{\stated_1} \subseteq \nbehof{\tgtp}{\stated_2}$ holds.
    By symmetry, $\nbehof{\tgtp}{\stated_1} = \nbehof{\tgtp}{\stated_2}$.

    We prove our claim coinductively on $\nbehof{\tgtp}{\nstated_1}$ split into simulation intervals (\Cref{lem:behaviorininterval}).
    The case of a final $\nstated_1$, i.e.\ $\nbehof{\tgtp}{\nstated_1} = \set{\tbeh{\dtpair{\varepsilon}{\varepsilon}}}$,
    $\nstated_2 \samepoint \nstated_1$ is final, too,
    and thus $\tbeh{\dtpair{\varepsilon}{\varepsilon}} \in \nbehof{\tgtp}{\nstated_2}$.
    In the (co-)inductive case, let
    $\tbeh{\dtpair{\tracek\lseq\tracel}{\dtracee\lseq\dtraced}} \in \nbehof{\tgtp}{\nstated_1}$
    with $\nstated_1 \simtranst{\tracek}{\dtracee} \nstatef_1$
    from a simulation interval $(\nstatec_1, \nstated_1) \synctrans{\tracem}{\dtracef}{\tracek}{\dtracee} (\nstatee_1, \nstatef_1)$.
    From $\nbehof{\prog}{\nstatec_1} = \nbehof{\prog}{\nstatec_2}$ follows $\nstatec_2 \ntranss{\tracem}{\dtracef} \nstatee_2$.
    \Cref{lem:programpoint-by-leakage} gives $\nstatee_1 \samepoint \nstatee_2$
    and same behavior of $\nstatec_1$ and $\nstatec_2$ entails $\nbehof{\prog}{\nstatee_1} = \nbehof{\prog}{\nstatee_2}$.
    %
    %
    Snippyness then yields the simulation interval $(\nstatec_2, \nstated_2) \synctrans{\tracem}{\dtracef}{\tracek}{\dtracee} (\nstatee_2, \nstatef_2)$,
    i.e.\ $\nstated_2 \ntranss{\tracek}{\dtracee} \nstatef_2$.
    \Cref{lem:programpoint-by-leakage} gives $\nstatef_1 \samepoint \nstatef_2$.
    We can now apply \makebox{(co-)induction} hypothesis for $\nbehof{\tgtp}{\nstatef_1} \subseteq \nbehof{\tgtp}{\nstatef_2}$,
    which implies $\tbeh{\dtpair{\tracek\lseq\tracel}{\dtracee\lseq\dtraced}} \in \nbehof{\tgtp}{\nstatef_2}$.
    Together with $\nstated_2 \ntranss{\tracek}{\dtracee} \nstatef_2$, we arrive at $\tbeh{\dtpair{\tracek\lseq\tracel}{\dtracee\lseq\dtraced}} \in \nbehof{\tgtp}{\nstated_2}$.
\end{proof}

\section{Case Study: Dead Code Elimination}
\label{section:deadcode}

In this section, we prove that dead code elimination $\tdconly$ satisfies $\propsnip$ as a demonstration of the proof method and a warm-up for the next section,
where we tackle register allocation transformations.
We first give a short rundown on the transformation of $\tdconly$ before crafting a snippy simulation that is parametric in the source program $\prog$.
Dead code elimination is responsible for removing instructions from $\prog$ whose computed values are not utilized anywhere.
It is the result of a data flow analysis that finds removable instructions.
Typically, the analysis follows constant propagation in order to identify as many aliasing memory accesses as possible.
To support that, we assume that load and store instructions $\IWloadconst{\rega}{\vara}{\adr}$ and $\IWstoreconst{\vara}{\adr}{\regc}$ can hold a constant address $\adr \in \adrs$
instead of a register,
where we require $\adr \in \sizeof{\vara}$.
For the remaining section, fix an arbitrary program $\prog$ with entry point $\entry$.

\subsubsection*{Flow analysis}
The static analysis for dead code elimination is a Liveness analysis.
Since we will later define our own flow analysis when fixing the weaknesses in register allocation, we recall flow analysis in a more general setting.
A flow analysis searches for a fixed-point solution to a system of flow inequalities in order to obtain approximate knowledge about all executions of a program.
Formally, a \emph{forward/backward flow analysis} finds a solution $\flow$ to the inequalities~\labelcref{eqn:flow-fwd}/\labelcref{eqn:flow-bwd},
where $\pc$ and $\pcsuc$ range over program counters so that $\pc$ is a predecessor to $\pcsuc$.
\begin{subequations}\label{eqn:flow}
\[
    \begin{minipage}{.5\displaywidth}
        \abovedisplayskip=0pt
        \belowdisplayskip=0pt
        \begin{equation}
            \begin{split}
                \flowof{\pcsuc} &\geq \transferof{\pc}{(\flowof{\pc})} \\
                \flowof{\entry} &\geq \flowinit
            \end{split}\label{eqn:flow-fwd}\tag{fwd}
        \end{equation}
    \end{minipage}
    \begin{minipage}{.5\displaywidth}
        \abovedisplayskip=0pt
        \belowdisplayskip=0pt
        \begin{equation}
            \begin{split}
                \flowof{\pc} &\geq \transferof{\pcsuc}{(\flowof{\pcsuc})} \\
                \flowof{\pc} &\geq \flowinit \qquad\qquad \progof{\pc} = \Iexit
            \end{split}\label{eqn:flow-bwd}\tag{bwd}
        \end{equation}
    \end{minipage}
\]
\end{subequations}
Flow values stem from a semi-lattice $(\lattice, \leq)$ and a solution $\flow : \pcs \to \lattice$ finds a flow value for each program point.
In case of a forward analysis, $\flowof{\pc}$ denotes the flow value at $\pc$ before the execution of $\progof{\pc}$.
For a backward analysis $\flowof{\pc}$ denotes the flow value at $\pc$ after the execution of $\progof{\pc}$.
The flow value $\flowinit$ is the initial flow value of entry / exit points of the program.
The functions $\transferat{\pc} : \lattice \to \lattice$ are monotonic and constitute the \emph{transfer} of the flow values along instructions.
%

A flow analysis can have \emph{additional constraints}.
Additional constraints are of shape $\flowof{\pc} \leq \latelem$.
They additionally require flow values of certain program counters $\pc$ not to exceed a bound $\latelem \in \lattice$.
If the least solution to the flow analysis does not satisfy the additional constraints, no solution does.

\subsubsection*{Dead Code Elimination}
Liveness analysis is a backward flow analysis.
The flow values are the sets of registers and memory locations that are live at any given program point in that their current value could be used later.
The flow lattice is $\lattice = \powerset{\regs \cup \mems}$,
and the initial flow value at any exit point is $\flowinit = \regs\cup\mems$,
but can be different dependent on calling conventions.
The transfer functions are folklore,
so we instead formulate the guarantee that comes with a solution.
Note that the guarantee holds for speculative semantics, too, because the analysis is branch-independent.
\begin{proposition}\label{lem:liveness}
    Whenever $\statec \ntranss{\tracel}{\dtraced} \nstatee\lseq\statee$ in $\prog$,
    if $\progof{\statee}$ uses a register $\regb$, then $\regb \in \transferof{\statee}{(\flowof{\statee})}$,
    and if $\progof{\statee}$ loads a memory location $\vara$ with offset $\adr$, then $(\vara, \adr) \in \transferof{\statee}{(\flowof{\statee})}$.
\end{proposition}
The transformation~$\tdconly$ uses a Liveness analysis solution $\flow$ of the backward flow inequalities \labelcref{eqn:flow-bwd} to remove unnecessary instructions.
A function $\tfdc : \lattice \to \instrs \to \instrs$ inspects the flow value at a given program point and removes an instruction if it writes a register or memory location that is not live.
The transformation of $\prog$ is then defined per program point with $\tdcof{\pc} = \tfdcof{(\flowof{\pc})}{(\progof{\pc})}$.
\begin{align*}
    \tfdcof{\latelem}{(\Iasgn{\rega}{\regb \op \regc}{\pcsuc})} &= 
    \begin{cases}{}
        \Inop{\pcsuc} & \rega \notin \latelem \\
        \Iasgn{\rega}{\regb \op \regc}{\pcsuc} & \rega \in \latelem
    \end{cases}
                                                                &
    \tfdcof{\latelem}{(\Iload{\rega}{\vara}{x}{\pcsuc})} &= 
    \begin{cases}{}
        \Inop{\pcsuc} & \rega \notin \latelem \\
        \Iload{\rega}{\vara}{x}{\pcsuc} & \rega \in \latelem
    \end{cases}
    \\
    \tfdcof{\latelem}{(\Istore{\vara}{x}{\regc}{\pcsuc})} &= 
    \begin{cases}{}
        \Inop{\pcsuc} &x = \adr, (\vara, \adr) \notin \latelem \\
        \Istore{\vara}{x}{\regc}{\pcsuc} & \text{otherwise}
    \end{cases}
                                                              &
    \tfdcof{\latelem}{\inst} &= \inst \quad \text{for other $\inst$}
\end{align*}

\subsubsection*{Dead Code Simulation}

We now craft a snippy simulation $(\simrel, \dtffamily)$ between $\tdc$ and $\prog$.
The simulation is parametric in $\prog$ as it depends on the Liveness analysis $\flow$.
That way, we craft a single simulation relation and provide one for every source program $\prog$ and its transformation $\tdc$.
First, we define the simulation $\stated \simrel \statec$ and directive transformation $\dtfat{\statec}{\stated}$ for speculation-free states,
and afterwards lift them to speculating states.
The simulation identifies states whenever they differ only on dead registers and memory locations:
\begin{align*}
    (\pc, \rasgn, \masgn) &\simrel (\pc, \rasgnp, \masgnp) &
                          &\iffsymb &
                              \begin{aligned}
    \forall \rega\in\transferof{\pc}{(\flowof{\pc})}.\,
                          & \rasgnof\rega = \rasgnpof{\rega}\\
    \land \quad \forall (\vara, \adr)\in\transferof{\pc}{(\flowof{\pc})}.\,
                          & \masgnof{\vara}{\adr} = \masgnpof{\vara}{\adr}
                              \end{aligned}
\end{align*}
In order to define the directive transformation $\dtfat{\statec}{\stated}$,
we need to think about the shape of simulation intervals.
Because $\tdconly$ leaves the control flow fully intact,
we can choose to create a lockstep simulation.
With that, we can choose $\dtfat{\statec}{\stated}$ to be the identity relation on $\directs$
and add transformations where $\tdconly$ replaced an instruction with $\IWnop$.
Let $\inst = \instof{\statec}$ and $\instp = \tdcof{\stated}$.
We set:
\begin{align*}
    {\dtfat{\statec}{\stated}} & = \identity{\directs} \cup 
    \begin{cases}
        \setcond{(\Dload{\varb}{\adrp}, \Dstep)}{\varb \in \vars, \adrp \in \adrs} & \inst = \Iload{\rega}{\vara}{\regb}{\pcsuc}, \;\;\instp = \Inop{\pcsuc} \\
        \setcond{(\Dstore{\varb}{\adrp}, \Dstep)}{\varb \in \vars, \adrp \in \adrs} & \inst = \Istore{\vara}{\regb}{\regc}{\pcsuc}, \;\;\instp = \Inop{\pcsuc} \\
        \varnothing & \text{otherwise}
    \end{cases}
\end{align*}

For speculating states, we simply lift $\simrel$ by setting $\nstated\lseq\stated \simrel \nstatec\lseq\statec$ if $\nstated \simrel \nstatec$ and $\stated \simrel \statec$ (and $\varepsilon \simrel \varepsilon$ for the base case).
For the directive transformation, we delegate to the executing states ${\dtfat{\nstatec\lseq\statec}{\nstated\lseq\stated}} = {\dtfat{\statec}{\stated}}$.
\begin{theorem}\label{thm:dcsnippy}
    $(\simrel, \dtffamily)$ is a snippy lockstep simulation.
\end{theorem}
\begin{proof}
    We need to prove that $(\simrel, \dtffamily)$
    \begin{enumerate*}[label=\bfseries (\roman*)]
        \item\label{item:dcissim} is a simulation (\Cref{def:simulation}),
        \item\label{item:dcissnippy} is snippy (\Cref{def:snippy}),
        and \item\label{item:dcissecrespecting} respects $\lows$ (\Cref{def:loweqrespecting}).
    \end{enumerate*}
    The first part \labelcref{item:dcissim} is considerably easier than in the general case, because the simulation is lockstep,
    i.e.\ both $\prog$ and $\tdc$ only perform a single step before finding new states in $\simrel$.
    Consider $\nstated_1\lseq\stated_1 \simrel \nstatec_1\lseq\statec_1$,
    $\nstated_2\lseq\stated_2 \simrel \nstatec_2\lseq\statec_2$,
    $\nstated_1 \lseq\stated_1 \samepoint \nstated_2 \lseq\stated_2$
    $\nstatec_1 \lseq\statec_1 \samepoint \nstatec_2 \lseq\statec_2$
    and
    $\nstated_1 \lseq\stated_1 \ntrans{\leak}{\direct} \nstatef_1$,
    where $\stated_1 \atpc \stated_2 \atpc \statec_1 \atpc \statec_2 \atpc \pc$.
    We need to show that there is
    $\directp \dtfat{\statec_1}{\stated_1} \direct$
    so that $\nstatec_1\lseq\statec_1 \ntrans{\leakp}{\directp} \nstatee_1$
    and $\nstatef_1 \simrel \nstatee_1$.
    Second we prove snippyness~\labelcref{item:dcissnippy}:
    We additionally consider the existence of $\nstatec_2\lseq\statec_2 \ntrans{\leakp}{\directp} \nstatee_2$.
    We then need to show that $\nstated_2\lseq\stated_2 \ntrans{\leak}{\direct} \nstatef_2$
    and $\directp \dtfat{\statec_2}{\stated_2} \direct$ (which implies $(\statec_2, \stated_2) \synctrans{\leak}{\direct}{\leakp}{\directp} (\statee_2, \statef_2)$).

    We first split off the case $\direct = \directp = \Drb$, where $\sizeof{\nstated_1} > 0$.
    For \labelcref{item:dcissim}, we need to prove $\nstatec_1\lseq\statec_1 \ntrans{\Erb}{\Drb} \nstatec_1$ because $\Drb \dtfat{\statec_1}{\stated_1} \Drb$.
    But $\nstated_1 \simrel \nstatec_1$ implies $\sizeof{\nstated_1} = \sizeof{\nstatec_1} > 0$ which meets the premise for the transition.
    For \labelcref{item:dcissnippy}, additionally consider $\nstatec_2 \lseq \statec_2 \ntrans{\Erb}{\Drb} \nstatec_2$.
    We need to show $\nstated_2\lseq\stated_2 \ntrans{\Erb}{\Drb} \nstated_2$.
    With the same argument as before, $\sizeof{\nstated_1} = \sizeof{\nstatec_1} = \sizeof{\nstatec_2} = \sizeof{\nstated_2} > 0$ satisfies the premise for the transition.
    Proving \labelcref{item:dcissim,item:dcissnippy} in the case of $\direct \neq \Drb \neq \directp$ is a large case distinction on $\langle\progof{\pc}, \tdcof{\pc}\rangle$.

    \begin{pcases}
        \pcasep{$\langle\Iload{\rega}{\vara}{\regb}{\pcsuc},\Inop{\pcsuc}\rangle$}
        There are two subcases: $\rasgniof{\statec_1}{\regb} = \adr$ is within $\sizeof{\vara}$ or not.
        We present the subcase $\adr \notin \sizeof{\vara}$.
        For \labelcref{item:dcissim}, consider $\stated_1 \trans{\Enone}{\Dstep} \statef_1$.
        We need to show that $\Dload{\varb}{\adrp}$ is executable in $\statec_1$,
        because $\Dload{\varb}{\adrp} \dtfat{\statec_1}{\stated_1} \Dstep$.
        Indeed, we have $\statec_1 \trans{\Eload{\adr}}{\Dload{\varb}{\adrp}} \statee_1$ because $\rasgniof{\statec_1}{\regb} = \adr \notin \sizeof{\vara}$.
        We also need to show that $\statef_1 \simrel \statee_1$.
        From the definition of $\tfdc$, we know $\rega \notin \flowof{\pc}$ because the instruction was replaced by $\IWnop$.
        The transfer $\transferof{\pc}{(\flowof{\pc})}$ is again $\flowof{\pc}$
        because a load to a dead register makes no registers or memory locations live.
        Together with \labelcref{eqn:flow-bwd}, we get $\transferof{\pc}{(\flowof{\pc})} = \flowof{\pc} \supseteq \transferof{\pcsuc}{(\flowof{\pcsuc})}$.
        Registers and memory of $\stated_1$ and $\statef_1$ are equal ($\IWnop$)
        and $\statec_1$ and $\statee_1$ only differ on $\rega$.
        This means $\stated_1 \simrel \statec_1$ implies $\statef_1 \simrel \statee_1$ because $\rega \notin \flowof{\pc} \supseteq \transferof{\pcsuc}{(\flowof{\pcsuc})}$.
        For \labelcref{item:dcissnippy}, further consider $\statec_2 \trans{\Eload{\adr}}{\Dload{\varb}{\adrp}} \statee_2$.
        We need to show that $\stated_2 \trans{\Enone}{\Dstep} \statef_2$ and $\statef_2 \simrel \statee_2$.
        The former is immediate from the semantics.
        For the latter, we apply the same arguments as in \labelcref{item:dcissim}.
        All other cases are very similar.
    \end{pcases}
    For \labelcref{item:dcissecrespecting}, we need to restrict our attention to Liveness analysis where all registers and memory locations are initially live.%
    \footnote{A weaker formalization of $\sec$-respecting relations would lift this restriction, but incur more presentational overhead.}
    Every solution $\flow$ is easily modified to satisfy this restriction.
    The proof that $\simrel$ respects $\sec$ is then straightforward:
    For initial states $\stated$ and $\statec$, $\stated \simrel \statec$ if and only if $\stated = \statec$.
\end{proof}

\section{Fixing Weaknesses in Register Allocation}
\label{section:regalloc}

Register allocation happens in the compilation phase that moves from an IR to the hardware instructions of the target architecture.
It moves from the unbounded number of virtual registers occurring in the IR version of the program to the finite set of hardware registers.
In order to do so, register allocation performs a relocation of register contents.
For each program point, a subset of the virtual registers is selected to be kept as hardware registers and the remaining virtual registers are spilled to the stack.
The literature describes various approaches towards selecting the set of registers to be spilled
\parencite{tichadouRegisterAllocation2022,polettoLinearScanRegister1999,traubQualitySpeedLinearscan1998,chaitinRegisterAllocationColoring1981},
leading to different transformations depending on the chosen algorithm.
%
%
In this paper, we express the transformations from register allocation more generally as a set of constraints.
A transformation constitutes a viable register allocation if it satisfies the constraints.
Practical allocation algorithms produce viable allocations making our results apply to all of them.

We first present the constraints a register allocation transformation needs to satisfy.
We then demonstrate how the transformation violates speculative non-interference preservation.
We continue to develop a static analysis that finds potential violations.
Finally, we fix the violations and craft a snippy simulation for the fixed transformation and prove that it satisfies $\propsnip$.

\subsection{Register Allocation}

Register allocation transforms a source program $\prog$ by inserting \emph{shuffle instructions}.
\begin{align*}
    \shuffleinstcmd \in \shuffleinstrs \;& \Coloneq\; \Imove{\rega}{\regb}{\pcsuc}
    \bnfalt \Ifill{\rega}{\stackl}{\pcsuc}
    \bnfalt \Ispill{\stackl}{\regb}{\pcsuc}
    \bnfalt \Islh{\rega}{\pcsuc}
    \bnfalt \Isfence{\pcsuc}
\end{align*}
The instructions extend $\instrs$ and are inserted in between the existing instructions in order to maintain the register relocation 
with the instructions move $\Cmove{\rega}{\regb}$, fill $\Cfill{\rega}{\stackl}$, and spill $\Cspill{\stackl}{\regb}$.
We also include $\IWslh{\rega}$ and $\IWsfence$ because we need them later to fix the transformation.
The semantics are as expected:
A move $\Cmove{\rega}{\regb}$ relocates contents from $\regb$ to $\rega$,
$\Cfill{\rega}{\stackl}$ reloads a spilled register from a constant address $\stackl \in \adrs$ in the stack,
and $\Cspill{\stackl}{\regb}$ spills a register to the stack.
We model the stack frame's section used for spilled registers with a fresh memory variable $\varstack$ not occurring in $\prog$.
We assume that $\varstack$ has appropriate size to fit all spilled registers and is typed $\lowsof{\varstack} = \low$.
Shuffle code is always straight line,
so for a shuffle sequence $\shufseqcmd \in \textit{Shuffle Inst}^*$ we introduce the notion $\progof{\pc} = \shufseq{\pcsuc}$ to express that $\prog$ executes $\shufseqcmd$ and then ends in $\pcsuc$.
%
%
\subsubsection*{Constraints for a valid transformation}

Register allocation inserts shuffle code between existing instructions to realize the register relocation.
Formally, a target program $\tra : \pcsof{\tra} \to \instrs$ is a register allocation if there exist functions $\instinj$ and $\regmap$.
The first is an injection $\instinj : \pcsof{\prog} \to \pcsof{\tra}$ of the original instructions of $\prog$ to their counterparts in~$\tra$.
The second function is a relocation mapping $\regmap : \pcsof{\tra} \to \regs \partialto (\regs \cup \stack)$,
where $\stack = \setcond{(\varstack, \stackl)}{\stackl \in \sizeof{\varstack}} \subseteq \mems$ is the stack frame for spilled variables.
At each program counter $\pcp \in \pcsof{\tra}$, $\regmapofof{\pcp}{\rega}$~is the relocation of the virtual register $\rega$ from $\prog$ to the hardware register or stack location in~$\tra$.%
\footnote{In general, the content of a register $\rega$ could be relocated to multiple locations. For simplicity of presentation, we forbid that.}
The functions are subject to the following conditions.
\begin{enumerate*}[font=\itshape]
    \item[Instruction matching:]\label{req:instmatch} The $\instinj$-injected instructions must operate on the same registers up to relocation by~$\regmap$.
    \item[Shuffle conformity:]\label{req:shufconf} The relocation $\regmap$ must conform to the (shuffle) instructions in $\tra$.
    \item[Obeying Liveness:]\label{req:liveness}~Every live register in $\prog$ is mapped under $\regmap$ and no location is doubly allocated.
\end{enumerate*}
To state the conditions formally, we introduce notation for the defined and used registers of an instruction.
A register $\regb$ is used by $\inst\in\instrs$ if the register is read out by the corresponding rule in \Cref{rules:spec-free}.
Similarly, $\rega$ is defined by $\inst$, if it is written in that rule.
For example, $\IWload{\rega}{\vara}{\regb}$ uses $\set{\regb}$ and defines $\set{\rega}$.
\begin{align*}
    \usesof{\inst} & = \setcond{\regb}{\text{$\regb \in \regs$ is used by $\inst$}}
                    &
    \defof{\inst} & = \setcond{\rega}{\text{$\regb \in \regs$ is defined by $\inst$}}
\end{align*}
\begin{code}[t]
\begin{doublecodebox}{}
    \ttfamily
    \tcbsubtitle[height=0pt,after skip=0pt]{\raggedleft Source $\prog$\phantom{aa}}
    \begin{lstlisting}
a = (b < buf_size)

br (a)|$\NextInst{\mkern1mu\makepc{3},\;\makepc{4}}$|
    buf[b] = secret

br (bytes)|$\NextInst{\mkern1mu\makepc{5},\;\makepc{5}}$|
ret
    \end{lstlisting}
\tcblower%
    \ttfamily
    \tcbsubtitle[height=0pt,after skip=0pt]{\raggedleft Target $\tra$}
    \begin{coderight}
a = (b < buf_size)
stk[0] = bytes
br (a)|$\NextInst{\mkern1mu\makepc{d},\;\makepc{e}}$|
    buf[b] = secret
a = stk[0]
br (a)|$\NextInst{\mkern1mu\makepc{g},\;\makepc{g}}$|
ret
    \end{coderight}
\end{doublecodebox}
\setlength{\belowcaptionskip}{-0.3cm}
\caption{\label{code:ra}An Example register allocation. Lines \makepc{b} and \makepc{e} are the inserted shuffle code.}
\end{code}

\myparagraph{Instruction Matching} requires that source instructions from $\prog$ reappear in $\tra$: 
For every $\pc \in \pcsof{\prog}$ and $\pcp = \instinjof{\pc}$, $\inst = \progof{\pc}$ has to match with $\instp = \traof{\pcp}$.
To match, the instruction $\instp$ must be the same,
but registers $\regb \in \usesof{\inst}$ are replaced with $\regmapofof{\pcp}{\regb} \in \regs$.
Similarly, registers $\rega \in \defof{\inst}$ are replaced with $\regmapofof{\pcsucp}{\rega} \in \regs$,
where $\pcsucp$ is the successor of $\pcp$.
Defined registers are found in the successor's relocation, because they are live only after executing $\instp$.
All other instructions at program points $\pcp \notin \imgof{\instinj}$ must be shuffle instructions.
Given $\statec \in \statesof{\prog}$ and $\stated \in \statesof{\tra}$,
we write $\statec \defrinstinj \stated$ when $\statec \atpc \pc$ and $\instinjof{\pc} \atpc \stated$.
We extend the notation to $\nstatec \rinstinj \nstated$ in the expected way.

\myparagraph{Shuffle Conformity} requires that the relocation $\regmap$ is upheld by the instructions in $\tra$.
Consider an instruction $\instp = \traof{\pcp}$ at $\pcp \in \pcsof{\tra}$ and let $\pcsucp$ be a successor.
First, registers and stack locations untouched by $\instp$ must stay at the same location in $\regmap$:
For any register $\regd$ with $\regmapofof{\pcp}{\regd} \notin \usesof{\instp}$
and $\regmapofof{\pcsucp}{\regd} \notin \defof{\instp}$,
$\regmapofof{\pcp}{\regd} = \regmapofof{\pcsucp}{\regd}$.
Second, if $\instp$ is a shuffle instruction there are additional requirements:
Shuffle instructions move one source register's content from one location in $\tra$ to another.
As a consequence, if $\instp$ is a shuffle instruction, there must be a source register $\rega$ being moved and the location moved to must be free.
We write $\freeof{\pcssp}{\regdp}$ and $\freeof{\pcssp}{(\varstack, \stackl)}$ for $\regdp, (\varstack, \stackl) \notin \imgof{\regmapof{\pcssp}}$.
Dependent on $\instp$, we add the following constraints, ($\rega$ is the source register):
\begin{align*}
    &\begin{aligned}
        \instp & = \Imove{\regap}{\regbp}{\pcsucp} : \\
        \instp & = \Ifill{\regap}{\stackl}{\pcsucp} : \\
        \instp & = \Ispill{\stackl}{\regbp}{\pcsucp} : \\
        \instp & = \Islh{\regap}{\pcsucp} : \\
    \end{aligned}
    &&&
    \begin{alignedat}{3}
        \exists \rega.&\regmapofof{\pcp}{\rega} = \regbp &&\land \regmapofof{\pcssp}{\rega} = \regap &&\land \freeof{\pcp}{\regap} \\
        \exists \rega.&\regmapofof{\pcp}{\rega} = (\varstack, \stackl) &&\land \regmapofof{\pcssp}{\rega} = \regap &&\land \freeof{\pcp}{\regap} \\
        \exists \rega.&\regmapofof{\pcp}{\rega} = \regbp &&\land \regmapofof{\pcssp}{\rega} = (\varstack, \stackl) &&\land \freeof{\pcp}{(\varstack, \stackl)} \\
        \exists \rega.&\regmapofof{\pcp}{\rega} = \regap &&\land \regmapofof{\pcssp}{\rega} = \regap
    \end{alignedat}
\end{align*}

\myparagraph{Obeying Liveness} means that all live variables in $\prog$ must be allocated.
There must be a Liveness solution $\flow$ (\Cref{lem:liveness}) for $\prog$
so that for all locations $\pc \in \pcsof{\prog}$,
all registers $\rega \in \flowof{\pc}$ live at $\pc$
are allocated, i.e.\ $\regmapofof{(\instinjof{\pc})}{\rega} \neq \bot$.
Further, a location cannot be allocated twice,
i.e.\ for all $\pc \in \pcsof{\tra}$, $\regmapof{\pc}$ forms an injection on the live registers at $\pc$.

\begin{definition}
    A transformation from $\prog$ to $\tra$ is a register allocation
    if there are instruction matching, shuffle conform, and Liveness obeying $(\instinj, \regmap)$.
\end{definition}
\begin{example}\label{example:ra}
    \Cref{code:ra} contains an example register allocation.
    It is a simplified version of \Cref{code:specv1}, which still exhibits the weakness of register allocation.
    The left program starts with a secret value in register \makereg{secret} and public values in \makereg{b} and \makereg{bytes}.
    It stores the secret into \makevar{buf} when the address \makereg{b} is in bounds of $\sizeof{\makevar{buf}}$.
    Then, it leaks \makereg{bytes} at \makepc{4}.
    The right program is after register allocation.
    Register allocation has inserted the spill and fill instructions \makepc{b} and \makepc{e}.
    The instruction injection $\instinj$ can be seen from side by side alignment.
    It is the mapping $\set{\makepc{1} \mapsto \makepc{b},\makepc{2} \mapsto \makepc{c},\makepc{3} \mapsto \makepc{d},\makepc{4} \mapsto \makepc{f},\makepc{5} \mapsto \makepc{g}}$.
    The register mapping $\regmap$ at \makepc{a} and \makepc{b} makes no relocations: $\regmapof{\makepc{a}} = \regmapof{\makepc{b}} = \identity{\regs}$.
    The spill instruction at \makepc{b} relocates \makereg{bytes} to the stack.
    It remains spilled from \makepc{c} to~\makepc{e}: $\regmapofof{\makepc{c}}{\makereg{bytes}} = \ldots = \regmapofof{\makepc{e}}{\makereg{bytes}} = (\makevar{stk},0)$.
    Further, because \makereg{a} is not live after \makepc{2}
    it is not allocated from \makepc{c} to \makepc{g}: $\regmapofof{\makepc{c}}{\makereg{a}} = \ldots = \regmapofof{\makepc{g}}{\makereg{a}} = \bot$.
    Instead, the register allocation reuses \makereg{a} in \makepc{f}:
    The fill instruction at \makepc{e} relocates \makereg{bytes} to \makereg{a}: $\regmapofof{\makepc{f}}{\makereg{bytes}} = \makereg{a}$.
\end{example}

\subsection{Register Allocation is not \texorpdfstring{$\propsnip$}{SNIP}}
\label{section:counterexample}

Transformations that satisfy the constraints for a valid register allocation
are known to be non-interference preserving, provided that the semantics are speculation-free \parencite{bartheStructuredLeakageApplications2021,bartheFormalVerificationConstanttime2019}.
We first assumed that the same is true for speculative semantics as well,
because register allocation only inserts shuffle instructions.
While shuffle instructions produce leakages,
they are constant address loads and stores
which means they only leak constant values.
To our surprise, when we tried to form a snippy simulation relation for register allocation,
we realized that it was not $\propsnip$.
\Cref{code:ra} demonstrates a minimalistic example for how register allocation introduces weaknesses into $\tra$.
The left program is $\propsni$,
because the leakage at Line \makepc{4} depends on the register \makereg{bytes} that holds a public value untouched since the start of the execution.
The right program, however, is susceptible to the same attack as described in \Cref{section:intro} for \Cref{code:specv1}:
Speculatively executing \makepc{c} when \makereg{b} holds a value out of $\sizeof{\makevar{buf}}$ can store the secret value from \makereg{secret} to \texttt{\makevar{stk}[0]}.
That value is then loaded at \makepc{e} into \makereg{a} and leaked at~\makepc{f}.
The fault for this attack is not with the constant leakage of shuffle instructions.
The weakness occurs from spilling itself, i.e.\ from the relocation of registers to memory:
Conceptually, speculative execution can access memory everywhere because it performs unsafe memory accesses.
But it cannot read or write register contents.
Register allocation moves registers into memory, effectively granting unsafe memory operations access to spilled registers.

We were able to reproduce the weakness with a real compiler.
\Cref{code:ra} is inspired by \Cref{code:specv1}
which is an excerpt of the \texttt{libsodium} function \texttt{chacha20\_encrypt\_bytes} with slight modifications.
The function is responsible for encryption of data using the \texttt{Chacha20} stream cipher and, similar to \Cref{code:specv1}, first copies the data into a stack-local buffer before executing the encryption algorithm on it.
Our tests were done on \texttt{LLVM 17}, and we tested each of \texttt{LLVM}'s register allocators (i.e. \texttt{basic}, \texttt{greedy}, \texttt{fast}, and \texttt{pbqp}).
They \emph{all} insert the instructions \makepc{b} and \makepc{e} (\Cref{code:ra}) when compiling our program without other optimizations.%
\footnote{Targeting \texttt{x86-64}. We used the \texttt{opt} passes \texttt{mem2reg,simplify-cfg,module-inline} before register allocation with \texttt{llc}.}
This means the vulnerability is inherent to register allocation, and cannot avoided by just opting towards a particular register allocator.

\subsection{Poison-Tracking Product}

We develop a static analysis that reveals weaknesses like the one in \Cref{section:counterexample}.
The analysis operates on a product construction between source program and register-allocated program.
To motivate the construction,
we inspect the simulation
for register allocation under speculation-free semantics,
and explain what fails under speculative semantics.
Under speculation-free semantics, a simulation for register allocation matches a state~$\statec$ of the source program $\prog$
to a state~$\stated$ of the target program $\tra$
whenever the instruction is matched, $\statec \rinstinj \stated$,
and the values of all registers and memory locations in $\statec$ and $\stated$ are equal up to relocation~$\regmap$.
When $\prog$ is memory safe,
one can then prove that any instruction executed from $\stated$ can be replayed from $\statec$.
The simulation preserves non-interference:
Apart from the leakages inserted by shuffle instructions (which are constant addresses),
the leakage that arises from execution in $\stated$ is equal to the leakage from $\statec$.
The previous section showed that it is not enough to just extend this approach to speculating states:
Speculation introduces unsafe memory operations that access the $\regmap$-mapped stack locations in $\tra$.
Such an access cannot be simulated by the source program~$\prog$, as it has the accessed value in a register where the memory operation cannot access it.
As a result, if we executed $\prog$ and $\tra$ in parallel we will eventually see speculating states $\nstatec$ and $\nstated$ with different values in registers of $\nstatec$ and their $\regmap$-mapped location in $\nstated$.
Further execution propagates the differences to other locations.

The goal in fixing register allocation is to make sure that differences in value do not lead to leakage of sensitive data.
An immediate mitigation would be to insert an $\IWsfence$ instruction before every load and store operation.
This would eliminate speculating unsafe memory accesses altogether.
However, this introduces many $\IWsfence$ instructions and reduces performance more than necessary.
Instead, we construct a product of $\prog$ and $\tra$ that tracks differences in values between $\nstatec$ and~$\nstated$.
We call registers and memory locations that hold different values \emph{poisoned},
because leaking them in $\tra$ might be unsafe:
If $\nstatec$ and $\nstated$ execute an instruction that leaks the value of a poisoned register,
we cannot rely on $\prog$ being $\propsni$ to justify that the leakage from $\nstated$ is safe.
Leakage of healthy (not poisoned) registers, however, is safe for that reason.
As a result, we only need to protect instructions that leak a poisoned register's value.

\subsubsection*{Poison types}
Poison types represent the registers and memory locations where $\prog$ and $\tra$ can have different values.
Consider states $\statec \in \statesof{\prog}$ and $\stated \in \statesof{\tra}$.
A poison type $\poison \in \ptypes$ is a function $(\regs \cup \mems \setminus \stack) \to \pvals$.
It assigns a poison value to each register and memory location of $\statec$.
The poison values $\pvals = \set{\poisoned, \weakpoisoned, \healthy}$ have the following meaning:
Healthy registers and memory locations~($\healthy$) are equal between $\statec$ and $\stated$ up to relocation by $\regmap$.
Poisoned registers and memory locations ($\poisoned$) can differ between $\statec$ and $\stated$ up to relocation by $\regmap$.
Finally, registers and memory locations can be weakly poisoned ($\weakpoisoned$).
This poison value is introduced because of $\IWslh{\rega}$.
If $\IWslh{\rega}$ occurs as a shuffle instruction, executing it speculatively sets $\rega$ to 0.
But because it is a shuffle instruction it does not occur in $\prog$.
That makes $\rega$'s value different between $\prog$ and $\tra$.
However, \emph{every} execution in $\tra$ that executes this shuffle instruction speculatively also sets $\rega$ to~0.
This makes it safe to leak in $\tra$, even though the value differs between $\prog$ and $\tra$.
We thus type registers and memory locations that hold a 0 due to an $\IWslh{\rega}$ instruction weakly poisoned~($\weakpoisoned$).
Formally, we express that states $\statec \in \statesof{\prog}$ and $\stated \in \statesof{\tra}$ are equal up to a poison-type $\poison$ and relocation~$\regmap$
with the notation $\statec \nrrmp{\poison} \stated$.
We then extend the notation to speculating states $\nstatec \in \nstatesof{\prog}$ and $\nstated \in \nstatesof{\tra}$ of equal speculation depth $\sizeof{\nstatec} = \sizeof{\nstated}$.
For that, we have sequences of poison types $\pseq \in \ptypes^*$,
one poison type for each level of speculation $\sizeof{\nstatec} = \sizeof{\nstated} = \sizeof{\pseq}$.
To formalize the notation $\nstatec \nrrmp{\pseq} \nstated$,
write $\pvof{\pval}{\poison}$ for the set of all registers and memory locations typed $\pval$ by~$\poison$,
$\pvof{\pval}{\poison} = \setcond{ \rega, (\vara, \adr) }{\poisonof{\rega} = \pval, \pmemof{\vara}{\adr} = \pval}$.
Further, we write $\evalin{\stated}{\rega} = \rasgniof{\stated}{\rega}$ and $\evalin{\stated}{(\vara, \adr)} = \masgniof{\stated}{\vara}{\adr}$.
Let $\statec = (\pc, \rasgn, \masgn) \in \statesof{\prog}$ and $\stated = (\pcp, \rasgnp, \masgnp) \in \statesof{\tra}$ range over states of $\prog$ and $\tra$.
We define:
\begin{align*}
        \begin{aligned}
            & \statec \defnrrmp{\poison} \stated \\
            & \ifftxt
        \end{aligned}
        &
        \begin{alignedat}{3}
            && \forall \rega \in \pvof{\healthy}{\poison}.\, &\rrmat{\stated}{\rega} = \rasgnof{\rega} \\
            &{}\land{}&\forall \rega \in \pvof{\weakpoisoned}{\poison}.\,& \rrmat{\stated}{\rega} = 0 \\
            &{}\land{}&\forall (\vara, \adr) \in \pvof{\healthy}{\poison}.\,& \masgnpof{\vara}{\adr} = \masgnof{\vara}{\adr} \\
            &{}\land{}&\forall (\vara, \adr) \in \pvof{\weakpoisoned}{\poison}.\,& \masgnpof{\vara}{\adr} = 0
        \end{alignedat} \label{def:nrrmp}
        &
        \begin{aligned}
                &\nstatec\lseq\statec \nrrmp{\pseq\lseq\poison} \nstated\lseq\stated \\
                &\ifftxt
        \end{aligned}
        &
        \begin{alignedat}{3}
            &&&\nstatec \nrrmp{\pseq} \nstated \\
            &{}\land{}&& \statec \nrrmp{\poison} \stated
        \end{alignedat}
\end{align*}
The definition of $\statec \nrrmp{\poison} \stated$ meets the intuition:
Healthy registers and memory locations need to coincide between $\statec$ and $\stated$ and weakly poisoned locations must be 0 in $\stated$.
For poisoned locations, there are no requirements.
For speculating states, the definition is applied at each speculation level.
We write $\healthy$, $\weakpoisoned$, and $\poisoned$ for the poison types that assign the respective poison value everywhere.
\begin{example}\label{example:ra-ptypes}
    We observe the poisoned values in the attack from \Cref{section:counterexample} in \Cref{code:ra} when running $\prog$ and $\tra$ side by side.
    The register allocation $(\instinj, \regmap)$ is described in \Cref{example:ra}.
    The initial states are $\statec = (\makepc{1}, \rasgn, \masgn)$ and $\stated = (\makepc{a}, \rasgn, \masgn)$,
    with $\rasgnof{\makereg{b}} \notin \sizeof{\makevar{buf}}$ and $\rasgnof{\makereg{secret}} = \val \neq \rasgnof{\makereg{bytes}}$.
    We also assume that already $\rasgnof{\makereg{a}} = \bfalse$ (so that the updates in \makepc{1} and \makepc{a} have no effect).
    After executing the store instruction at \makepc{3} with the directive $\Dstore{\makevar{buf}}{0}$,
    $\prog$'s state is $\statee = (\makepc{4}, \rasgn, \subst{\masgn}{(\makevar{buf},0)}{\val})$.
    Similarly, after executing the store instruction at $\makepc{d}$ with the directive~$\Dstore{\makevar{stk}}{0}$,
    $\tra$'s state is $\statef = (\makepc{e},\rasgn,\subst{\masgn}{(\makevar{stk}, 0)}{\val})$.
    We see a poisoned value for \makereg{bytes}:
    At $\makepc{e}$, it is still located in $\regmapofof{\makepc{e}}{\makereg{bytes}} = (\makevar{stk}, 0)$ due to the spill at \makepc{b}.
    But the values differ with $\rasgniof{\statef}{\makereg{bytes}} \neq \val = \masgnpof{\makevar{stk}}{0}$.
    Further, $\masgniof{\statee}{\makevar{buf}}{0} \neq \masgniof{\statef}{\makevar{buf}}{0}$ holds a poisoned value.
    With a poison type that is fully healthy except on $\makereg{bytes}$ and $(\makevar{buf}, 0)$, $\poison = \subst{\healthy}{\makereg{bytes},(\makevar{buf}, 0)}{\poisoned}$,
    we have $\statee \nrrmp{\poison} \statef$.
\end{example}
\begin{figure}
    \begin{ruleframes}
        \begin{ruleframe}[label=rules:pstack,top=4.3pt]{Stack}

            \definerule{poison-load-stkunsafe}
            {
                \instof{\statec} = \Iload{\rega}{\vara}{\regb}{\pcsuc}
                \\\\
                \poisonof{\regb} = \healthy
                \\
                \poisonp = \subst{\poison}{\rega}{\poisoned}
                \\\\
                \statec \trans{\Eload{\adr}}{\Dload{\varb}{\adrppp}} \statee
                \\
                \stated \trans{\Eload{\adr}}{\Dload{\varstack}{\adrp}} \statef
            }
            {
                (\statec, \stated, \poison)
                \prodtrans{\Eload{\adr}}{\Dload{\varb}{\adrppp}}{\Eload{\adr}}{\Dload{\varstack}{\adrp}}
                (\statee, \statef, \poisonp)
            }
            {\label{rule:pintro-load}}

            \definerule{poison-store-stkunsafe}
            {
                \instof{\statec} = \Istore{\vara}{\regb}{\regc}{\pcsuc}
                \\
                \regmapofof{\statef}{\regd} = (\varstack, \adrp)
                \\\\
                \poisonof{\regb} = \healthy
                \\
                \poisonp = \subst{\poison}{\regd,(\vara,\adrppp)}{\poisoned}
                \\\\
                \statec \trans{\Estore{\adr}}{\Dstore{\vara}{\adrppp}} \statee
                \\
                \stated \trans{\Estore{\adr}}{\Dstore{\varstack}{\adrp}} \statef
            }
            {
                (\statec, \stated, \poison)
                \prodtrans{\Estore{\adr}}{\Dstore{\vara}{\adrppp}}{\Estore{\adr}}{\Dstore{\varstack}{\adrp}}
                (\statee, \statef, \poisonp)
            }
            {\label{rule:pintro-store}}
        \end{ruleframe}

        \begin{ruleframewithoverlay}[label=rules:ploads,top=7pt]{Loads}{$\instof{\statec} = \Iload{\rega}{\vara}{\regb}{\pcsuc}$}
            \definerule{poison-load-safe}
            {
                \poisonof{\regb} = \weakpoisoned
                \\
                \statec \trans{\Eload{\adrpp}}{\Dstep} \statee
                \\
                \stated \trans{\Eload{0}}{\Dstep} \statef
            }
            {
                (\statec, \stated, \poison)
                \prodtrans{\Eload{\adrpp}}{\Dstep}{\Eload{0}}{\Dstep}
                (\statee, \statef, \subst{\poison}{\rega}{\poisoned})
            }
            {\label{rule:pload-ss}}

            \definerule{poison-load-unsafe}
            {
                \poisonof{\regb} = \weakpoisoned
                \\
                \statec \trans{\Eload{\adrpp}}{\Dload{\vara}{\adrppp}} \statee
                \\
                \stated \trans{\Eload{0}}{\Dstep} \statef
            }
            {
                (\statec, \stated, \poison)
                \prodtrans{\Eload{\adrpp}}{\Dload{\vara}{\adrppp}}{\Eload{0}}{\Dstep}
                (\statee, \statef, \subst{\poison}{\rega}{\poisoned})
            }
            {\label{rule:pload-us}}

            \definerule{healthy-load-safe}
            {
                \poisonof{\regb} = \healthy
                \\
                \statec \trans{\Eload{\adr}}{\Dstep} \statee
                \\
                \stated \trans{\Eload{\adr}}{\Dstep} \statef
            }
            {
                (\statec, \stated, \poison)
                \prodtrans{\Eload{\adr}}{\Dstep}{\Eload{\adr}}{\Dstep}
                (\statee, \statef, \subst{\poison}{\rega}{\pmemof{\vara}{\adr}})
            }
            {\label{rule:pload-hsafe}}

            \definerule{healthy-load-unsafe}
            {
                \poisonof{\regb} = \healthy
                \\
                \poisonp = \subst{\poison}{\rega}{\pmemof{\varb}{\adrp}}
                \\\\
                \statec \trans{\Eload{\adr}}{\Dload{\varb}{\adrp}} \statee
                \\
                \stated \trans{\Eload{\adr}}{\Dload{\varb}{\adrp}} \statef
                \\
                \varb \neq \varstack
            }
            {
                (\statec, \stated, \poison)
                \prodtrans{\Eload{\adr}}{\Dload{\varb}{\adrp}}{\Eload{\adr}}{\Dload{\varb}{\adrp}}
                (\statee, \statef, \poisonp)
            }
            {\label{rule:pload-hunsafe}}
        \end{ruleframewithoverlay}

        \begin{ruleframewithoverlay}[label=rules:pstores,top=4.3pt]{Stores}{$\instof{\statec} = \Istore{\vara}{\regb}{\regc}{\pcsuc}$}
            \definerule{poison-store-safe}
            {
                \poisonp = \subst{\poison}{(\vara, \adrpp), (\vara, 0)}{\poisoned}
                \\\\
                \poisonof{\regb} = \weakpoisoned
                \\
                \statec \trans{\Estore{\adrpp}}{\Dstep} \statee
                \\
                \stated \trans{\Estore{0}}{\Dstep} \statef
            }
            {
                (\statec, \stated, \poison)
                \prodtrans{\Estore{\adrpp}}{\Dstep}{\Estore{0}}{\Dstep}
                (\statee, \statef, \poisonp)
            }
            {\label{rule:pstore-ss}}

            \definerule{poison-store-unsafe}
            {
                \poisonof{\regb} = \weakpoisoned
                \\
                \statec \trans{\Estore{\adrpp}}{\Dstore{\vara}{\adrppp}} \statee
                \\
                \stated \trans{\Estore{0}}{\Dstep} \statef
            }
            {
                (\statec, \stated, \poison)
                \prodtrans{\Estore{\adrpp}}{\Dstore{\vara}{\adrppp}}{\Estore{0}}{\Dstep}
                (\statee, \statef, \subst{\poison}{(\vara,\adrppp),(\vara,0)}{\poisonof{\regc}})
            }
            {\label{rule:pstore-us}}

            \definerule{healthy-store-safe}
            {
                \poisonp = \subst{\poison}{(\vara, \adr)}{\poisonof{\regc}}
                \\\\
                \poisonof{\regb} = \healthy
                \\
                \statec \trans{\Estore{\adr}}{\Dstep} \statee
                \\
                \stated \trans{\Estore{\adr}}{\Dstep} \statef
            }
            {
                (\statec, \stated, \poison)
                \prodtrans{\Estore{\adr}}{\Dstep}{\Estore{\adr}}{\Dstep}
                (\statee, \statef, \poisonp)
            }
            {\label{rule:pstore-hsafe}}

            \definerule{healthy-store-unsafe}
            {
                \poisonof{\regb} = \healthy
                \\
                \statec \trans{\Estore{\adr}}{\Dstore{\varb}{\adrp}} \statee
                \\
                \stated \trans{\Estore{\adr}}{\Dstore{\varb}{\adrp}} \statef
                \\
                \varb \neq \varstack
            }
            {
                (\statec, \stated, \poison)
                \prodtrans{\Estore{\adr}}{\Dstore{\varb}{\adrp}}{\Estore{\adr}}{\Dstore{\varb}{\adrp}}
                (\statee, \statef, \subst{\poison}{(\varb, \adrp)}{\poisonof{\regc}})
            }
            {\label{rule:pstore-hunsafe}}
        \end{ruleframewithoverlay}

        \begin{ruleframe}[label=rules:pshuffle,top=4.3pt]{Shuffles}
            \definerule{poison-fill}
            {
                \instof{\stated} = \Ifill{\regap}{\stackl}{\pcsuc}
                \\
                \stated \trans{\Eload{\stackl}}{\Dstep} \statef
            }
            {
                (\statec, \stated, \poison)
                \prodtrans{\varepsilon}{\varepsilon}{\Eload{\stackl}}{\Dstep}
                (\statec, \statef, \poison)
            }
            {\label{rule:pfill}}

            \definerule{poison-spill}
            {
                \instof{\stated} = \Ispill{\stackl}{\regbp}{\pcsuc}
                \\
                \stated \trans{\Estore{\stackl}}{\Dstep} \statef
            }
            {
                (\statec, \stated, \poison)
                \prodtrans{\varepsilon}{\varepsilon}{\Estore{\stackl}}{\Dstep}
                (\statec, \statef, \poison)
            }
            {\label{rule:pspill}}

            \definerule{poison-move}
            {
                \instof{\stated} = \Imove{\regap}{\regbp}{\pcsuc}
                \\
                \stated \trans{\Enone}{\Dstep} \statef
            }
            {
                (\statec, \stated, \poison)
                \prodtrans{\varepsilon}{\varepsilon}{\Enone}{\Dstep}
                (\statec, \statef, \poison)
            }
            {\label{rule:pmove}}

            \definerule{poison-shuffle-sfence}
            {
                \instof{\stated} = \Isfence{\pcsuc}
                \\
                \stated \trans{\Enone}{\Dstep} \statef
            }
            {
                (\statec, \stated, \poison)
                \sprodtrans{\varepsilon}{\varepsilon}{\Enone}{\Dstep}
                (\statec, \statef, \poison)
            }
            {\label{rule:pshuf-sfence}}

            \definerule{poison-shuffle-slh}
            {
                \instof{\nstated} = \Islh{\regap}{\pcsuc}
                \\
                \regmapofof{\nstated}{\rega} = \regap
                \\
                \nstated \trans{\Enone}{\Dstep} \nstatef
                \\
                \pval = \inlineife{\sizeof{\pseq} \sameas 0}{\poisonof{\rega}}{\weakpoisoned}
            }
            {
                (\nstatec, \nstated, \pseq\lseq\poison)
                \sprodtrans{\varepsilon}{\varepsilon}{\Enone}{\Dstep}
                (\nstatec, \nstatef, \pseq\lseq\subst{\poison}{\rega}{\pval})
            }
            {\label{rule:pshuf-slh}}
        \end{ruleframe}
        \begin{ruleframe}[label=rules:pspec,top=4.3pt,bottom=0pt]{Speculation Sensitive}
            \definerule{poison-step}
            {
                (\statec, \stated, \poison)
                \prodtrans{\leak}{\direct}{\leakp}{\directp}
                (\statee, \statef, \poisonp)
            }
            {
                (\nstatec\lseq\statec, \nstated\lseq\stated, \pseq\lseq\poison)
                \sprodtrans{\leak}{\direct}{\leakp}{\directp}
                (\nstatec\lseq\statee, \nstated\lseq\statef, \pseq\lseq\poisonp)
            }
            {\label{rule:pstep}}

            \definerule{poison-rollback}
            {
                \sizeof{\nstatec} = \sizeof{\nstated} = \sizeof{\pseq} \geq 1
            }
            {
                (\nstatec\lseq\statec, \nstated\lseq\stated, \pseq\lseq\poison)
                \sprodtrans{\Erbless}{\Drb}{\Erbless}{\Drb}
                (\nstatec, \nstated, \pseq)
            }
            {\label{rule:prb}}

            \definerule{healthy-spec}
            {
                \instof{\nstatec} = \Iif{\regb}{\pcsuc_{\btrue}}{\pcsuc_{\bfalse}}
                \\
                \poisonof{\regb} = \healthy
                \\
                \nstatec \trans{\Eif{\bvalue}}{\Dspec} \nstatec\lseq\statee
                \\
                \nstated \trans{\Eif{\bvalue}}{\Dspec} \nstated\lseq\statef
            }
            {
                (\nstatec, \nstated, \pseq\lseq\poison)
                \sprodtrans{\Eif{\bvalue}}{\Dspec}{\Eif{\bvalue}}{\Dspec}
                (\nstatec\lseq\statee, \nstated\lseq\statef, \pseq\lseq\poison\lseq\poison)
            }
            {\label{rule:hspec}}

            \definerule{poison-sfence}
            {
                \instof{\statec} = \Isfence{\pcsuc}
                \\
                \statec \trans{\Enone}{\Dstep} \statee
                \\
                \stated \trans{\Enone}{\Dstep} \statef
            }
            {
                (\statec, \stated, \poison)
                \sprodtrans{\Enone}{\Dstep}{\Enone}{\Dstep}
                (\statee, \statef, \poison)
            }
            {\label{rule:psfence}}

            \definerule{poison-slh}
            {
                \instof{\nstatec} = \Islh{\rega}{\pcsuc}
                \\
                \nstatec \ntrans{\Enone}{\Dstep} \nstatee
                \\
                \nstated \ntrans{\Enone}{\Dstep} \nstatef
                \\
                \pval = \inlineife{\sizeof{\pseq} \sameas 0}{\pregof{\rega}}{\healthy}
            }
            {
                (\nstatec, \nstated, \pseq\lseq\poison)
                \sprodtrans{\Enone}{\Dstep}{\Enone}{\Dstep}
                (\nstatee, \nstatef, \pseq\lseq\subst{\poison}{\rega}{\pval})
            }
            {\label{rule:pslh}}
        \end{ruleframe}
    \end{ruleframes}
\end{figure}

\subsubsection*{Poison product}
The poison product $\prodra$ tracks how poison types are updated
when executing $\prog$ and $\tra$ side by side.
We design the side by side execution so that transitions in $\tra$ are replayed by $\prog$.
We later want to craft a simulation from the product,
so every transition of $\tra$ must be included in $\prodra$.
However, when a transition of $\tra$ can be replayed by multiple transitions in $\prog$,
which happens when $\prog$ performs an unsafe memory access,
we choose a memory location for the access so that the transition poisons as few registers and memory locations as possible.
The states of $\prodra$ take the shape $(\nstatec,\nstated,\pseq)$
with $\nstatec \in \nstatesof{\prog}$, $\nstated \in \nstatesof{\tra}$, $\pseq \in \ptypes^*$,
and~$\nstatec \nrrmp{\pseq} \nstated$
(which already implies $\sizeof{\nstatec} = \sizeof{\nstated} = \sizeof{\pseq}$).
There are two types of states:
Instruction-matched states and shuffling states.
Instruction-matched states satisfy $\nstatec \rinstinj \nstated$.
Shuffling states have $\nstated$ at a shuffle sequence after which the states would be instruction-matched again.
That is, $\nstatec = \nstatecp\lseq\statec$ and $\nstated = \nstatedp\lseq\stated$,
so that $\traof{\stated} = \shufseq{\pcsuc}$ with $\instinjof{\statec} = \pcsuc$.
Each type of state has transitions in $\prodra$:
\begin{align*}
    (\nstatec, \nstated, \pseq) &\sprodtrans{\leakp}{\directp}{\leak}{\direct} (\nstatee, \nstatef, \pseqp) & 
    (\nstatec, \nstated, \pseq) &\sprodtrans{\varepsilon}{\varepsilon}{\leak}{\direct} (\nstatec, \nstatef, \pseqp)
\end{align*}
The first type of transition is enabled only for instruction-matched $\nstatec \rinstinj \nstated$.
It executes a transition simultaneously in both programs via $\nstatec \ntrans{\leakp}{\directp} \nstatee$ and $\nstated \ntrans{\leak}{\direct} \nstatef$.
The second type is enabled only for shuffling states.
In that case, $\nstatec$ stutters while $\nstated$ progresses through the shuffle sequence.
In both cases, the transition needs to update the poison values depending on the transition rules taken in $\prog$ and $\tra$.
We present how $\prodra$ updates poison values in \Cref{rules:pstack,rules:ploads,rules:pstores,rules:pspec,rules:pshuffle}.
We mirror the approach in the semantics and provide the updates in two steps.
In a first step, we define the updates for instruction-matched, speculation-free states in \Cref{rules:pstack,rules:ploads,rules:pstores}.
They represent the transitions for the speculation-free semantics (\Cref{rules:spec-free}).
In the second step we lift the speculation-free updates with \Cref{rules:pspec}
to define the transitions of the product on the speculative semantics (\Cref{rules:spec}).
\Cref{rules:pshuffle} defines the transitions for shuffling states.
The instruction-matched transitions for $\IWnop$, $\IWasgn{\rega}{\regb\op\regc}$ and $\IWif{\regb}$ can be found in the appendix.
We explain the transitions in detail.
The transitions \labelcref{rule:pintro-load,rule:pintro-store} are the source of poison values.
They represent the situation from above, where we argued that values between $\prog$ and $\tra$ can differ:
In \labelcref{rule:pintro-load} a target program's speculating unsafe load (recognizable from the directive $\Dload{\varstack}{\adrp}$) loads from a spilled register's stack location.
The source program cannot perform that load and thus has to perform any other unsafe load.
This leads to different loaded values and poisons the register loaded to.
While we could choose the location of the unsafe access in $\prog$,
the register is poisoned in either case, so we allow arbitrary unsafe loads in $\prog$.
Similarly, in \labelcref{rule:pintro-store} a speculating unsafe store poisons the source program's register $\regd$ by overwriting it on the stack.
Again, the source program cannot replay that store as $\regd$ is a register and needs to store somewhere else.
In this case, we always let the source program store to $\vara$ instead,
because this reduces the number of poisoned memory locations for static analysis.
The overwritten stack-allocated register $\regd$ and the memory location in $\vara$ are both poisoned.
Notice how both rules require the source program to access memory unsafely, as well.
This is because the register holding the offset address must be healthy, $\pregof{\regb} = \healthy$,
and thus has the same value in $\statec$ and $\stated$.
The address must be healthy, because a poisonous address could be unsafe to leak.
A weakly poisoned address, on the other hand, would be 0 in $\tra$ which is always a safe access.
This is because we chose $\IWslh{\rega}$ to wipe a register to 0.
Any other constant would work the same way, but the number of rules in the product would increase.

\Cref{rules:ploads,rules:pstores} for instruction-matched transitions just propagate already poisoned registers and memory locations.
For example, \labelcref{rule:pload-hsafe} considers the case where the addressing register is healthy.
That means that both source and target program load from the same memory location.
We can recognize from the premise that the memory access is safe:
Source and target program execute the load with the directive $\Dstep$, so they execute with \labelcref{rule:load}.
The rule then propagates the poison value from the memory location loaded from to the register loaded to.
As before, the product transition can only be taken when the leaked content is not poisoned as that could represent an unsafe leakage.
The case of a weakly poisoned address is separately handled by \labelcref{rule:pload-ss,rule:pload-us}.
\Cref{rules:pstores} propagate poison values for stores.

\Cref{rules:pspec} provides the poison updates for instruction-matched states that execute speculation sensitive instructions.
\Cref{rule:pstep} lifts the speculation-free transitions to the speculative setting.
\labelcref{rule:hspec} forbids weakly poisoned branching conditions even though they are safe to leak.
This is to avoid that $\prog$ and $\tra$ arrive at different program points (up to $\instinj$).\footnote{Branching on weakly poisoned registers could be supported but further complicates the product definition.}
%
%
\labelcref{rule:pslh} creates healthy values when the product is speculating, because source and target program both wipe the register to 0.
If not speculating, the poison values are untouched.

\Cref{rules:pshuffle} provides the second type of transitions, for shuffling states, in $\prodra$.
The speculation insensitive shuffle instructions are again brought to speculating states with \labelcref{rule:pstep}.
The insensitive shuffle instructions only move values and don't modify them.
There is no need to update the poison type because the relocation is already included in $\regmap$ with shuffle conformity.
Thus, the poison type is untouched.
Only \labelcref{rule:pshuf-slh} can be a new source of poisoned values.
As discussed earlier, when speculatively executed, it wipes the register's value to 0, which makes it weakly poisoned.
On speculation-free states, $\IWsfence$ and $\IWslh{\regap}$ do not update the poison type.

The transitions of $\prodra$ are well-defined in the sense that executing a transition leads to a state that again satisfies the constraints of being a state (\Cref{lem:prod-well-def}).
Also, the speculation-free states never become poisoned,
because \labelcref{rule:pintro-load,rule:pintro-store} cannot happen:
$\prog$ is memory safe under speculation-free semantics.
Further, the only other rule to introduce poison values is \labelcref{rule:pshuf-slh},
but it introduces no poison value in speculation-free states.

\begin{lemma}\label{lem:prod-well-def}
    Given $(\nstatec, \nstated, \pseq)$ and a transition $(\nstatec, \nstated, \pseq) \sprodtrans{\leakp}{\directp}{\leak}{\direct} (\nstatee, \nstatef, \pseqp)$,
    then $\nstatee \nrrmp{\pseqp} \nstatef$.
\end{lemma}

\begin{lemma}\label{lem:spec-free-prod}
    Speculation-free states are never poisoned.
    I.e.\ $(\statec, \stated, \healthy) \sprodtranss{\tracek}{\dtracee}{\tracel}{\dtraced} (\statee, \statef, \poison)$
    implies $\poison = \healthy$.
\end{lemma}

\begin{example}
    Consider the execution of $\prodra$ in \Cref{fig:prodrun}.
    It depicts the side by side execution from \Cref{example:ra-ptypes} in $\prodra$.
    Each state of $\prodra$ is depicted as a small table listing the program counters of $\prog$ and $\tra$,
    the values of the registers, and the values of relevant memory locations (we use $\btrue,\bfalse$ for Boolean typed values).
    The right column contains the poison type associated with the pair of states.
    The first transition simultaneously executes the instruction-matched assignment at \makepc{1} and~\makepc{a}.
    The second transition executes the shuffle instruction that spills \makereg{bytes} to $\makevar{stk}[0]$ at \makepc{b}, while $\prog$ waits.
    The third transition speculates.
    Notice that \labelcref{rule:hspec} only admits this transition because \makereg{a} is $\healthy$.
    The next transition is the unsafe store to \makevar{stk} in $\tra$,
    overwriting the spilled \makereg{bytes}.
    Notice that \labelcref{rule:pintro-store} has changed the poison value for \makereg{bytes} to $\poisoned$,
    because the memory location written to was $\regmapofof{\makepc{e}}{\makereg{bytes}} = (\makevar{stk},0)$.
    The rule also poisons the memory location where $\prog$ writes to.
    The last transition is the shuffle instruction that relocates the poisoned \makereg{bytes} to~\makereg{a}.
    The product is now stuck: \labelcref{rule:hspec} is disabled because the poison value for \makereg{bytes} is $\poisoned$ and thus not safe to leak in $\tra$.
    Indeed, $\regmapofof{\makepc{f}}{\makereg{bytes}} = \makereg{a}$ holds the secret value 42.
\end{example}
\begin{figure}
    \begin{tikzpicture}
        [every node/.style={inner sep=0pt,outer sep=0pt}]
        \tikzmath{ \distv = 1; \disth = 0.5; }

        \node[scale=0.7] (legend) {%
                \begin{NiceTabular}{r}[color-inside]
                    \hline
                    \\
                    \hline
                    \makereg{bytes} \\
                    \makereg{a} \\
                    \makereg{b} \\
                    \makereg{secret} \\
                    \makevar{buf}[0] \\
                    \makevar{stk}[0] \\
                    \hline
                \end{NiceTabular}
            };
        \node[scale=0.7,above=0 of legend] (legend2) {%
                \begin{NiceTabular}{r}[color-inside]
                    \hline
                    \makereg{bytes} \\
                    \makereg{a} \\
                    \makereg{b} \\
                    \makereg{secret} \\
                    \makevar{buf}[0] \\
                    \makevar{stk}[0]
                \end{NiceTabular}
            };
        \node[scale=0.7,anchor=south,right=0.2 of legend] (a) {%
                \begin{NiceTabular}{|r|r|l|}[color-inside]
                    \hline
                    \makepc{1} & \makepc{a} & \\
                    \hline
                    32 & 32 & $\healthy$\\
                    \deadcell $\bfalse$ & \deadcell $\bfalse$ & $\healthy$\\
                    8 & 8 & $\healthy$\\
                    42 & 42 & $\healthy$\\
                    0 & 0 & $\healthy$\\
                     & 0 & \\
                    \hline
                \end{NiceTabular}
            };
        \node[scale=0.7,anchor=south,right=\disth of a] (b) {%
                \begin{NiceTabular}{|r|r|l|}[color-inside]
                    \hline
                    \makepc{2} & \makepc{b} & \\
                    \hline
                    32 & 32 & $\healthy$\\
                    \cellcolor{peach!60}$\bfalse$ & \cellcolor{peach!60}$\bfalse$ & $\healthy$\\
                    8 & 8 & $\healthy$\\
                    42 & 42 & $\healthy$\\
                    0 & 0 & $\healthy$\\
                     & 0 \\
                    \hline
                \end{NiceTabular}
            };
        \node[scale=0.7,anchor=south,right=\disth of b] (c) {%
                \begin{NiceTabular}{|r|r|l|}[color-inside]
                    \hline
                    \makepc{2}& \makepc{c} & \\
                    \hline
                    32 & \deadcell 32 & $\healthy$\\
                    $\bfalse$ & $\bfalse$ & $\healthy$\\
                    8 & 8 & $\healthy$\\
                    42 & 42 & $\healthy$\\
                    0 & 0 & $\healthy$\\
                     & \cellcolor{peach!60}32 & \\
                    \hline
                \end{NiceTabular}
            };
        \node[scale=0.7,right=\disth of c] (d1) {%
                \begin{NiceTabular}{|r|r|l|}[color-inside]
                    \hline
                    \makepc{2}& \makepc{c} & \\
                    \hline
                    32 & \deadcell 32 & $\healthy$\\
                    $\bfalse$ & $\bfalse$ & $\healthy$\\
                    8 & 8 & $\healthy$\\
                    42 & 42 & $\healthy$\\
                    0 & 0 & $\healthy$\\
                     & 32 & \\
                    \hline
                \end{NiceTabular}
            };
        \node[scale=0.7,anchor=north,above=0 of d1,fill=peach!60] (d2) {%
                \begin{NiceTabular}{|r|r|l|}[color-inside]
                    \hline
                    \makepc{3} & \makepc{d} & \\
                    \hline
                    32 & \deadcell 32 & $\healthy$\\
                    \deadcell $\bfalse$ & \deadcell  $\bfalse$ & $\healthy$\\
                    8 & 8 & $\healthy$\\
                    42 & 42 & $\healthy$\\
                    0 & 0 & $\healthy$\\
                      & 32 &
                \end{NiceTabular}
            };

        \node[scale=0.7,right=\disth of d1] (e1) {%
                \begin{NiceTabular}{|r|r|l|}[color-inside]
                    \hline
                    \makepc{2}& \makepc{c} & \\
                    \hline
                    32 & \deadcell 32 & $\healthy$\\
                    $\bfalse$ & $\bfalse$ & $\healthy$\\
                    8 & 8 & $\healthy$\\
                    42 & 42 & $\healthy$\\
                    0 & 0 & $\healthy$\\
                      & 32 & \\
                    \hline
                \end{NiceTabular}
            };
        \node[scale=0.7,anchor=north,above=0 of e1] (e2) {%
                \begin{NiceTabular}{|r|r|l|}[color-inside]
                    \hline
                    \makepc{4} & \makepc{e} & \\
                    \hline
                    32 & \deadcell 32 & \cellcolor{peach!60}$\poisoned$\\
                    \deadcell $\bfalse$ & \deadcell  $\bfalse$ & $\healthy$\\
                    8 & 8 & $\healthy$\\
                    42 & 42 & $\healthy$\\
                    \cellcolor{peach!60} 42 & 0 & \cellcolor{peach!60} $\poisoned$\\
                     & \cellcolor{peach!60} 42
                \end{NiceTabular}
            };

        \node[scale=0.7,right=\disth of e1] (f1) {%
                \begin{NiceTabular}{|r|r|l|}[color-inside]
                    \hline
                    \makepc{2}& \makepc{c} & \\
                    \hline
                    32 & \deadcell 32 & $\healthy$\\
                    $\bfalse$ & $\bfalse$ & $\healthy$\\
                    8 & 8 & $\healthy$\\
                    42 & 42 & $\healthy$\\
                    0 & 0 & $\healthy$\\
                     & 32 & \\
                    \hline
                \end{NiceTabular}
            };
        \node[scale=0.7,anchor=north,above=0 of f1] (f2) {%
                \begin{NiceTabular}{|r|r|l|}[color-inside]
                    \hline
                    \makepc{4} & \textcolor{peach}{\makepc{f}} & \\
                    \hline
                    32 & \deadcell 32 & $\poisoned$\\
                    \deadcell $\bfalse$ & \cellcolor{peach!60} 42 & $\healthy$\\
                    8 & 8 & $\healthy$\\
                    42 & 42 & $\healthy$\\
                    42 & 0 & $\poisoned$ \\
                     & 42 & 
                \end{NiceTabular}
            };
        \tikzmath{ \offset = 0.03;}
        \draw[ntrans] ([shift={(0,-\offset)}]a.north east) -- ([shift={(0.1,-\offset)}]a.north east)
            -- node[right=0.2,scale=.7,anchor=south east,pos=0.9] {\rotatebox{270}{\labelcref{rule:pasgn}}} ([shift={(0.1,\offset)}]a.south east) -- ([shift={(0,\offset)}]b.south west);
        \draw[ntrans] ([shift={(0,-\offset)}]b.north east) -- ([shift={(0.1,-\offset)}]b.north east)
            -- node[right=0.2,scale=.7,anchor=south east,pos=0.9] {\rotatebox{270}{\labelcref{rule:pspill}}} ([shift={(0.1,\offset)}]b.south east) -- ([shift={(0,\offset)}]c.south west);
        \draw[ntrans] ([shift={(0,-\offset)}]c.north east) -- ([shift={(0.1,-\offset)}]c.north east)
            -- node[right=0.2,scale=.7,anchor=south east,pos=0.9] {\rotatebox{270}{\labelcref{rule:hspec}}} ([shift={(0.1,\offset)}]c.south east) -- ([shift={(0,\offset)}]d1.south west);
        \draw[ntrans] ([shift={(0,-\offset)}]d2.north east) -- ([shift={(-0.1,-\offset)}]e2.north west)
            -- node[left=0.2,scale=.7,anchor=north west,pos=0.1] {\rotatebox{270}{\labelcref{rule:pintro-store}}} ([shift={(-0.1,\offset)}]e2.south west) -- ([shift={(0,\offset)}]e2.south west);
        \draw[ntrans] ([shift={(0,-\offset)}]e2.north east) -- ([shift={(-0.1,-\offset)}]f2.north west)
            -- node[left=0.2,scale=.7,anchor=north west,pos=0.1] {\rotatebox{270}{\labelcref{rule:pfill}}} ([shift={(-0.1,\offset)}]f2.south west) -- ([shift={(0,\offset)}]f2.south west);
    \end{tikzpicture}
    \caption{\label{fig:prodrun} An execution of $\prodra$ on \Cref{code:ra}. Updated values are \textcolor{peach!60}{highlighted}. Dead registers are \textcolor{overlay2}{gray}.}
\end{figure}

\subsection{Static poison analysis of \texorpdfstring{$\prodra$}{the product}}

The poison product finds weaknesses:
Whenever a poisoned register's value would be leaked in~$\tra$,
the product cannot execute the transition.
To find the program points where a poisoned register can be leaked,
we design a static analysis that over-approximates the poison values any execution could produce.
The analysis constructs, for each program point pair $(\pc, \pcp)$, an approximate poison type~$\poison$.
This poison type is approximate in that a statically $\healthy$-typed register is always healthy in any execution that reaches $(\pc, \pcp)$ in a speculating state.
Similarly, if it is statically typed $\weakpoisoned$ it is always weakly poisoned in a speculating state.
However, if it is statically typed $\poisoned$ it might also be healthy or weakly poisoned.
Formally, the analysis constructs a function $\passigns : \prodpcs \to \ptypes$, where $\prodpcs$ are the program points of $\prodra$,
\begin{align*}
    \prodpcs & \defeq \setcond{(\pc, \pcp)}{\exists (\nstatec\lseq\statec, \nstated\lseq\stated, \pseq\lseq\poison).\, \pc \atpc \statec \land \pcp \atpc \stated} \\
             & = \setcond{(\pc, \pcp), (\pc, \pcssp)}{\instinjof{\pc} = \pcp \land \traof{\pcssp} = \shufseq{\pcp}}\,.
\end{align*}
The pairs $(\pc, \pcp)$ are instruction-matched program counters and $(\pc, \pcssp)$ are from shuffling states.
%

We design our analysis as a forward flow analysis (\Cref{eqn:flow-fwd}).
For that, we need to define a flow lattice $\lattice$, transfer functions, and the initial flow value $\flowinit$.
The flow analysis then yields $\passigns$ as a solution.
The lattice is constructed on $\ptypes$.
We create an ordering $\healthy < \poisoned$ and $\weakpoisoned < \poisoned$ on $\pvals$.
In order to arrive at a lattice,
we further extend it by an artificial least element $\bot \in \pvals$.\footnote{This is a standard construction. In the remainder, we assume that transfer functions preserve $\bot$.}
We then lift the ordering point-wise to poison types for the flow lattice $\lattice = (\ptypes, \leq)$.
The ordering is chosen with the intention that when $\poison \leq \poisonp$,
then $\statec \nrrmp{\poison} \stated$ implies $\statec \nrrmp{\poisonp} \stated$ (which would not be the case if we had set $\healthy < \weakpoisoned$).
The initial value is set to healthy, $\flowinit = \healthy$, because the initial states of $\prog$ and $\tra$ are fully equal up to $\regmap$.
The transfer functions $\ptransferat{(\pc, \pcp)} : \ptypes \to \ptypes$ need to approximate the poison types in $\prodra$.
Conceptually, they update the current poison type by simultaneously executing all updates that $\prodra$ could do.
For instruction-matched $\instinjof{\pc} = \pcp$ that means to look at the instruction $\inst = \progof{\pc}$ which is the same as $\traof{\pcp}$ up to $\regmap$.
Then, we poison all registers and memory locations that a rule for $\inst$ from \Cref{rules:pstack,rules:ploads,rules:pstores,rules:pspec} could poison.
This means the transfer function is solely dependent on the instruction $\inst = \progof{\pc}$, and we define it via $\ptransferat{(\pc, \pcp)} = \ptransferat{\inst}$ below.
For a shuffling product state, $\prodra$ offers for each instruction $\instp = \traof{\pcssp}$ only one update which we find in \Cref{rules:pshuffle}.
We again define transfer solely dependent on $\instp$ via $\ptransferat{(\pc, \pcssp)} = \ptransferpat{\instp}$.
For a shuffling $\instp = \IWslh{\regap}$ we assume that the source register for $\regap$ is $\rega$, i.e.\ $\regmapofof{\pcssp}{\rega} = \regap$.
We only present the interesting cases of $\ptransferat{\inst}$ and $\ptransferpat{\instp}$.
The initial poison type for $(\entry, \entryp)$ is healthy.
\begin{align*}
    &
    \begin{aligned}
        \ptransferof{\Iload{\rega}{\vara}{\regb}{\pcsuc}}{\poisons} & = \subst{\poisons}{\rega}{\poisoned} \\
        \ptransferof{\Islh{\rega}{\pcsuc}}{\poisons} & = \subst{\poisons}{\rega}{\healthy} \\
        \ptransferpof{\Islh{\regap}{\pcsuc}}{\poisons} & = \subst{\poisons}{\rega}{\weakpoisoned} \\
        \ptransferpof{\Isfence{\pcsuc}}{\poisons} &= \healthy
    \end{aligned}
    &
    \begin{aligned}
        \ptransferof{\Iasgn{\rega}{\regb \op \regc}{\pcsuc}}{\poisons} & =
        \begin{cases}
            \subst{\poisons}{\rega}{\healthy} & \poisonsof{\regb} = \poisonsof{\regc} = \healthy \\
            \subst{\poisons}{\rega}{\poisoned} & \text{otherwise}
        \end{cases} \\
        \ptransferof{\Istore{\vara}{\regb}{\regc}{\pcsuc}}{\poisons} & =
        \subst{\substjoin{\poisons}{\mems}{\poisonsof{\regc}}}{\regs, \vara}{\poisoned} \\
    \end{aligned}
\end{align*}
All transfer functions are monotonic.
All but one transfer function are easily defined in order to approximate the rules of $\prodra$.
The exception is for $\IWstore{\vara}{\regb}{\regc}$ instructions which we explain.
The first substitution $\substjoin{\poisons}{\mems}{\poisonsof{\regc}}$ sets all memory locations $(\varb, \adr)$ to $\poisonsof{(\varb, \adr)} \join \poisonof{\regc}$.
This approximates \labelcref{rule:pload-hsafe,,rule:pload-hunsafe}, because both $\prog$ and $\tra$ store to the same location which will have poison value $\poisonof{\regc}$.
All other locations maintain their poison value.
To approximate this behavior statically, we take the join on the two poison values.
The second substitution sets all of $\vara$ and all registers $\regs$ to $\poisoned$.%
\footnote{We could be more precise and poison only those registers spilled at the current program point.}
This approximates \labelcref{rule:pintro-store}: $\regs$ needs to be poisoned because a store to the stack overwrites the contents of a register in the source program.
At the same time, the rule overwrites $\vara$ in the source program $\prog$, which leads to $\vara$ being poisoned as well.
Additionally, \labelcref{rule:pstore-ss,rule:pstore-us} write to $\vara$ only, meaning the poisoning of $\vara$ already approximates them as well.

In order to formally express how a solution $\passigns$ approximates the poison values of all reachable states in $\prodra$,
we introduce shorthand notations for sequences of poison types fully consisting of poison types from $\passigns$.
Fix a solution $\passigns$ to the flow equations.
We define:
\begin{mathpar}[\lineskiplimit=0.6em\lineskip=0.5em plus 0.2em]
    \pseqofof{\statec}{\stated} = \healthy
    \and
    \pseqofof{\nstatec\lseq\statec}{\nstated\lseq\stated} = \pseqofof{\nstatec}{\nstated} \lseq \passignsof{(\statec, \stated)}
    \\
    \nstatec \defrelrm \nstated  \Leftrightarrow \nstatec \nrrmp{\pseqofof{\nstatec}{\nstated}} \nstated
    \and
    \nstatec\lseq\statec \defrrmp{\poison} \nstated\lseq\stated  \Leftrightarrow \nstatec\lseq\statec \nrrmp{\pseqofof{\nstatec}{\nstated}\lseq\poison} \nstated\lseq\stated 
    \and
    (\nstatec\lseq\statec, \nstated\lseq\stated, \poison) = (\nstatec\lseq\statec, \nstated\lseq\stated, \pseqofof{\nstatec}{\nstated}\lseq\poison)
\end{mathpar}
The notation $\pseqofof{\nstatec}{\nstated}$, where $\nstatec$, $\nstated$ are speculating states,
stands for a sequence of poison types of length $\sizeof{\pseqofof{\nstatec}{\nstated}} = \sizeof{\nstatec} = \sizeof{\nstated}$.
Intuitively, $\pseqofof{\nstatec}{\nstated}$ consists of the poison types of $\passigns$, applied to the program counters of $\nstatec$ and $\nstated$ at every level.
The only exception is the lowest level in the speculating states.
They are always set to $\healthy$, because we know from \Cref{lem:spec-free-prod}
that those states are reachable speculation-free
and can never have poisoned registers or memory locations.
The notation $\nstatec \relrm \nstated$, $\nstatec \rrmp{\poison} \nstated$, and $(\nstatec\lseq\statec, \nstated\lseq\stated, \poison)$
are shorthand notations to avoid a lot of repeating $\pseqofof{\nstatec}{\nstated}$.
%
%
\begin{lemma}\label{lem:solutionapprox}
    Whenever $(\nstatec, \nstated, \pseqofof{\nstatec}{\nstated}) \sprodtranss{\leakp}{\directp}{\leak}{\direct} (\nstatee, \nstatef, \pseqp)$
    then $\pseqp \leq \pseqofof{\nstatee}{\nstatef}$.
\end{lemma}

\subsection{Fixing Register Allocation}

We can use our static analysis solution $\passigns$ to identify whether $\tra$ has weaknesses:
If $\passigns$ guarantees that leakages are never poisonous, then $\tra$ has no register allocation induced weaknesses.

\begin{definition}
    A register allocation $(\instinj, \regmap)$ between $\prog$ and $\tra$ is poison-typable,
    if there is a solution $\passigns$ to \labelcref{eqn:flow-fwd}
    which for every $(\pc, \pcp)$ with $\instinjof{\pc} = \pcp$
    satisfies the additional constraints
    \begin{align*}
        \poisonsof{\regb} &= \weakpoisoned \lor \poisonof{\regb} = \healthy &&\text{if} &
        \progof{\pc} &= \Iload{\rega}{\vara}{\regb}{\pcsuc}
            \lor \progof{\pc} = \Istore{\vara}{\regb}{\regc}{\pcsuc}\;,
            \\
        \poisonsof{\regb} & = \healthy
                        &&\text{if} &
        \progof{\pc} &= \Iif{\regb}{\pcsuc_{\btrue}}{\pcsuc_{\bfalse}}\;.
    \end{align*}
\end{definition}

Our fix is applicable to any register allocation $(\instinj, \regmap)$ between $\prog$ and $\tra$:
We check if $(\instinj, \regmap)$ is poison-typable.
For that, we solve the flow analysis (e.g. \parencite{kildallUnifiedApproachGlobal1973}) to obtain a static poison assignment $\passigns$ and check the additional constraints.
If it is not poison-typable, then an additional constraint is violated for some program point $(\pc, \pcp)$ and register $\regbp$ of $\prodra$.
We insert an $\IWsfence$ or $\IWslh{\regbp}$ instruction into $\tra$ at the end of the shuffle sequence right before $\pcp$ and obtain a new register allocation where that additional constraint is now satisfied.
We then repeat the process until we obtain a poison-typable register allocation.

\subsection{Poison-typable Register Allocation is \texorpdfstring{$\propsnip$}{SNiP}}
We now prove that making register allocation poison-typable already makes it $\propsnip$.
With our proof method from \Cref{thm:snippysound},
this reduces to crafting a snippy simulation.
Again, the crafted simulation is parametric, so that it works for all poison-typable register allocations.
\begin{theorem}
    If $(\instinj, \regmap)$ is a poison-typeable register allocation between $\prog$ and $\tra$,
    then there exists a snippy simulation $(\simrel, \dtffamily)$ between $\tra$ and $\prog$.
\end{theorem}
For the remainder of the section, fix a poison-typable register allocation $(\regmap, \instinj)$ between $\prog$ and $\tra$,
and the static poison assignment $\passigns$.
Further, let $\entry$ be the entry point for $\prog$ and $\entryp$ for~$\tra$.

\subsubsection*{Defining \texorpdfstring{$(\simrel, \dtffamily)$}{(≺,◃)}}
\begin{figure}
    \begin{tikzpicture}
        [
        intervalpath1/.style = {line width=0.8pt, teal,rounded corners, shorten <=4pt, shorten >=5pt},
        intervalpath2/.style = {line width=0.8pt, red,rounded corners, shorten <=4pt, shorten >=5pt},
        ntranscube/.style = {->,line width=1pt,dashed,dash pattern=on 3pt off 3pt}
        ]
       \tikzmath{\disthorizontal = 2; \distvertical = 0.7; }

        \node (c) {$\nstatec$};
        \node[right=\disthorizontal of c] (e) {$\nstatee$};
        \node[below=0.8 of c,anchor=south west,draw] (proglabel) {$\prog$};
        \node[below=\distvertical of e] (cp) {$\nstatee\rlap{$'$}$};

        \node[right=\disthorizontal of e] (d) {$\nstated$};
        \node[right=\disthorizontal of d] (h) {$\nstateh$};
        \node[right=\disthorizontal of h] (dots) {$\ldots\vphantom{X}$};
        \node[right=\disthorizontal of dots] (f) {$\nstatef$};
        \node[below=\distvertical of h] (hp) {$\nstateh\rlap{$'$}$};
        \node[below=0.8 of f,anchor=south east,draw] (proglabeltra) {$\tra$};

        \draw[intervalpath1] (c.center) -- (c.south east) -- (e.south west) -- (cp.north west) -- (cp.center);
        \draw[intervalpath2] (c.center) -- ([shift={(0,0.05)}]c.south east) -- ([shift={(0,0.05)}]e.south west) -- (e.center);
        \draw[intervalpath1] (d.center) -- (d.south east) -- (h.south west) -- (hp.north west) -- (hp.center);
        \draw[intervalpath1] (d.center) -- (d.south east) [rounded corners=10pt] -- ([shift={(-.37,0)}]dots.south west) [rounded corners] -- (hp.north east) -- (hp.center);
        \draw[intervalpath1] (d.center) -- (d.south east) [rounded corners=10pt] -- ([shift={(-0.01,0)}]dots.south) [rounded corners] -- ([shift={(0.1,-0.05)}]hp.east) -- ([shift={(0,-0.05)}]hp.center);
        \draw[intervalpath2] (d.center) -- ([shift={(0,0.05)}]d.south east) -- ([shift={(0,0.05)}]f.south west) -- (f.center);
        \draw[ntranscube] (c) -- node[above] {$\smash{\dtpair{\leakp}{\directp}}$} (e);
        \draw[ntranscube] (d) -- node[above] {$\smash{\dtpair{\leak}{\direct}}$} (h);
        \draw[ntranscube] (h) -- node[above] {$\smash{\dtpair{\leak_1}{\Dstep}}$} (dots);
        \draw[ntranscube] (dots) -- node[above] {$\smash{\dtpair{\leak_m}{\Dstep}}$} (f);

        \draw[ntranscube] (e) -- node[right] {$\smash{\dtpair{\Erb}{\Drb}}$} (cp);
        \draw[ntranscube] (h) -- node[right] {$\smash{\dtpair{\Erb}{\Drb}}$} (hp);
        \draw ([shift={(0,0.1)}]h.north) -- ([shift={(0,0.2)}]h.north) -- node[below=-1pt] {\tiny shuffle only} ([shift={(0,0.2)}]f.north) -- ([shift={(0,0.1)}]f.north);

        \draw[ntranscube] (dots) -- node[fill=white,below right] {$\dtpair{\Erb}{\Drb}$} (hp);
    \end{tikzpicture}
    \caption{\label{fig:phiinterval} The shape of $\instinj$-intervals. \textcolor{teal}{Teal} and \textcolor{red}{red} paths form separate intervals.}
\end{figure}
We define ${\simrel} \subseteq \nstatesof{\tra} \times \nstatesof{\prog}$ similar to the simulation for register allocation without speculative semantics:
Source state $\nstatec$ and target state $\nstated$ are instruction-matched and coincide in values up to relocation by $\regmap$.
The difference is that the states do not need to coincide on registers and memory locations poisoned by $\passigns$:
\begin{align*}
    \nstated \simrel \nstatec & \ifftxt \nstatec \rinstinj \nstated \land \nstatec \relrm \nstated \,.
\end{align*}
For the directive transformation $\dtfat{\nstatec}{\nstated}$,
we rely on $\prodra$'s transitions.
We say that a transition sequence $(\nstatec, \nstated, \poison) \sprodtranss{\tracek}{\dtracee}{\tracel}{\dtraced} (\nstatee, \nstatef, \poisonp)$ is a $\instinj$-interval
if $\nstatec \rinstinj \nstated$, $\nstatee \rinstinj \nstatef$, and no intermediary state pairs are instruction matched.
The intent is that $\instinj$-intervals are precisely the simulation intervals once we have proven that our defined relation is a simulation.
\Cref{fig:phiinterval} depicts the shape of $\instinj$-intervals:
The product first executes the instruction-matched instruction on both $\nstatec$ and $\nstated$.
Then, the target program can perform any number of shuffle steps,
until either the shuffle sequence is fully executed (the pair of \textcolor{red}{red} paths),
or a rollback happened before that (any pair of \textcolor{teal}{teal} paths).
In that case, $\prodra$ rolls back on $\nstatee$ as well.
We use the $\instinj$-intervals as directive transformations,
\[
    \dtfat{\nstatec}{\nstated} \quad\defeq\quad
    \setcond{(\dtracee, \dtraced)}
    {\text{$(\nstatec, \nstated, \poison) \sprodtranss{\tracek}{\dtracee}{\tracel}{\dtraced} (\nstatee, \nstatef, \poisonp)$ is a $\instinj$-interval}} \, .
\]
\begin{lemma}\label{lem:trasnippy}
    $(\simrel, \dtffamily)$ is a snippy simulation between $\tra$ and $\prog$.
\end{lemma}

\section{Related Work}
\label{section:related}

We already discussed the closely related work in the context of compiler correctness.
Here, we give a broader picture and elaborate on methods for proving non-interference for single programs.
Note the difference: when reasoning about compiler passes, we reason over all programs.
For a broader overview, we defer the reader to a recent survey \parencite{cauligiSoKPracticalFoundations2022}.

\subsubsection*{Speculation Sources}

Since its discovery in 2018,
Spectre attacks have been rediscovered in multiple variants.
The main difference between the variants lies in the hardware feature that is trained in order to trigger a misspeculation.
We call the respective feature the source of speculation.
The first version of the attack trains the Prediction History Table (\spht) of the processor, in order to inflict a mispredicted branching speculation \parencite{kocherSpectreAttacksExploiting2019}.
%
Other variants train the Branch Target Buffer (\sbtb) to mispredict indirect branching instructions \parencite{kocherSpectreAttacksExploiting2019},
the Return Stack Buffer~(\srsb) to mispredict return points \parencite{koruyehSpectreReturnsSpeculation2018},
both of which highjack the speculative control flow to execute leaking gadgets speculatively.
They can be mitigated in software with a \texttt{retpoline} gadget \cite{turnerRetpolineSoftwareConstruct2018}.
The Speculative Store Bypass (\sssb) mechanism, also called Store-To-Load Forwarding (\sstl), reads from memory even though pending stores have unresolved addresses \parencite{hornSpeculativeExecutionVariant2018},
and the Predictive Store Forwarding (\spsf) mechanism forwards pending stores to loads with unresolved address \parencite{guancialeInSpectreBreakingFixing2020}.
These mechanisms speculatively load values that either should have been overwritten in the meantime, or should never arrive in memory at the loaded address,
creating further potential for unwanted information-flow.
Memory speculation sources can be disabled in hardware with mediocre performance penalty.
The Meltdown attack introduces speculation through an out-of-order read from elevated-permission memory regions, racing against the MMU to detect the violation before the read memory can be leaked~\parencite{lippMeltdownReadingKernel2018}.
While the speculation sources are varying, all attack variants leak the secrets through the side-channels covered by the constant-time leakage model \parencite{guarnieriHardwareSoftwareContractsSecure2021}.
In this paper, we consider only the \spht{} speculation source.
We believe that our notion of simulation also holds for other speculation sources, but we are less sure about what compiler passes satisfy $\propsnip$ when they are considered.

\subsubsection*{Properties}

The constant time programming guideline requires no influence of sensitive data towards the leakage observable by the attacker.
Formally, \emph{non-interference} \parencite{goguenSecurityPoliciesSecurity1982} (\Cref{def:sni}) expresses this property.
%
%
%
Non-interference is a hyper-property \parencite{clarksonHyperproperties2010}:
It requires to reason about two executions of the program.
Hyper-properties tend to be harder to verify than single trace properties due to synchronization issues with the traces compared.
Non-interference for side-channel leakage evades these issues:
The constant time leakage model exposes the program counter to the attacker.
Thus, when two traces from attacker-indistinguishable initial states do not coincide in their control flow, the program can immediately be rejected as insecure.
Such a high degree of synchronization allows for a sound approximation we call taint safety \parencite{sabelfeldLanguagebasedInformationflowSecurity2003,myersJFlowPracticalMostlystatic1999}.

\begin{tablecites}[table:tools]{l||X[r]|X[r]|X[r]|X[r]|X[r]|X[r]}{
        Tools that check a program against speculative non-interference:
        Pitchfork \cite{cauligiConstanttimeFoundationsNew2020},
        Spectector \cite{guarnieriSpectectorPrincipledDetection2020},
        RelSE \cite{danielHuntingHaunterEfficient2021},
        Blade~\cite{vassenaAutomaticallyEliminatingSpeculative2021},
        Jasmin~SCT \cite{bartheHighAssuranceCryptographySpectre2021},
        Typing V1 \cite{shivakumarTypingHighSpeedCryptography2023}.
    }{}
                    & Pitchfork
                    & Spectector
                    & RelSE
                    & Blade
                    & Jasmin~SCT
                    & Typing V1
                    \\ \hline
Source              & PHT, SSB      & PHT           & PHT, SSB  & PHT       & PHT           & PHT                   \\
Property            & TS            & NI            & NI        & TS        & TS            & TS                    \\
Method              & EX            & SE            & SE        & SA        & SA            & SA                    \\
Speculation         & SW            & SW            & SW / SB   & US        & US            & US                    \\
Directives          & Y             & Y (Oracle)    & N         & Y         & Y             & Y                     \\
Memory Safety       & U             & U             & U         & U         & SS            & S                     \\
Bug / Proof         & Bug           & Bug           & Bug       & Proof     & Proof         & Proof
\end{tablecites}
\subsubsection*{Tools}

Existing tools check single implementations for non-interference.
Our work on compiler transformations complements this line of work.
It is now possible to check source programs for non-interference,
and rely on the compiler that guarantees to preserve non-interference to the executed program for \emph{any source program}.
The tools deal with two main challenges:
Non-determinism from speculation and the two executions required for non-interference.
%
%
\Cref{table:tools} presents a list of tools and classifies their approaches.
The
\begin{enumerate*}[font=\bfseries]
    \item[Source] lists the considered speculation sources.
    \item[Property] is the formal property checked:
        \textsf{NI} is non-interference under speculative semantics.
        \textsf{TS} is taint safety.
    \item[Method] is either \textsf{SA} for static analysis such as flow- or type-systems,
        \textsf{SE} for symbolic execution,
        or \textsf{EX} for state space exploration.
    \item[Speculation] lists how the tool models and copes with speculation based non-determinism.
        \textsf{SW} means that the semantics have a bounded speculation window, i.e.\ a bound to the number of speculatively executed instructions.
        \textsf{SB} means that the semantics have a bounded store buffer that limits the number of speculatively executed store instructions.
        \textsf{US} means unbounded speculation, i.e.\ the semantics speculate arbitrarily long.
    \item[Mem Safety] describes the memory model and memory requirements on the source program.
        \textsf{U} stands for unstructured memory,
        \textsf{S} for structured memory and memory safety under speculation-free semantics,
        and \textsf{SS} is memory safety even under speculative semantics.
\end{enumerate*}

\textsf{Jasmin SCT} (\Cref{table:tools}) comes along with the \textsf{Jasmin} compiler (\Cref{table:compilerworks}),
that is proven to preserve $\propni$ under speculation-free semantics.
In \parencite{bartheHighAssuranceCryptographySpectre2021}, the authors suggest an additional requirement to cope with speculative semantics:
Source programs need to be memory safe even when executed under speculative semantics (\textsf{SS}).
However, there is no proof that the compiler preserves non-interference under this requirement.
\textsf{SS} is also a performance breaking requirement:
They report that \textsf{SS}-implementations are \textasciitilde{}20\% slower than the insecure reference implementations;
in contrast to \textasciitilde{}1\% overhead reported for recently protected implementations without \textsf{SS}.

\subsubsection*{Similarities and differences to a simultaneously developed proof technique}

Another proof method for preserving speculative side-channel security through compilers has been developed independently of our work \cite{arranzolmosPreservationSpeculativeConstanttime2025}.%
\footnote{To be appearing in the same conference.}
There are subtle differences in the definition of semantics, security property, and simulation:
First, their semantics comes without rollbacks.
The idea is that a difference in leakage can always be obtained by an execution that contains no rollback.
This result has been obtained previously \cite{bartheHighAssuranceCryptographySpectre2021}.
Second, their security property, \textsf{SCT}, is the same as $\propsni$ from this work.
However, the authors employ a slightly modified semantics which allows them to phrase the property differently.
Finally, our simulation constraint is a constant-time cube (\Cref{fig:snippysim}).
The constraint formulated by \citet{arranzolmosPreservationSpeculativeConstanttime2025} instead requires the existence of two functions:
One back-translates directives similar to how our simulation finds a sequence of source directives to replay the target directives.
The other forward-translates leakage from source leakage to target leakage.
The existence of both functions already guarantees that our constant-time cube is satisfied.
%
%
The bigger difference is in the application of our proof methods.
Our work spots a weakness in register allocation, and we develop a static analysis to protect against security-threatening spills.
\citet{arranzolmosPreservationSpeculativeConstanttime2025} target the Jasmin compiler with their work.
Jasmin's source language already requires the programmer to tag variables as register or stack variable.
Therefore, the Jasmin compiler performs no spilling, avoiding the weakness by demanding the programmer to appropriately choose register-allocated variables.
Their work instead focuses on proving nine other passes secure by extending preservation proofs from leakage semantics to speculative semantics.

\section{Conclusion}
\label{section:conclusion}

We have developed a method for proving that compiler transformations preserve non-interference from source to target programs \emph{under speculative semantics}.
When experimenting with our method, we found that it worked well on simple transformations like dead code elimination,
but we had trouble applying it to register allocation.
As it turned out, the fault was not on our side but register allocation is actually insecure.
Our method led us to discover a new vulnerability introduced by register allocation.
We have confirmed the existence of this vulnerability in the mainstream compiler \texttt{LLVM} on code from \texttt{libsodium}, a modern cryptographic library.
Interestingly, our proof method also guided us towards a fix:
We have presented a new static analysis to identify weaknesses introduced by register allocation, and an automated procedure to fix them.
With these additions, we have been able to prove that register allocation preserves non-interference.

As future work, we would like to integrate our proof method with certified compilers, investigate transformations that we left out so far, and consider further sources of speculation in the semantics.

\printbibliography

\appendix

\section{Missing proofs from Section 2}

\begin{proof}[Proof of \Cref{lem:programpoint-by-leakage}]
    Let $\nstatec_1 = \nstatecp_1\lseq(\pc, \rasgn_1, \masgn_1)$,
    $\nstatec_2 = \nstatecp_2\lseq(\pc, \rasgn_2, \masgn_2)$,
    We do the proof by induction on the first transition $\nstatec_1 \ntrans{\leak}{\direct} \nstated_1$ and $\nstatec_2 \ntrans{\leak}{\direct} \nstated_1$.
    We do case distinction.

    \begin{pcases}
        \pcasep{$\direct = \Drb = \leak$}, then $\sizeof{\nstatec_1} = \sizeof{\nstatec_2} > 1$.
        Further, $\nstated_1 = \nstatecp_1$ and $\nstated_2 = \nstatecp_2$.
        Thus, $\nstated_1 \samepoint \nstated_2$ follows immediately from $\nstatec_1 \samepoint \nstatec_2$

        \pcasep{$\direct = \Dspec$ and $\leak = \Eif{\bvalue}$}
        Then, $\nstatec_1 \ntrans{\Eif{\bvalue}}{\Dspec} \nstatec_1\lseq(\pcsuc_{\lnot \bvalue}, \rasgn_1, \masgn_1)$
        and $\nstatec_2 \ntrans{\Eif{\bvalue}}{\Dspec} \nstatec_2\lseq(\pcsuc_{\lnot \bvalue}, \rasgn_2, \masgn_2)$.
        Clearly, $\nstated_2 = \nstatec_2 \lseq (\pcsuc_{\lnot \bvalue}, \rasgn_2, \masgn_2)
        \samepoint \nstatec_1 \lseq (\pcsuc_{\lnot \bvalue}, \rasgn_1, \masgn_1)$.
        The case of $\direct = \Dif$ and $\leak = \Eif{\bvalue}$ is similar.

        \pcasep{All other cases} Because $\instof{\pc}$ has only a single successor $\pcsuc$, and only modifies the executing state,
        it is clear from $\nstatecp_1 \samepoint \nstatecp_2$ and both executing states being in $\pcsuc$, that $\nstated_1 \samepoint \nstated_2$. \qedhere
    \end{pcases}
\end{proof}

\section{Missing Proofs from Section 4}

\subsection*{\Cref{lem:behaviorininterval}}
\begin{proof}
    The $\supseteq$ direction for terminating behavior is immediate: ${\simtranst{}{}} \subseteq {\ntranss{}{}}$.
    For diverging behavior, consider any $\dbeh{\dtpair{\itrace}{\idtrace}}$ with $\nstated \simtransti{\itrace}{\idtrace}$.
    By definition, there is $\nstated \simtranst{\tracel}{\dtraced} \nstatef \simtransti{\itracep}{\idtracep}$.
    Notably, $\simtranst{\tracel}{\dtraced}$ cannot do stuttering transitions,
    thus there are transitions $\nstated \ntrans{\leak_0}{\direct_0} \ldots \ntrans{\leak_n}{\direct_n} \nstatef$
    and $\leak_0\lseq \ldots \lseq \leak_n = \tracel$, $\direct_0\lseq\ldots\lseq\direct_n = \dtraced$.
    Further, if there is $\nstated \simtranst{\tracel}{\dtraced} \nstatef$ then also $(\nstatec, \nstated) \synctrans{\tracek}{\dtracee}{\tracel}{\dtraced} (\nstatee, \nstatef)$.
    Thus, $\nstatec \simtranss{\tracek}{\dtracee} \nstatee$ and $\nstatef \simrel \nstatee$.
    By coinduction, $\nstatef \ntransi{\itracep}{\idtracep}$, which completes this direction.

    For $\subseteq$, be given the proof that justifies $\nstated \simrel \nstatec$.
    Show that when $\nstated \ntransi{\itrace}{\idtrace}$ ($\nstated \ntranss{\tracel}{\dtraced} \nstateh$, $\nstateh$ final),
    then (either $\nstated = \nstateh$ are final, or)
    there is $\nstated \simtranst{\tracem}{\dtracef} \nstatef \simrel \nstatee$,
    with $\nstatef \simtransti{\itracep}{\idtracep}$ ($\nstatef \simtransts{\tracek}{\dtracee} \nstateh$)
    so that $\itrace = \tracem\lseq\itracep$ and $\idtrace = \dtracef\lseq\idtracep$ ($\tracel = \tracem\lseq\tracek$ and $\dtraced = \dtracef\lseq\dtracee$).

    To do so, we first show the following:
    For every $\nstated \simrel \nstatec$ with $\nstated \ntranss{\tracel}{\dtraced} \nstateh$ (note that $\nstateh$ need not be final),
    either $\nstated \ntranss{\tracem}{\dtracef} \nstatef \ntranss{\tracek}{\dtracee} \nstateh$
    so that $\nstated \simtranst{\tracem}{\dtracef} \nstatef$ or there is a proof node
    $\proofnode{\sderivetgtsim{\eitherguarded{\simrel}}{\dtfat{\nstatec}{\nstated}}{\nstateh}{\nstatec}{\dtraced}}$.
    We do so by induction on the structure of $\nstated \ntranss{\tracel}{\dtraced} \nstateh$.
    For the base case ($\nstated \ntranss{\varepsilon}{\varepsilon} \nstated$), we already know that the proof has a node
    $\proofnode{\sderivetgtsim{\guarded{\simrel}}{\dtfat{\nstatec}{\nstated}}{\nstated}{\nstatec}{\dtraced}}$ for justification of $\nstated \simrel \nstatec$.
    For the inductive case, consider $\nstated \ntranss{\tracel}{\dtraced} \nstatehp \ntrans{\leak}{\direct} \nstateh$.
    Then, by induction, either
    $\nstated \simtranst{\tracem}{\dtracef} \nstatef \ntranss{\tracek}{\dtracee} \nstatehp$ or there is a proof node
    $\proofnode{\sderivetgtsim{\eitherguarded{\simrel}}{\dtfat{\nstatec}{\nstated}}{\nstatehp}{\nstatec}{\dtraced}}$.
    In the first case, we are done.
    In the second case, we do case distinction by the rule that derives
    $\proofnode{\sderivetgtsim{\eitherguarded{\simrel}}{\dtfat{\nstatec}{\nstated}}{\nstatehp}{\nstatec}{\dtraced}}$.
    In case of \cref{rule:direct-tf}, we get $\nstated \simtranst{\tracel}{\dtraced} \nstatehp$ (because the proof also derives a source sequence).
    Otherwise, the case is \cref{rule:tgt}.
    In that case, 
    $\proofnode{\sderivetgtsim{\eitherguarded{\simrel}}{\dtfat{\nstatec}{\nstated}}{\nstateh}{\nstatec}{\dtraced\lseq\direct}}
    \proofedges{\direct}
    \proofnode{\sderivetgtsim{\eitherguarded{\simrel}}{\dtfat{\nstatec}{\nstated}}{\nstateh'}{\nstatec}{\dtraced}}$ as desired.

    Then, we utilize this to first prove the slightly different statement:
    For every $\nstated \simrel \nstatec$
    with $\nstated \ntransi{\itrace}{\idtrace}$ ($\nstated \ntranss{\tracel}{\dtraced} \nstateh$, $\nstateh$ final),
    (either $\nstated = \nstateh$ is final, or)
    there is $\nstated \simtranst{\tracem}{\dtracef} \nstatef \simrel \nstatee$,
    \underline{with $\nstatef \ntransi{\itracep}{\idtracep}$ ($\nstatef \ntranss{\tracek}{\dtracee} \nstateh$)}
    so that $\itrace = \tracem\lseq\itracep$ and $\idtrace = \dtracef\lseq\idtracep$ ($\tracel = \tracem\lseq\tracek$ and $\dtraced = \dtracef\lseq\dtracee$).
    Towards contradiction, assume it is not the case. 
    Then, the previous fact yields us an infinite chain of nodes in the proof tree,
    a contradiction to the well-foundedness of the proof:
    First, due to $\nstated \ntranss{\varepsilon}{\varepsilon} \nstated$,
    $\proofnode{\sderivetgtsim{\guarded{\simrel}}{\dtfat{\nstatec}{\nstated}}{\nstated}{\nstatec}{\varepsilon}}$
    is part of the proof.
    And when by induction the transitions
    $\nstated \ntranss{\tracem}{\dtracef} \nstatefp \ntrans{\leak}{\direct} \nstatef \ntransi{\itracep}{\idtracep}$,
    $\dtracef = \direct_0\lseq \ldots \lseq \direct_n$ imply that
    the proof contains
    $\proofnode{\sderivetgtsim{\guarded{\simrel}}{\dtfat{\nstatec}{\nstated}}{\nstated}{\nstatec}{\varepsilon}}
    \rproofedges{\direct_0} \ldots \rproofedges{\direct_n}
    \proofnode{\sderivetgtsim{\eitherguarded{\simrel}}{\dtfat{\nstatec}{\nstated}}{\nstatefp}{\nstatec}{\dtracef}}$,
    then it cannot be proven by \cref{rule:direct-tf} as that would create
    $\nstated \simtranst{\tracem}{\dtracef} \nstatef$.
    Thus, it must be proven via \cref{rule:tgt} which further requires 
    $\proofnode{\sderivetgtsim{\eitherguarded{\simrel}}{\dtfat{\nstatec}{\nstated}}{\nstatefp}{\nstatec}{\dtracef}}
    \rproofedges{\direct}
    \proofnode{\sderivetgtsim{\unguarded{\simrel}}{\dtfat{\nstatec}{\nstated}}{\nstatef}{\nstatec}{\dtracef\lseq\direct}}$.
    
    Finally, we derive the result by coinduction, which yields $\nstatef \simtransti{\itracep}{\idtracep}$.
    With $\nstated \simtranst{\tracem}{\dtracef} \nstatef$, we compose $\nstated \simtransti{\tracem \lseq \itracep}{\dtracef\lseq\idtracep}$.
\end{proof}

\section{Missing proofs from section 5 and Transfer for Liveness Analysis}

We define the transfer functions for Liveness analysis on instructions.
Write $\transferat{\inst}$ for $\transferat{\pc}$ with $\progof{\pc} = \inst$.
\begin{align*}
    &
    \begin{aligned}
        \transferof{\Iload{\rega}{\vara}{x}{\pcsuc}}{\latelem} & =
        \begin{cases}
            \latelem & \rega \notin \latelem \\
            (\latelem \setminus \set{\rega}) \cup \set{(\vara,\adr)} & x = \adr \\
            (\latelem \setminus \set{\rega}) \cup \mems) & x = \regb
        \end{cases} \\
        \transferof{\Istore{\vara}{\rega}{x}{\pcsuc}}{\latelem} & =
        \begin{cases}
            \latelem & x = \adr \in \vsizeof{\vara}, (\vara, \adr) \notin \latelem \\
            (\latelem \setminus \set{(\vara,\adr)}) & x = \adr \in \vsizeof{\vara}, (\vara, \adr) \in \nmem \\
            \latelem \cup \set{\regb} & x = \regb
        \end{cases}
    \end{aligned}
    &
    \begin{aligned}
        \transferof{\Inop{\pcsuc}}{\latelem} & = \latelem \\
        \transferof{\Iasgn{\rega}{\regb \op \regc}{\pcsuc}}{\latelem} & =
        \begin{cases}
            (\latelem \setminus \set{\rega}) \cup \set{\regb, \regc} & \rega \in \latelem \\
            \latelem & \rega \notin \nregis
        \end{cases} \\
        \transferof{\Iif{\rega}{\pcsuc_{\btrue}}{\pcsuc_{\bfalse}}}{\latelem} &
        = \latelem \cup \set{\rega} \\
        \transferof{\Iexit}{\latelem} &= \flowinit = \mems
    \end{aligned}
\end{align*}

\subsection*{Remaining cases for \Cref{thm:dcsnippy}}

\begin{pcases}
    \pcasep{$\langle\Iasgn{\rega}{\regb \op \regc}{\pcsuc},\Inop{\pcsuc}\rangle$}
    For the first part, we have $\stated_1 \trans{\Enone}{\Dstep} \statef_1$ and need to show $\statec_1 \trans{\Enone}{\Dstep} \statee_1$ and $\statef_1 \simrel \statee_1$.
    The former is immediate from semantics.
    For the latter,
    we know that $\rega \notin \flowof{\pc}$ because the instruction was replaced by $\IWnop$.
    The transfer of $\flowof{\pc}$ along $\inst$ is then $\transferof{\pc}{(\flowof{\pc})} = \flowof{\pc}$
    because an assignment to a non-live register does not make any registers live.
    Together with \labelcref{eqn:flow-bwd}, we get $\transferof{\pc}{(\flowof{\pc})} = \flowof{\pc} \supseteq \transferof{\pcsuc}{(\flowof{\pcsuc})}$.
    Because registers and memory of $\stated_1$ and $\statef_1$ are equal and $\statec_1$ and $\statee_1$ only differ on $\rega$,
    $\stated_1 \simrel \statec_1$ implies $\statef_1 \simrel \statee_1$ because $\rega \notin \flowof{\pc} \supseteq \transferof{\pcsuc}{(\flowof{\pcsuc})}$.
    For the second part,
    further assume $\statec_2 \trans{\Enone}{\Dstep} \statee_2$.
    We need to show that $\stated_2 \trans{\Enone}{\Dstep} \statef_2$ and $\statef_2 \simrel \statee_2$.
    The former is again immediate from semantics.
    For the latter, we have the same arguments as for the first part:
    Memory and register contents of $\stated_2$ and $\statef_2$ are equal and $\statec_2$ and $\statee_2$ only differ in $\rega$.

    \pcasep{$\Istore{\vara}{\regb}{\regc}{\pcsuc},\Istore{\vara}{\regb}{\regc}{\pcsuc}$}
    There are two subcases - either safe or unsafe store.

    For safe store we get $\stated_1 \trans{\Estore{\adr}}{\Dstep} \statef_1$
    and since $\regb \in \transferof{\pc}{(\flowof{\pc})}$, 
    $\rasgniof{\stated_1}{\regb} = \rasgniof{\statec_1}{\regb}$,
    so $\statec_1 \trans{\Estore{\adr}}{\Dstep} \statee_1$.
    Also, $\transferof{\pc}{(\flowof{\pc})} \cup \set{(\vara,\adr)} \geq \flowof{\pc} \geq \transferof{\pcsuc}{(\flowof{\pcsuc})}$
    and $\masgni{\stated_1} = \masgni{\statef_1}$ and $\masgni{\statec_1} = \masgni{\statee_1}$ on all slots except $(\vara,\adr)$.
    But due to $\regc \in \transferof{\pc}{(\flowof{\pc})}$, we also have $\masgniof{\stated_1}{\vara}{\adr} = \masgniof{\statef_1}{\vara}{\adr}$,
    so $\statef_1 \simrel \statee_1$.
    Further $\stated_2 \trans{\Estore{\adr}}{\Dstep} \statef_2$ and $\Dstep \dtfat{\statec_2}{\stated_2} \Dstep$ are a given.

    For unsafe store we have an analogue proof, except both directives are $\Dstore{\varb}{\adrp}$.

    \pcasep{$\Istore{\vara}{\regb}{\regc}{\pcsuc},\Inop{\pcsuc}$}
    Analogue to the previous case, except $\leakp = \Enone$ and the argument for $\masgniof{\stated_1}{\vara}{\adr} = \masgniof{\statef_1}{\vara}{\adr}$ does not hold anymore.
    However, since $(\vara,\adr) \notin \flowof{\pc} \geq \transferof{\pcsuc}{(\flowof{\pcsuc})}$, we don't need it for $\statef_1 \simrel \statee_1$.

    \pcasep{$\Iload{\rega}{\vara}{\regb}{\pcsuc},\Iload{\rega}{\vara}{\regb}{\pcsuc}$ and $\Iload{\rega}{\vara}{\regb}{\pcsuc},\Inop{\pcsuc}$}
    The argumentation is analogue to the previous two cases.

    \pcasep{$\Inop{\pcsuc},\Inop{\pcsuc}$}
    We get $\stated_1 \trans{\Enone}{\Dstep} \statef_1$ and $\statec_1 \trans{\Enone}{\Dstep} \statee_1$.
    Also, $\transferof{\pc}{(\flowof{\pc})} = \flowof{\pc} \geq \transferof{\pcsuc}{(\flowof{\pcsuc})}$
    and $\rasgni{\stated_1} = \rasgni{\statef_1}$ and $\rasgni{\statec_1} = \rasgni{\statee_1}$,
    so $\statef_1 \simrel \statee_1$.
    Further $\stated_2 \trans{\Enone}{\Dstep} \statef_2$ and $\Dstep \dtfat{\statec_2}{\stated_2} \Dstep$ are a given.

    \pcasep{$\Iasgn{\rega}{\regb \op \regc}{\pcsuc},\Iasgn{\rega}{\regb \op \regc}{\pcsuc}$}
    We get $\stated_1 \trans{\Enone}{\Dstep} \statef_1$ and $\statec_1 \trans{\Enone}{\Dstep} \statee_1$.
    Also, $\transferof{\pc}{(\flowof{\pc})} \cup \set{\rega} \geq \flowof{\pc} \geq \transferof{\pcsuc}{(\flowof{\pcsuc})}$
    and $\rasgni{\stated_1} = \rasgni{\statef_1}$ and $\rasgni{\statec_1} = \rasgni{\statee_1}$ on all registers except $\rega$.
    But due to $\regb, \regc \in \transferof{\pc}{(\flowof{\pc})}$, we also have $\rasgniof{\stated_1}{\rega} = \rasgniof{\statef_1}{\rega}$,
    so $\statef_1 \simrel \statee_1$.
    Further $\stated_2 \trans{\Enone}{\Dstep} \statef_2$ and $\Dstep \dtfat{\statec_2}{\stated_2} \Dstep$ are a given.

    \pcasep{$\Iexit$} No transition can be made by $\stated_1$, the case does not exist.

    \pcasep{$\Iif{\rega}{\pcsuc_{\btrue}}{\pcsuc_{\bfalse}},\Iif{\rega}{\pcsuc_{\btrue}}{\pcsuc_{\bfalse}}$}
    We get $\nstated_1\lseq\stated_1 \ntrans{\Eif{\lnot\bvalue}}{\Dspec{\lnot\bvalue}} \nstated_1\lseq\statef_{1,\bvalue}\lseq\statef_{1,\lnot\bvalue} = \nstatef_1$ where $\bvalue = (\rasgniof{\stated_1}{\rega} \sameas 0)$
    and by $\rega \in \transferof{\pc}{(\flowof{\pc})}$,
    $\rasgniof{\stated_1}{\rega} = \rasgniof{\statec_1}{\rega}$ so $\nstatec_1\lseq\statec_1 \ntrans{\Eif{\bvalue}}{\Dspec{\lnot\bvalue}} \nstatec_1\lseq\statee_{1,\bvalue}\lseq\statee_{1,\lnot\bvalue} = \nstatee_1$.
    Further, $\transferof{\pc}{(\flowof{\pc})} \geq \flowof{\pc} \geq \transferof{\pcsuc_{\bvalue}}{(\flowof{\pcsuc_{\bvalue}})}$.
    Thus, $\statef_{1,\bvalue} \simrel \statee_{1,\bvalue}$ and $\statef_{1,\lnot\bvalue} \simrel \statee_{1,\lnot\bvalue}$.
    Further, if $\nstatec_2 \ntrans{\Eif{\lnot\bvalue}}{\Dspec{\lnot\bvalue}} \nstatee_2$, by same arguments, we have $\nstated_2 \ntrans{\Eif{\lnot\bvalue}}{\Dspec{\lnot\bvalue}} \nstatef_2$
    and $\Dspec{\lnot\bvalue} \dtfat{\statec_2}{\stated_2} \Dspec{\lnot\bvalue}$.
\end{pcases}

\section{Additional Material and Missing proofs for Section 6}

\begin{ruleframe}[label=rules:shuffle]{Shuffle Semantics}
    \definerule{move}
    {\progof{\pc} = \Imove{\rega}{\regb}{\pcsuc}}
    {(\pc, \rasgn, \masgn) \trans{\Enone}{\Dstep} (\pcsuc, \subst{\rasgn}{\rega}{\rasgnof{\regb}}, \masgn)}
    {\label{rule:move}}

    \definerule{fill}
    {\progof{\pc} = \Ifill{\rega}{\stackl}{\pcsuc}}
    {(\pc, \rasgn, \masgn) \trans{\Eload{\stackl}}{\Dstep} (\pcsuc, \subst{\rasgn}{\rega}{\masgnof{\varstack}{\stackl}}, \masgn)}
    {\label{rule:fill}}

    \definerule{spill}
    {\progof{\pc} = \Ispill{\stackl}{\rega}{\pcsuc}}
    {(\pc, \rasgn, \masgn) \trans{\Estore{\stackl}}{\Dstep} (\pcsuc, \rasgn, \subst{\masgn}{(\varstack, \stackl)}{\rasgnof\rega})}
    {\label{rule:spill}}
\end{ruleframe}

\subsection*{Missing rules for \texorpdfstring{$\prodra$}{P :: [P]ra}}

\begin{ruleframe}[label=rules:pbasics]{Basic}
    \definerule{poison-nop}
    {
        \instof{\statec} = \Inop{\pcsuc}
        \\
        \statec \trans{\Enone}{\Dstep} \statee
        \\
        \stated \trans{\Enone}{\Dstep} \statef
    }
    {
        (\statec, \stated, \poison)
        \prodtrans{\Enone}{\Dstep}{\Enone}{\Dstep}
        (\statee, \statef, \poison)
    }
    {\label{rule:pnop}}

    \definerule{healthy-branch}
    {
        \instof{\statec} = \Iif{\regb}{\pcsuc_{\btrue}}{\pcsuc_{\bfalse}}
        \\
        \poisonof{\regb} = \healthy
        \\
        \statec \trans{\Eif{\bvalue}}{\Dif} \statee
        \\
        \stated \trans{\Eif{\bvalue}}{\Dif} \statef
    }
    {
        (\statec, \stated, \poison)
        \sprodtrans{\Eif{\bvalue}}{\Dif}{\Eif{\bvalue}}{\Dif}
        (\statee, \statef, \poison)
    }
    {\label{rule:hbranch}}

    \definerule{poison-asgn}
    {
        \instof{\statec} = \Iasgn{\rega}{\regb \op \regc}{\pcsuc}
        \\
        \statec \trans{\Enone}{\Dstep} \statee
        \\
        \stated \trans{\Enone}{\Dstep} \statef
        \\
        \pval = \inlineife{(\pregof{\regb} = \pregof{\regc} = \healthy)}{\healthy}{\poisoned}
    }
    {
        (\statec, \stated, \poison)
        \prodtrans{\Enone}{\Dstep}{\Enone}{\Dstep}
        (\statee, \statef, \subst{\poison}{\rega}{\pval})
    }
    {\label{rule:pasgn}}
\end{ruleframe}

\Cref{rule:pasgn} propagates poison values for assignments.
Notice that weakly poisoned values become poisoned because the operator's result value is most likely a non-zero value.
Thus, only if both arguments are healthy, we propagate $\healthy$ to the assigned register and $\poisoned$ otherwise.
\labelcref{rule:hbranch} forbids weakly poisoned branching conditions even though they are safe to leak.
This is to avoid that $\prog$ and $\tra$ arrive at different program points (up to $\instinj$).
Branching with weakly poisoned registers could be supported by letting one state speculate while the other does regular branching to keep them at the same program point.
However, this would violate the condition $\sizeof{\nstatec} = \sizeof{\nstated}$ and leads to a less intuitive product definition.

\subsection*{\Cref{lem:prod-well-def}}
\begin{proof}
    We do a case distinction on the transition rule for $(\nstatec, \nstated, \pseq) \sprodtrans{\leakp}{\directp}{\leak}{\direct} (\nstatee, \nstatef, \pseqp)$.
    To that end, let $\pseq = \pseq'\lseq\poison$ and $\pseqp = \pseqp'\lseq \poisonp$,
    $\nstatee = \nstateep\lseq\statee$ and $\nstatef = \nstatefp\lseq\statef$,
    and $\nstatec = \nstatecp\lseq\statec$ and $\nstated = \nstatedp\lseq\stated$.
    Let further be $\inst = \progof{\statec}$ and $\instp = \traof{\stated}$.
    For all cases except \labelcref{rule:hspec,rule:prb} it suffices to show that $\statee \nrrmp{\poisonp} \statef$.
    All those rules only modify the top states and poison values, so $\pseqp' = \pseqp'$.
    The definition of $\nrrmp{\pseqp}$ is then satisfied from $\nstatecp \nrrmp{\pseqp'} \nstatedp$.

    We first do the separate two cases who change the size of the speculating states.
    \begin{pcases}
        \pcasep{\labelcref{rule:hspec}}
        Then, $\inst = \IWif{\regb}$, $\poisonof{\regb} = \healthy$, and $\pseqp = \pseq'\lseq\poison\lseq\poison$.
        Further, $\statee$ and $\statec$ as well as $\stated$ and $\statef$ are fully equal except for the program counter.
        Because $\statec \nrrmp{\poison} \stated$ and $\nrrmp{\poison}$ is indifferent to the program counter, $\statee \nrrmp{\poison} \statef$.

        \pcasep{\labelcref{rule:prb}}
        Then, $\nstatee = \nstatecp \nrrmp{\pseq'} \nstatedp = \nstatef$.
    \end{pcases}

    \noindent Now to the other cases.
    \begin{pcases}
        \pcasep{\labelcref{rule:pfill,,rule:pspill,,rule:pmove,,rule:pshuf-sfence,,rule:pnop,,rule:hbranch,,rule:psfence,,rule:pslh}}
        $\poisonp = \poison$.
        Further, $\statee$ and $\statec$ as well as $\stated$ and $\statef$ are fully equal except for the program counter.
        Because $\statec \nrrmp{\poison} \stated$ and $\nrrmp{\poison}$ is indifferent to the program counter, $\statee \nrrmp{\poison} \statef$.

        \pcasep{\labelcref{rule:pintro-load,,rule:pload-ss,,rule:pload-us}}
        Then, $\inst = \IWload{\rega}{\vara}{\regb}$ and
        $\poisonp = \subst{\poison}{\rega}{\poison}$.
        Further, $\statee$ and $\statec$ as well as $\stated$ and $\statef$ are equal except for the program counter and $\rega$ as well as $\regap = \regmapofof{\statef}{\rega}$.
        But $\rega$ is poisoned in $\poisonp$, so $\nrrmp{\poisonp}$ is indifferent to its value, $\statee \nrrmp{\poisonp} \statef$.

        \pcasep{\labelcref{rule:pload-hsafe}}
        Then, $\inst = \IWload{\rega}{\vara}{\regb}$ and $\adr = \rasgniof{\statec}{\regb}$ and
        $\poisonp = \subst{\poison}{\rega}{\pmemof{\vara}{\adr}} \leq \subst{\poison}{\rega}{\poison}$.
        Further, $\statee$ and $\statec$ as well as $\stated$ and $\statef$ are equal except for the program counter and $\rega$ as well as $\regap = \regmapofof{\statef}{\rega}$.
        The value of $\rega$ is $\masgniof{\statec}{\vara}{\adr}$ and of $\regap$ is $\masgniof{\stated}{\vara}{\adr}$.
        But $\rega$ has the poison value of $\poison$ for $(\vara,\adr)$ in $\poisonp$, so $\statee \nrrmp{\poisonp} \statef$ follows from $\statec \nrrmp{\poison} \stated$.

        \pcasep{\labelcref{rule:pload-hunsafe}} Analogue to the previous case, but $(\vara, \adr)$ swapped to $(\varb, \adrp)$.

        \pcasep{\labelcref{rule:pintro-store}}
        Then, $\inst = \IWstore{\vara}{\regb}{\regc}$ and $\directp = \Dstore{\vara}{\adrppp}$ and $\direct = \Dstore{\varstack}{\adrp}$ and $\regmapofof{\statef}{\regd} = (\varstack, \adrp)$.
        We get $\poisonp = \subst{\poison}{\regd, (\varstack, \adrp)}{\poisoned}$.
        Further, values of $\statee$ and $\statec$ are equal except for $\regd$ and
        values of $\stated$ and $\statef$ are equal except $\regmapofof{\statef}{\regd} = (\varstack, \adrp)$.
        Again, those are poisoned in $\poisonp$,  so $\statee \nrrmp{\poisonp} \statef$ follows from $\statec \nrrmp{\poison} \stated$.

        \pcasep{\labelcref{rule:pstore-ss}}
        Then, $\inst = \IWstore{\vara}{\regb}{\regc}$ and $\rasgniof{\statec}{\regb} = \adrpp$ and $\rasgniof{\stated}{(\regmapofof{\stated}{\regb}} = 0$ due to weak poisonedness.
        We get $\poisonp = \subst{\poison}{(\vara, \adrpp), (\vara, 0)}{\poisoned}$.
        Further, values of $\statee$ and $\statec$ are equal except for $(\vara, \adrpp)$ and
        values of $\stated$ and $\statef$ are equal except $(\vara, 0)$.
        Again, those are poisoned in $\poisonp$,  so $\statee \nrrmp{\poisonp} \statef$ follows from $\statec \nrrmp{\poison} \stated$.

        \pcasep{\labelcref{rule:pstore-us}} Analogue to the previous case, but $\adrpp$ swapped to $\adrppp$.

        \pcasep{\labelcref{rule:pstore-hsafe}}
        Then, $\inst = \IWstore{\vara}{\regb}{\regc}$ and $\rasgniof{\statec}{\regb} = \rasgniof{\stated}{(\regmapofof{\stated}{\regb})} = \adr \in \sizeof{\vara}$ due to healthiness.
        Let further $\regmapofof{\stated}{\regc} = \regcp$.
        We get $\poisonp = \subst{\poison}{(\vara, \adr)}{\poisonof{\regc}}$.
        Further, values of $\statee$ and $\statec$ are equal except for $(\vara, \adrpp)$ and
        values of $\stated$ and $\statef$ are equal except $(\vara, 0)$.
        The value of $(\vara, \adr)$ in $\statee$ is $\rasgniof{\statec}{\regc}$ and in $\statef$ is $\rasgniof{\stated}{\regcp}$.
        But $(\vara, \adr)$ has the poison value of $\poison$ for $\regc$ in $\poisonp$,
        so $\statee \nrrmp{\poisonp} \statef$ follows from $\statec \nrrmp{\poison} \stated$.

        \pcasep{\labelcref{rule:pstore-hunsafe}}
        Analogue, but $(\vara, \adr)$ swapped for $(\varb, \adrp)$.

        \pcasep{\labelcref{rule:pshuf-slh}}
        Then, $\instp = \IWslh{\regap}$, where $\regmapofof{\statec}{\rega} = \regmapofof{\stated}{\rega} = \regap$.
        We get $\poisonp = \subst{\poison}{\rega}{\weakpoisoned}$.
        $\stated$ is equal to $\statef$ and if $\nstatec$ is speculating, then $\poisonp$ is set to $\weakpoisoned$ and $\rasgniof{\statef}{\regap} = 0$.
        Otherwise, $\stated$ is equal to $\statef$ in values.
        In both cases, $\statee \nrrmp{\poisonp} \statef$ follows from $\statec \nrrmp{\poison} \stated$. \qedhere
    \end{pcases}
\end{proof}

\subsection*{Remaining transfer functions for \texorpdfstring{$\prodra$}{P :: [P]ra}}

\begin{align*}
    \ptransferof{\Inop{\pcsuc}}{\poisons} &
    = \ptransferpof{\Imove{\regap}{\regbp}{\pcsuc}}{\poisons}
    = \ptransferpof{\Ifill{\regap}{\stackl}{\pcsuc}}{\poisons}
    = \ptransferpof{\Ispill{\stackl}{\regbp}{\pcsuc}}{\poisons} = \poisons \\
    \ptransferof{\Isfence{\pcsuc}}{\poisons} & = \healthy \\
    \ptransferof{\Iif{\regb}{\pcsuc_{\btrue}}{\pcsuc_{\bfalse}}}{\poisons} & =
    \begin{cases}
        \poisons & \pregsof{\regb} = \healthy \\
        \poisoned & \pregsof{\regb} \geq \weakpoisoned
    \end{cases}
\end{align*}

\subsection*{\Cref{lem:solutionapprox}}
\begin{proof}
    By induction, all transitions update poison types monotonically.
    For the induction step, consider $\nstatec = \nstatecp\lseq\statec$ and $\nstated = \nstatedp\lseq\stated$.
    We do a case distinction on the transition rule for $(\nstatec, \nstated, \pseqofof{\nstatec}{\nstated}) \sprodtrans{\leakp}{\directp}{\leak}{\direct} (\nstatee, \nstatef, \pseqp)$.
    To that end, let $\pseqofof{\nstatec}{\nstated} = \pseq'\lseq\poison$ and $\pseqp = \pseqp'\lseq \poisonp$, $\nstatee = \nstateep\lseq\statee$ and $\nstatef = \nstatefp\lseq\statef$.
    Let further be $\inst = \progof{\statec}$ and $\instp = \traof{\stated}$.
    We need to show for all cases except \labelcref{rule:hspec,rule:prb} that $\poisonp \leq \transferof{\inst}{\poison}$ (respectively $\poisonp \leq \ptransferpof{\instp}{\poison}$ for shuffling states).
    Then, by \cref{eqn:flow-fwd}, $\poisonp \leq \transferof{\inst}{\poison} \leq \passignsof{(\statee, \statef)}$ 
    (respectively $\poisonp \leq \ptransferpof{\instp}{\poison} \leq \passignsof{(\statee, \statef)}$).
    All those rules only modify the top states and poison values, so $\pseq' = \pseqp' = \pseqofof{\nstateep}{\nstatefp}$.
    Together, we get $\pseqp \leq \pseqofof{\nstatee}{\nstatef}$.

    We first do the separate two cases who change the size of the speculating states.
    \begin{pcases}
        \pcasep{\labelcref{rule:hspec}}
        Then, $\inst = \IWif{\regb}$, $\poisonof{\regb} = \healthy$, and $\pseqp = \pseq\lseq\poison$.
        We have $\poison = \ptransferof{\inst}{\poison} \leq \passignsof{(\statec, \stated)}$.
        And because $\nstatee = \nstatec\lseq\statee$ and $\nstatef = \nstated\lseq\statef$, $\pseq \leq \pseq$, we have $\pseq\lseq\poison \leq \pseq\lseq\passignsof{(\statec,\stated)} = \pseqofof{\nstatee}{\nstatef}$.

        \pcasep{\labelcref{rule:prb}}
        Then, $\nstatee = \nstatecp$ and $\nstatef = \nstatedp$, and since $\pseqofof{\nstatec}{\nstated} = \pseqofof{\nstatecp}{\nstatedp}\lseq\poison$,
        $\pseqp' = \pseqofof{\nstatee}{\nstatef}$.
    \end{pcases}

    \noindent Now to the other cases.
    \begin{pcases}

        \pcasep{\labelcref{rule:pfill,,rule:pspill,,rule:pmove,,rule:pshuf-sfence}} $\poisonp = \poison = \ptransferpof{\instp}{\poison}$.

        \pcasep{\labelcref{rule:pnop,,rule:hbranch,,rule:psfence,,rule:pslh}} $\poisonp = \poison = \transferof{\inst}{\poison}$.

        \pcasep{\labelcref{rule:pintro-load,,rule:pload-ss,,rule:pload-us}}
        Then, $\inst = \IWload{\rega}{\vara}{\regb}$ and
        $\poisonp = \subst{\poison}{\rega}{\poison} = \transferof{\inst}{\poison}$.

        \pcasep{\labelcref{rule:pload-hsafe}}
        Then, $\inst = \IWload{\rega}{\vara}{\regb}$ and $\adr = \rasgniof{\statec}{\regb}$ and
        $\poisonp = \subst{\poison}{\rega}{\pmemof{\vara}{\adr}} \leq \subst{\poison}{\rega}{\poison} = \transferof{\inst}{\poison}$.

        \pcasep{\labelcref{rule:pload-hunsafe}}
        Then, $\inst = \IWload{\rega}{\vara}{\regb}$ and $\directp = \Dload{\varb}{\adrp}$ and
        $\poisonp = \subst{\poison}{\rega}{\pmemof{\varb}{\adrp}} \leq \subst{\poison}{\rega}{\poison} = \transferof{\inst}{\poison}$.

        \pcasep{\labelcref{rule:pintro-store}}
        Then, $\inst = \IWstore{\vara}{\regb}{\regc}$ and $\directp = \Dstore{\vara}{\adrppp}$ and $\direct = \Dstore{\varstack}{\adrp}$ and $\regmapofof{\stated}{\regd} = (\varstack, \adrp)$.
        We get $\poisonp = \subst{\poison}{\regd, (\varstack, \adrp)}{\poisoned} \leq \subst{\substjoin{\poison}{\mems}{\poisonof{\regc}}}{\regs,\vara}{\poisoned} = \transferof{\inst}{\poison}$.

        \pcasep{\labelcref{rule:pstore-ss}}
        Then, $\inst = \IWstore{\vara}{\regb}{\regc}$ and $\rasgniof{\statec}{\regb} = \adrpp$ and $\rasgniof{\stated}{(\regmapofof{\stated}{\regb}} = 0$ due to weak poisonedness.
        We get $\poisonp = \subst{\poison}{(\vara, \adrpp), (\vara, 0)}{\poisoned} \leq \subst{\substjoin{\poison}{\mems}{\poisonof{\regc}}}{\regs,\vara}{\poisoned} = \transferof{\inst}{\poison}$.

        \pcasep{\labelcref{rule:pstore-us}}
        Then, $\inst = \IWstore{\vara}{\regb}{\regc}$ and $\directp = \Dstore{\vara}{\adrppp}$ and $\rasgniof{\stated}{(\regmapofof{\stated}{\regb}} = 0$ due to weak poisonedness.
        We get $\poisonp = \subst{\poison}{(\vara, \adrppp), (\vara, 0)}{\poisoned} \leq \subst{\substjoin{\poison}{\mems}{\poisonof{\regc}}}{\regs,\vara}{\poisoned} = \transferof{\inst}{\poison}$.

        \pcasep{\labelcref{rule:pstore-hsafe}}
        Then, $\inst = \IWstore{\vara}{\regb}{\regc}$ and $\rasgniof{\statec}{\regb} = \rasgniof{\stated}{(\regmapofof{\stated}{\regb}} = \adr \in \sizeof{\vara}$ due to healthiness.
        We get $\poisonp = \subst{\poison}{(\vara, \adr)}{\poisonof{\regc}} \leq \subst{\substjoin{\poison}{\mems}{\poisonof{\regc}}}{\regs,\vara}{\poisoned} = \transferof{\inst}{\poison}$.

        \pcasep{\labelcref{rule:pstore-hunsafe}}
        Then, $\inst = \IWstore{\vara}{\regb}{\regc}$ and $\direct = \directp = \Dstore{\varb}{\adrp}$.
        We get $\poisonp = \subst{\poison}{(\varb, \adrp)}{\poisonof{\regc}} \leq \substjoin{\subst{\poison}{\mems}{\poisonof{\regc}}}{\regs,\vara}{\poisoned} = \transferof{\inst}{\poison}$.

        \pcasep{\labelcref{rule:pshuf-slh}}
        Then, $\instp = \IWslh{\regap}$, where $\regmapofof{\statec}{\rega} = \regmapofof{\stated}{\rega} = \regap$.
        We get $\poisonp = \subst{\poison}{\rega}{\weakpoisoned} = \ptransferpof{\instp}{\poison}$. \qedhere
    \end{pcases}
\end{proof}

\subsection*{Full proof for \Cref{lem:trasnippy}}

Like for dead code elimination in \Cref{section:deadcode},
we need to prove that $(\simrel, \dtffamily)$
\begin{enumerate*}[label=\bfseries (\roman*)]
    \item is a simulation (\Cref{def:simulation}),
    \item respects $\lows$ (\Cref{def:loweqrespecting}),
        and \item is snippy (\Cref{def:snippy}).
\end{enumerate*}

The intermediary lemmas are proved subsequently in their own sections.

\noindent\begin{proof}[Proof that $(\simrel, \dtffamily)$ is $\sec$-respecting]
    We required $\lowsof{\varstack} = \low$ so when $\stated_1 \seceq \stated_2$,
    and $\statec_1 \nrrmp{\healthy} \stated_1$ as well as $\statec_2 \nrrmp{\healthy} \stated_2$,
    then $\statec_1$ and $\stated_1$ equal on all values in memory other than $\varstack$ (remember that $\poison$ is not defined on $\varstack$).
    Similarly, $\statec_2$ and $\stated_2$ equal on memory.
    So, $\statec_1 \seceq \statec_2$ if and only if $\stated_1 \seceq \stated_2$.
\end{proof}

\noindent\begin{proof}[Proof that $(\simrel, \dtffamily)$ is a simulation]
    In order to prove that $\simrel$ is a simulation (\Cref{def:simulation}),
    we first need to show that for all initial $\stated$ of $\tra$
    there is an initial $\statec$ for $\prog$ with $\stated \simrel \statec$.
    Secondly, for all pairs of states $\nstated \simrel \nstatec$, we need to derive $\proofnode{\sderivetgtsim{\guarded{\simrel}}{\dtfat{\nstatec}{\nstated}}{\nstated}{\nstatec}{\varepsilon}}$ in \Cref{rules:simulation}.
    For the initial states, let $\stated = (\entryp, \rasgnp, \masgnp)$.
    We construct $\statec = (\entry, \rasgn, \masgn)$.
    First, we know from instruction matching, that $\instinjof{\entry} = \entryp$, thus $\statec \rinstinj \stated$.
    We can then choose $\masgn = \masgnp$ and $\rasgnof{\rega} = \evalin{\stated}{\regmapofof{\entryp}{\rega}}$.

    For the second part, consider $\nstated \simrel \nstatec$.
    We need to provide a proof for $\proofnode{\sderivetgtsim{\guarded{\simrel}}{\dtfat{\nstatec}{\nstated}}{\nstated}{\nstatec}{\varepsilon}}$.
    The case where $\nstated$ is final is trivial, so consider non-final $\nstated$ and $\nstatec$.
    We need two small helping-lemmas that state that $\instinj$-injected instructions and shuffle instructions of~$\tra$ are preserved to the product.
    The proof of these we skip, but they rely on the fact that $(\instinj, \regmap)$ is poison-typable with $\passigns$.
    \begin{lemma}
        \label{lem:syncedrelrmp}
        Let $\nstated \simrel \nstatec$,
        $\nstated \ntrans{\leak}{\direct} \nstateh$,
        and $\pseq\lseq\poison = \pseqofof{\nstatec}{\nstated}$.
        Then, $(\nstatec, \nstated, \poison) \sprodtrans{\leakp}{\directp}{\leak}{\direct} (\nstatee, \nstateh, \poisonp)$
        with $\nstatee \rrmp{\poisonp} \nstateh$.
    \end{lemma}
    \begin{lemma}
        \label{lem:shufflerelrmp}
        Let $\nstatee \rrmp{\poison} \nstateh$
        and $\nstateh \ntranss{\tracel}{\dtraced} \nstatef$ be shuffle only in $\tra$.
        Then, $\prodra$ has a unique execution $(\nstatee, \nstateh, \poison)
        \sprodtranss{\varepsilon}{\varepsilon}{\tracel}{\dtraced}
        (\nstatee, \nstatef, \poisonp)$ with $\nstatee \rrmp{\poisonp} \nstatef$.
    \end{lemma}
    Let $\nstated = \nstatedp\lseq\stated$ and $\nstatec = \nstatecp\lseq\statec$, and $\passignsof{(\statec, \stated)} = \poison$.
    We sketch how to derive $\proofnode{\sderivetgtsim{\guarded{\simrel}}{\dtfat{\nstatec}{\nstated}}{\nstated}{\nstatec}{\varepsilon}}$.
    We will first explore executions from $\nstated$ in $\tra$ with \labelcref{rule:tgt},
    perform a directive transformation with \labelcref{rule:direct-tf},
    replay the execution from $\nstatec$ in $\prog$,
    and finally prove that the reached states belong to $\simrel$ to end the proof with \labelcref{rule:coind}.
    First, we explore executions from $\nstated$ in $\tra$ with \labelcref{rule:tgt}.
    Every execution in $\tra$ from $\nstated$ eventually enters a $\instinj$-injected state.
    So we explore executions $\nstated \ntrans{\leak}{\direct} \nstateh \ntranss{\tracel}{\dtraced} \nstatef$,
    where $\nstatef$ is again $\instinj$-injected but no intermediary state is.
    The explored executions have the shape of the \textcolor{teal}{teal} and \textcolor{red}{red} paths of $\tra$ in \Cref{fig:phiinterval}.
    We thus need to prove $\proofnode{\sderivetgtsim{\unguarded{\simrel}}{\dtfat{\nstatec}{\nstated}}{\nstatef}{\nstatec}{\direct\lseq\dtraced}}$.
    The next step is to perform directive transformation with $\dtfat{\nstatec}{\nstated}$.
    But we defined the transformation on $\instinj$-intervals.
    So we find the appropriate transitions for $\nstated \ntrans{\leak}{\direct} \nstateh \ntranss{\tracel}{\dtraced} \nstatef$ in $\prodra$.
    For the first transition,
    \Cref{lem:syncedrelrmp} yields $(\nstatec, \nstated, \poison) \sprodtrans{\leakp}{\directp}{\leak}{\direct} (\nstatee, \nstateh, \poisonp)$
    so that $\nstatee \rrmp{\poisonp} \nstateh$.
    For the remaining transitions, there are two cases:

    \pcase{\textcolor{red}{No $\Drb$ occurs in $\dtraced$.}}
    Then, $\nstateh \ntranss{\tracel}{\dtraced} \nstatef$ is shuffle-only,
    and we can apply \Cref{lem:shufflerelrmp} for transitions
    $(\nstatee, \nstateh, \poisonp)
    \sprodtranss{\varepsilon}{\varepsilon}{\tracel}{\dtraced}
    (\nstatee, \nstatef, \poisonpp)$ with $\nstatee \rrmp{\poisonpp} \nstatef$.
    Together with the first transition, $\prodra$ can execute the transitions $(\nstatec, \nstated, \poison)
    \sprodtrans{\leakp}{\directp}{\leak}{\direct}
    (\nstatee, \nstateh, \poisonp)
    \sprodtranss{\varepsilon}{\varepsilon}{\tracel}{\dtraced}
    (\nstatee, \nstatef, \poisonpp)$.
    This is precisely the \textcolor{red}{red} $\instinj$-interval, so $\directp \dtfat{\nstatec}{\nstated} \direct\lseq\dtraced$.
    We can thus use \labelcref{rule:direct-tf} and are left proving $\proofnode{\sderivesrcsim{\guarded{\simrel}}{\dtfat{\nstatec}{\nstated}}{\nstatef}{\nstatec}{\directp}}$.
    However, $(\nstatec, \nstated, \poison) \sprodtrans{\leakp}{\directp}{\leak}{\direct} (\nstatee, \nstateh, \poisonp)$ already provides us the replay $\nstatec \ntrans{\leakp}{\directp} \nstatee$ in $\prog$.
    We can use it with \labelcref{rule:src},
    and it remains to show $\proofnode{\sderivesrcsim{\unguarded{\simrel}}{\dtfat{\nstatec}{\nstated}}{\nstatef}{\nstatee}{\varepsilon}}$.
    In order to use \labelcref{rule:coind}, we need to establish $\nstatef \simrel \nstatee$,
    i.e.\ $\nstatee \rinstinj \nstatef$ and $\nstatee \relrm \nstatef$.
    The condition $\nstatee \rinstinj \nstatef$ holds because it has been the condition for terminating exploration in $\tra$.
    For the latter condition, we have already established $\nstatee \rrmp{\poisonpp} \nstatef$.
    Thanks to \Cref{lem:solutionapprox},
    we also have $\pseqofof{\nstateep}{\nstatefp}\lseq\poisonpp \leq \pseqofof{\nstatee}{\nstatef}$,
    where $\nstatee = \nstateep\lseq\statee$ and $\nstatef = \nstatefp\lseq\statef$.
    We designed the lattice so that $\nrrmp{\poison}$ is monotonic in $\poison$, so they yield $\nstatee \relrm \nstatef$.
    Thus, we get $\nstatef \simrel \nstatee$ and can use \labelcref{rule:coind}.

    \pcase{\textcolor{teal}{Otherwise, let $\Drb$ occur in $\dtraced$.}}
    Then, the explored execution from $\nstateh$ is $\nstateh \ntranss{\tracelp}{\dtracedp} \nstatey \ntrans{\Erbless}{\Drb} \nstatef$ with $\dtraced = \dtracedp\lseq\Drb$.
    Indeed, $\dtraced$ must end with $\Drb$, because $\nstatef$ is already $\instinj$-injected:
    When $\nstateh = \nstatehp\lseq\stateh$ and $\nstatee = \nstateep\lseq\statee$,
    then $\nstatef = \nstatehp$.
    Further, $\nstatehp$ and $\nstateep$ are already contained in $\nstated$ and $\nstatec$, respectively.
    Because $\nstatec \rinstinj \nstated$ we also have $\nstateep \rinstinj \nstatehp = \nstatef$.
    So $\nstatef$ is $\instinj$-injected and exploration stopped upon discovering it.
    Again, by \Cref{lem:shufflerelrmp}
    there is 
    $(\nstatee, \nstateh, \poisonp)
    \sprodtranss{\varepsilon}{\varepsilon}{\tracel}{\dtraced}
    (\nstatee, \nstatey, \poisonpp)$ with $\nstatee \rrmp{\poisonpp} \nstatey$.
    By definition of $\prodra$, $(\nstatee, \nstatey, \poisonp) \prodtrans{\Erb}{\Drb}{\Erb}{\Drb} (\nstateep, \nstatehp, \pseqofof{\nstateep}{\nstatehp})$.
    This is one of the \textcolor{teal}{teal} $\instinj$-intervals, so $\directp\lseq\Drb \dtfat{\nstatec}{\nstated} \direct\lseq\dtracedp\lseq\Drb$.
    The remaining arguments for \labelcref{rule:src,rule:coind} are similar to the previous case.
\end{proof}

\begin{corollary}
    For $\nstated \simrel \nstatec$, $(\nstatec, \nstated, \pseqofof{\nstatec}{\nstated}) \sprodtranss{\tracek}{\dtracee}{\tracel}{\dtraced} (\nstatee, \nstatef, \pseqp)$ is a $\instinj$-interval 
    iff $(\nstatec, \nstated) \synctrans{\tracek}{\dtracee}{\tracel}{\dtraced} (\nstatee, \nstatef)$.
\end{corollary}

\noindent\begin{proof}[Proof that $(\simrel, \dtffamily)$ is snippy]
    In order to prove the simulation snippy, we are given the black parts of \Cref{fig:snippysim}.
    So, consider $\nstated_1 \simrel \nstatec_1$, $\nstated_2 \simrel \nstatec_2$, $\nstated_1 \samepoint \nstated_2$, and $\nstatec_1 \samepoint \nstatec_2$.
    Further, consider a simulation interval $(\nstatec_1, \nstated_1) \synctrans{\tracek}{\dtracee}{\tracel}{\dtraced} (\nstatee_1, \nstatef_1)$ and $\nstatec_2 \ntrans{\tracek}{\dtracee} \nstatee_2$.
    We need to prove $(\nstatec_2, \nstated_2) \synctrans{\tracek}{\dtracee}{\tracel}{\dtraced} (\nstatee_2, \nstatef_2)$.
    As simulation intervals are the same as $\instinj$-intervals,
    we have the $\instinj$-interval $(\nstatec_1, \nstated_1, \pseqofof{\nstatec_1}{\nstated_1}) \sprodtranss{\tracek}{\dtracee}{\tracel}{\dtraced} (\nstatee_1, \nstatef_1, \poisonp)$,
    and want to prove that $(\nstatec_2, \nstated_2, \pseqofof{\nstatec_2}{\nstated_2}) \sprodtrans{\tracek}{\dtracee}{\tracel}{\dtraced} (\nstatee_2, \nstatef_2, \poisonp)$.
    We again split the interval into the instruction-matched transition and a shuffle sequence:
    $(\nstatec_1, \nstated_1, \pseqofof{\nstatec_1}{\nstatec_1}) \sprodtrans{\leakp}{\directp}{\leak}{\direct} (\nstateg_1, \nstateh_1, \poisonpp) \sprodtranss{\tracekp}{\dtraceep}{\tracelp}{\dtracedp} (\nstatee_1, \nstatef_1, \poisonp)$.
    The difficult part is to reproduce the instruction-matched transition from $(\nstatec_2, \nstated_2, \pseqofof{\nstatec_2}{\nstated_2})$.
    Because $\pseqofof{\nstatec_1}{\nstated_1}$ is only dependent on program counters, $\pseqofof{\nstatec_1}{\nstated_1} = \pseqofof{\nstatec_2}{\nstated_2}$.
    \begin{lemma}
        There is a transition $(\nstatec_2, \nstated_2, \pseqofof{\nstatec_2}{\nstated_2}) \sprodtrans{\leakp}{\directp}{\leak}{\direct} (\nstatee_2, \nstatef_2, \poisonp)$ in $\prodra$.
        \label{lem:repeat-instinj}
    \end{lemma}
    We reproduce the second part of the $\instinj$-interval,
    $(\nstateg_1, \nstateh_1, \poisonpp_1) \sprodtranss{\tracekp}{\dtraceep}{\tracelp}{\dtracedp} (\nstatee_1, \nstatef_1, \poisonp_1)$.
    Once more, we perform case distinction on whether $\dtracedp$ and $\dtraceep$ follow the \textcolor{red}{red} or \textcolor{teal}{teal} paths in \Cref{fig:phiinterval}.
    \pcase{\textcolor{red}{No $\Drb$ occurs in $\dtracedp$.}}
    In this case, $\dtracedp$ is shuffle-only and $\dtraceep = \varepsilon$.
    Because shuffle semantics are deterministic and all have the same directive $\Dstep$,
    the following lemma is straightforward to show.
    Together with \Cref{lem:shufflerelrmp}, the shuffle sequence is reproduced from $(\nstateg_2, \nstateh_2, \poisonp)$.
    \begin{lemma}
        If $\nstatef_1 \samepoint \nstatef_2$ and $\instof{\nstatef_1} = \shufseq{\pcsucp}$ with shuffle-only $\nstatef_1 \ntranss{\tracel}{\dtraced} \nstateh_1$,
        then $\nstatef_2 \ntranss{\tracel}{\dtraced} \nstateh_2$.
        \label{lem:repeat-shuffle}
    \end{lemma}
    Indeed, we complete reproduction in this case:
    With $(\nstatec_2, \nstated_2, \poison) \sprodtrans{\leakp}{\directp}{\leak}{\direct} (\nstateg_2, \nstateh_2, \poisonpp)$ from \Cref{lem:repeat-instinj}
    and $(\nstateg_2, \nstateh_2, \poisonpp_2) \sprodtranss{\varepsilon}{\varepsilon}{\tracelp}{\dtracedp} (\nstatee_2, \nstatef_2, \poisonp)$ from \Cref{lem:repeat-shuffle,lem:shufflerelrmp},
    we constructed $(\nstatec_2, \nstated_2, \pseqofof{\nstatec_2}{\nstated_2}) \sprodtrans{\tracek}{\dtracee}{\tracel}{\dtraced} (\nstatee_2, \nstatef_2, \poisonp)$.

    \pcase{\textcolor{teal}{Otherwise, let $\Drb$ occur in $\dtraced$.}}
    Then, $(\nstateg_1, \nstateh_1, \poisonpp) \sprodtranss{\tracekp}{\dtraceep}{\tracelp}{\dtracedp} (\nstatee_1, \nstatef_1, \poisonp)$
    is actually the teal path $(\nstateg_1, \nstateh_1, \poisonpp) \sprodtranss{\varepsilon}{\varepsilon}{\tracelp'}{\dtracedp'} (\nstategp_1, \nstatehp_1, \poisonpp') \sprodtrans{\Erb}{\Drb}{\Erb}{\Drb} (\nstatee_1, \nstatef_1, \poisonp)$,
    where $\dtracedp'$ is a shuffle sequence.
    Same as in the previous case, we can reproduce $(\nstateg_2, \nstateh_2, \poisonpp) \sprodtranss{\varepsilon}{\varepsilon}{\tracelp'}{\dtracedp'} (\nstategp_2, \nstatehp_2, \poisonpp')$.
    The product can clearly execute $(\nstategp_2, \nstatehp_2, \poisonpp') \sprodtrans{\Erb}{\Drb}{\Erb}{\Drb} (\nstatee_2, \nstatef_2, \poisonp)$
    because $\sizeof{\nstategp_2} = \sizeof{\nstatehp_2} = \sizeof{\nstategp_1} = \sizeof{\nstatehp_1} > 1$.
    Reproduction of $(\nstatec_2, \nstated_2, \pseqofof{\nstatec_2}{\nstated_2}) \sprodtrans{\tracek}{\dtracee}{\tracel}{\dtraced} (\nstatee_2, \nstatef_2, \poisonp)$ is done.
\end{proof}

\subsection*{\Cref{lem:syncedrelrmp}}
\begin{proof}
    By construction of the product.
    Case distinction on $\nstated_1 \ntrans{\leak}{\direct} \nstateh_1$.
    Throughout the proof, we assume $\regmapofof{\nstated}{\rega} = \regap$, 
    $\regmapofof{\nstated}{\regb} = \regbp$,
    and $\regmapofof{\nstated}{\regc} = \regcp$.
    Further, we invariantly let $\nstated = \nstatedp\lseq\stated$,
    $\nstatec = \nstatecp\lseq\statec$,
    $\nstateh = \nstatehp\lseq\stateh$,
    $\nstatee = \nstateep\lseq\statee$,
    Also $\instof{\nstated} = \instp$ and $\instof{\nstatec} = \inst$.
    $\nstatec \rrmp{\poisonp} \nstateh$ follows from \Cref{lem:prod-well-def}.

    \begin{pcases}
        \pcasep{\labelcref{rule:nop,,rule:asgn}} We skip \labelcref{rule:nop}.
        For \labelcref{rule:asgn}, let $\instp = \Iasgn{\regap}{\regbp \op \regcp}{\pcsucp}$.
        Then $\inst = \Iasgn{\rega}{\regb \op \regc}{\pcsuc}$.
        Clearly there is $\nstatee$ with $\nstatec \ntrans{\Enone}{\Dstep} \nstatee$.
        \Cref{rule:pasgn} yields
        $(\nstatec, \nstated, \poison) \sprodtrans{\Enone}{\Dstep}{\Enone}{\Dstep} (\nstatee, \nstateh, \poisonp)$.
        $\nstatee \rrmp{\poisonp} \nstateh$ follows from construction of \Cref{rule:pasgn}.

        \pcasep{\labelcref{rule:load}}
        For \labelcref{rule:load}, let $\instp = \Iload{\regap}{\vara}{\regbp}{\pcsucp}$ and $\leak = \Eload{\adr}$.
        We have $\poisonof{\regb} \geq \weakpoisoned$,
        thus $\rasgniof{\nstated}{\regbp} = 0$ or $\rasgniof{\nstated}{\regbp} = \rasgniof{\nstatec}{\regb}$.
        \begin{smcases}
            \smcase{$\rasgniof{\nstatec}{\regb} = \adrp \in \sizeof{\vara}$}
            Clearly there is $\nstatee$ with $\nstatec \ntrans{\Eload{\adrp}}{\Dstep} \nstatee$. \\
            \labelcref{rule:pload-ss} yields $(\nstatec, \nstated, \poison) \prodtrans{\Eload{\adrp}}{\Dstep}{\Eload{\adr}}{\Dstep} (\nstatee, \nstateh, \poisonp)$.

            \smcase{$\rasgniof{\nstatec}{\regb} = \adrp \notin \sizeof{\vara}$}
            Clearly there is $\nstatee$ with $\nstatec \ntrans{\Eload{\adrp}}{\Dload{\vara}{0}} \nstatee$. \\
            \labelcref{rule:pload-us} yields $(\nstatec, \nstated, \poison) \prodtrans{\Eload{\adrp}}{\Dload{\vara}{0}}{\Eload{\adr}}{\Dstep} (\nstatee, \nstateh, \poisonp)$.
        \end{smcases}

        \pcasep{\labelcref{rule:store}}
        For \labelcref{rule:store}, let $\instp = \Istore{\vara}{\regbp}{\regcp}{\pcsucp}$ and $\leak = \Estore{\adr}$.
        We have $\poisonof{\regb} \geq \weakpoisoned$,
        thus $\rasgniof{\nstated}{\regbp} = 0$ or $\rasgniof{\nstated}{\regbp} = \rasgniof{\nstatec}{\regb}$.

        \begin{smcases}
            \smcase{$\rasgniof{\nstatec}{\regb} = \adrp \in \sizeof{\vara}$}
            Clearly there is $\nstatee$ with $\nstatec \ntrans{\Estore{\adrp}}{\Dstep} \nstatee$. \\
            \labelcref{rule:pstore-ss} yields $(\nstatec, \nstated, \poison) \prodtrans{\Estore{\adrp}}{\Dstep}{\Estore{\adr}}{\Dstep} (\nstatee, \nstateh, \poisonp)$.

            \smcase{$\rasgniof{\nstatec}{\regb} = \adrp \notin \sizeof{\vara}$}
            Clearly there is $\nstatee$ with $\nstatec \ntrans{\Estore{\adrp}}{\Dstore{\vara}{0}} \nstatee$. \\
            \labelcref{rule:pstore-us} yields $(\nstatec, \nstated, \poison) \prodtrans{\Estore{\adrp}}{\Dstore{\vara}{0}}{\Estore{\adr}}{\Dstep} (\nstatee, \nstateh, \poisonp)$.
        \end{smcases}

        \pcasep{\labelcref{rule:load-unsafe}}
        For \labelcref{rule:load-unsafe}, let $\instp = \Iload{\regap}{\vara}{\regbp}{\pcsucp}$ and $\leak = \Eload{\adr}$ and $\direct = \Dload{\varb}{\adrp}$.
        We have $\poisonof{\regb} \geq \weakpoisoned$,
        thus $\rasgniof{\nstated}{\regbp} = 0$ or $\rasgniof{\nstated}{\regbp} = \rasgniof{\nstatec}{\regb}$.
        Due to $0 \in \sizeof{\vara}$, only $\rasgniof{\nstated}{\regbp} = \rasgniof{\nstatec}{\regb}$ needs to be considered.
        We consider two cases:

        \begin{smcases}
            \smcase{$\varb \neq \varstack$}
            There is $\nstatee$ with $\nstatec \ntrans{\Eload{\adrp}}{\Dload{\varb}{\adrp}} \nstatee$.\\
            \labelcref{rule:pload-hunsafe} yields $(\nstatec, \nstated, \poison) \prodtrans{\Eload{\adrp}}{\Dstep}{\Eload{\adr}}{\Dstep} (\nstatee, \nstateh, \poisonp)$.

            \smcase{$\varb = \varstack$}
            There is $\nstatee$ with $\nstatec \ntrans{\Eload{\adrp}}{\Dload{\vara}{0}} \nstatee$. \\
            \labelcref{rule:pintro-load} yields $(\nstatec, \nstated, \poison) \prodtrans{\Eload{\adrp}}{\Dload{\vara}{0}}{\Eload{\adr}}{\Dstep} (\nstatee, \nstateh, \poisonp)$.
        \end{smcases}

        \pcasep{\labelcref{rule:store-unsafe}}
        For \labelcref{rule:store-unsafe}, let $\instp = \Istore{\vara}{\regbp}{\regcp}{\pcsucp}$ and $\leak = \Estore{\adr}$ and $\direct = \Dstore{\varb}{\adrp}$.
        We have $\poisonof{\regb} \geq \weakpoisoned$,
        thus $\rasgniof{\nstated}{\regbp} = 0$ or $\rasgniof{\nstated}{\regbp} = \rasgniof{\nstatec}{\regb}$.
        Due to $0 \in \sizeof{\vara}$, only $\rasgniof{\nstated}{\regbp} = \rasgniof{\nstatec}{\regb}$ needs to be considered.
        We consider two cases:

        \begin{smcases}
            \smcase{$\varb \neq \varstack$}
            There is $\nstatee$ with $\nstatec \ntrans{\Estore{\adrp}}{\Dstore{\varb}{\adrp}} \nstatee$. \\
            \labelcref{rule:pstore-hunsafe} yields $(\nstatec, \nstated, \poison) \prodtrans{\Estore{\adrp}}{\Dstep}{\Estore{\adr}}{\Dstep} (\nstatee, \nstateh, \poisonp)$.

            \smcase{$\varb = \varstack$}
            There is $\nstatee$ with $\nstatec \ntrans{\Estore{\adrp}}{\Dstore{\vara}{0}} \nstatee$. \\
            \labelcref{rule:pintro-store} yields $(\nstatec, \nstated, \poison) \prodtrans{\Estore{\adrp}}{\Dstore{\vara}{0}}{\Estore{\adr}}{\Dstep} (\nstatee, \nstateh, \poisonp)$.
        \end{smcases}

        \pcasep{\labelcref{rule:branch}}
        For \labelcref{rule:branch}, let $\instp = \Iif{\regbp}{\pcsucp_{\btrue}}{\pcsucp_{\bfalse}}$ and $\leak = \Eif{\bvalue}$ and $\direct = \Dif$.
        We have $\poisonof{\regb} = \healthy$,
        thus $\rasgniof{\nstated}{\regbp} = \rasgniof{\nstatec}{\regb}$.
        Then there is $\nstatee$ with $\nstatec \ntrans{\Eif{\bvalue}}{\Dif} \nstatee$.
        \labelcref{rule:hbranch} yields $(\nstatec, \nstated, \poison) \prodtrans{\Eif{\bvalue}}{\Dstep}{\Eif{\bvalue}}{\Dstep} (\nstatee, \nstateh, \poisonp)$.

        \pcasep{\labelcref{rule:nspec}}
        For \labelcref{rule:nspec}, let $\instp = \Iif{\regbp}{\pcsucp_{\btrue}}{\pcsucp_{\bfalse}}$ and $\leak = \Eif{\bvalue}$ and $\direct = \Dspec$.
        We have $\poisonof{\regb} = \healthy$,
        thus $\rasgniof{\nstated}{\regbp} = \rasgniof{\nstatec}{\regb}$.
        Then there is $\nstatee$ with $\nstatec \ntrans{\Eif{\bvalue}}{\Dspec} \nstatee$.
        \labelcref{rule:hspec} yields $(\nstatec, \nstated, \poison) \prodtrans{\Eif{\bvalue}}{\Dspec}{\Eif{\bvalue}}{\Dspec} (\nstatee, \nstateh, \poisonp)$.

        \pcasep{\labelcref{rule:nsfence}}
        For \labelcref{rule:nsfence}, let $\instp = \Isfence{\pcsucp}$ and $\leak = \Enone$ and $\direct = \Dstep$.
        Then $\sizeof{\nstated} = 1 = \sizeof{\nstatec}$.
        Thus there is $\nstatee$ with $\nstatec \ntrans{\Enone}{\Dstep} \nstatee$.
        \labelcref{rule:psfence} yields $(\nstatec, \nstated, \poison) \prodtrans{\Enone}{\Dstep}{\Enone}{\Dstep} (\nstatee, \nstateh, \poisonp)$.

        \pcasep{\labelcref{rule:nslh}}
        For \labelcref{rule:nslh}, let $\instp = \Islh{\regap}{\pcsucp}$ and $\leak = \Enone$ and $\direct = \Dstep$.
        Thus there is $\nstatee$ with $\nstatec \ntrans{\Enone}{\Dstep} \nstatee$.
        \labelcref{rule:pslh} yields $(\nstatec, \nstated, \poison) \prodtrans{\Enone}{\Dstep}{\Enone}{\Dstep} (\nstatee, \nstateh, \poisonp)$. \qedhere
    \end{pcases}
\end{proof}

\subsection*{\Cref{lem:shufflerelrmp}}
\begin{proof}
    Prove this for a single step, the rest is induction.
    Let further be $\nstatee = \nstateep\lseq\statee$ and $\nstateh = \nstatehp\lseq\stateh$. Let further  be $\poison = (\preg, \pmem)$.
    Let $\statee = (\pc, \rasgn, \masgn)$, $\stateh = (\pcp, \rasgnp, \masgnp)$, $\nstatec \rrmp{\poison} \nstateh$ and $\traof{\pcp} = \shuffleinst{\pcssp}$.
    We do case distinction on $\shuffleinstcmd$.

    \begin{pcases}
        \pcasep{$\Imove{\regap}{\regbp}{\pcssp}$}
        Then, $\stateh = (\pcp, \rasgnp, \masgnp) \trans{\Enone}{\Dstep} (\pcssp, \rasgnpp, \masgnp) = \statef$ with $\subst{\rasgnp}{\regap}{\rasgnpof{\regbp}} = \rasgnpp$.
        By shuffle conformity, we know that there is $\regb$ with $\regmapofof{\pcp}{\regb} = \regbp$.
        Further, we know that $\regmapof{\pcssp}$ coincides with $\regmapof{\pcp}$ except for $\regb$.
        Also, $\poison = \poisonp$.
        By induction, $\nstatec \rrmp{\poison} \nstateh$,
        we have $\rrmat{\statef}{\regb} = \rasgnppof{\regap} = \rasgnpof{\regbp} = \rrmat{\stateh}{\regb} = \rasgnof{\regb}$
        as required for $\nstatec \rrmp{\poison} \nstatef$.
        \labelcref{rule:pmove} yields the transition.

        \pcasep{$\Ifill{\regap}{\stackl}{\pcssp}$}
        Then, $\stateh = (\pcp, \rasgnp, \masgnp) \trans{\Eload{\stackl}}{\Dstep} (\pcssp, \rasgnpp, \masgnp) = \statef$ with $\subst{\rasgnp}{\regap}{\masgnof{\varstack}{\stackl}} = \rasgnpp$.
        By shuffle conformity, we know that there is $\regb$ with $\regmapofof{\pcp}{\regb} = \stackl$.
        Further, we know that $\regmapof{\pcssp}$ coincides with $\regmapof{\pcp}$ except for $\regb$, where $\regmapofof{\pcssp}{\regb} = \regap$.
        Also, $\poison = \poisonp$.
        By induction, $\nstatec \rrmp{\poison} \nstateh$,
        we have $\rrmat{\statef}{\regb} = \rasgnppof{\regap} = \masgnpof{\varstack}{\stackl} = \rrmat{\stateh}{\regb} = \rasgnof{\regb}$
        as required for $\nstatec \rrmp{\poison} \nstatef$.
        \labelcref{rule:pfill} yields the transition.

        \pcasep{$\Ispill{\stackl}{\regbp}{\pcssp}$}
        Then, $\stateh = (\pcp, \rasgnp, \masgnp) \trans{\Estore{\stackl}}{\Dstep} (\pcssp, \rasgnp, \masgnpp) = \statef$ with $\subst{\masgnp}{(\varstack,\stackl)}{\rasgnpof{\regbp}} = \masgnpp$.
        By shuffle conformity, we know that there is $\regb$ with $\regmapofof{\pcp}{\regb} = \regbp$.
        Further, we know that $\regmapof{\pcssp}$ coincides with $\regmapof{\pcp}$ except for $\regb$, where $\regmapofof{\pcssp}{\regb} = \stackl$.
        Also, $\poison = \poisonp$.
        By induction, $\nstatec \rrmp{\poison} \nstateh$,
        we have $\rrmat{\statef}{\regb} = \masgnppof{\varstack}{\stackl} = \rasgnpof{\regbp} = \rrmat{\stateh}{\regb} = \rasgnof{\regb}$
        as required for $\nstatec \rrmp{\poison} \nstatef$.
        \labelcref{rule:pspill} yields the transition.

        \pcasep{$\Isfence{\pcssp}$}
        Then $\sizeof{\nstatehp} = 0$ and $\stateh = (\pcp, \rasgnp, \masgnp) \trans{\Enone}{\Dstep} (\pcssp, \rasgnp, \masgnp) = \statef$.
        Also, $\poison = \poisonp$.
        $\statec \rrmp{\poison} \stateh$ implies $\statec \rrmp{\poison} \statef = \nstatef$.
        \labelcref{rule:pshuf-sfence} yields the transition.

        \pcasep{$\Islh{\regap}{\pcssp}$}
        If $\sizeof{\nstatehp} = 0$, this case is fully analogue to the previous one.
        Thus, let $\sizeof{\nstatehp} \geq 1$.
        Then, $\stateh = (\pcp, \rasgnp, \masgnp) \trans{\Enone}{\Dstep} (\pcssp, \rasgnpp, \masgnp) = \statef$ with $\subst{\rasgnp}{\regap}{0} = \rasgnpp$.
        By shuffle conformity, we know that $\regmapof{\pcp} = \regmapof{\pcssp}$.
        Also $\poisonp = (\subst{\preg}{\rega}{\weakpoisoned}, \pmem)$.
        Indeed, $\nstatec \rrmp{\poison} \nstateh$ and $\rrmat{\statef}{\rega} = 0$ yields $\nstatec \rrmp{\poison} \nstatef$.
        \labelcref{rule:pshuf-slh} yields the transition. \qedhere
    \end{pcases}
\end{proof}

\subsection*{\Cref{lem:repeat-instinj}}

\begin{proof}
    By case distinction on $(\nstatec_1, \nstated_1, \poison) \sprodtrans{\leakp}{\directp}{\leak}{\direct} (\nstatee_1, \nstatef_1, \poisonp)$.
    For all cases, and $i \in \set{1,2}$:
    Let
    $\nstated_i = \nstatedp_i\lseq\stated_i$,
    $\nstatec_i = \nstatecp_i\lseq\statec_i$,
    $\nstatef_i = \nstatefp_i\lseq\statef_i$,
    and
    $\nstatee_i = \nstateep_i\lseq\statee_i$,
    and $\stated_1 \atpc \pcp \atpc \stated_2$
    as well as $\statec_1 \atpc \pc \atpc \statec_2$.
    Further, let $\traof{\pcp} = \instp$ and $\progof{\pc} = \inst$,
    and assume $\regmapofof{\pcp}{\rega} = \regap$, 
    $\regmapofof{\pcp}{\regb} = \regbp$,
    and $\regmapofof{\pcp}{\regc} = \regcp$
    if they occur in $\inst$ and $\instp$.
    Finally, let $\pseqofof{\nstatec_1}{\nstated_1} = \pseqofof{\nstatec_2}{\nstated_2} = \pseq\lseq\poison$.
    Please note, that $\poisonp$ is not dependent on values:
    Whenever the same rule in $\prodra$ is executed, then the same $\poisonp$ is obtained.
    The arguments for equality of $\poisonp$ are thus skipped.
    \begin{pcases}
        \pcasep{\labelcref{rule:pstore-hsafe}}
        In this case, $\inst = \Istore{\vara}{\regb}{\regc}{\pcsuc}$ and $\instp = \Istore{\vara}{\regbp}{\regcp}{\pcssp}$.
        The presumptions for transition $(\nstatec_1, \nstated_1, \poison) \sprodtrans{\leakp}{\directp}{\leak}{\direct} (\nstatee_1, \nstatef_1, \poisonp)$ yield $\stated_1 \trans{\Estore{\adr}}{\Dstep} \statef_1$,
        and $\statec_1 \trans{\Estore{\adr}}{\Dstep} \statee_1$,
        and $\poisonof{\regb} = \healthy$.
        Our assumption is that $\statec_2 \trans{\Estore{\adr}}{\Dstep} \statee_2$ can be executed.
        The leaked address is $\adr = \rasgniof{\statec_2}{\regb} \in \sizeof{\vara}$.
        Due to $\poisonof{\regb} = \healthy$ and $\nstatec_2 \rrmp{\poison} \nstated_2$,
        $\rasgniof{\stated_2}{\regbp} = \rasgniof{\statec_2}{\regb} = \adr$.
        This suffices for $\nstated_2 \ntrans{\Estore{\adr}}{\Dstep} \nstatef_2$,
        and in turn for $(\nstatec_2, \nstated_2, \poison) \sprodtrans{\Estore{\adr}}{\Dstep}{\Estore{\adr}}{\Dstep} (\nstatee_2,\nstatef_2,\poisonp)$.

        \pcasep{\labelcref{rule:pintro-load}}
        Then, $\inst = \Iload{\rega}{\vara}{\regb}{\pcsuc}$ and $\inst = \Iload{\rega}{\vara}{\regb}{\pcsuc}$.
        The presumptions yield $\stated_1 \trans{\Eload{\adr}}{\Dload{\varstack}{\adrp}} \statef_1$,
        and $\statec_1 \trans{\Eload{\adr}}{\Dload{\varb}{\adrppp}} \statee_1$,
        and $\poisonof{\regb} = \healthy$.
        Our assumption is $\statec_2 \trans{\Eload{\adr}}{\Dload{\varb}{\adrppp}} \statee_2$.
        Due to $\poisonof{\regb} = \healthy$,
        $\rasgniof{\stated_1}{\regbp} = \rasgniof{\statec_1}{\regb} = \adr = \rasgniof{\statec_2}{\regb} = \rasgniof{\stated_2}{\regbp} \notin \sizeof{\vara}$.
        That justifies $\nstated_2 \ntrans{\Eload{\adr}}{\Dload{\varstack}{\adrp}} \nstatef_2$.
        That suffices for $(\nstatec_2, \nstated_2, \poison) \sprodtrans{\Eload{\adr}}{\Dload{\varb}{\adrppp}}{\Eload{\adr}}{\Dload{\varstack}{\adrp}} (\nstatee_2,\nstatef_2,\poisonp)$.

        \pcasep{\labelcref{rule:pnop,,rule:pasgn}} Trivial.

        \pcasep{\labelcref{rule:pintro-store}}
        Then, $\inst = \Istore{\vara}{\regb}{\regc}{\pcsuc}$,
        $\stated_1 \trans{\Estore{\adr}}{\Dstore{\varstack}{\adrp}} \statef_1$,
        and $\statec_1 \trans{\Estore{\adr}}{\Dstore{\varb}{\adrppp}} \statee_1$,
        and $\poisonof{\regb} = \healthy$.
        Our assumption is $\statec_2 \trans{\Estore{\adr}}{\Dstore{\varb}{\adrppp}} \statee_2$.
        Due to $\poisonof{\regb} = \healthy$,
        $\rasgniof{\stated_1}{\regbp} = \rasgniof{\statec_1}{\regb} = \adr = \rasgniof{\statec_2}{\regb} = \rasgniof{\stated_2}{\regbp} \notin \sizeof{\vara}$.
        That justifies $\nstated_2 \ntrans{\Eload{\adr}}{\Dload{\varstack}{\adrp}} \nstatef_2$.
        That suffices for $(\nstatec_2, \nstated_2, \poison) \sprodtrans{\Estore{\adr}}{\Dstore{\varb}{\adrppp}}{\Eload{\adr}}{\Dload{\varstack}{\adrp}} (\nstatee_2,\nstatef_2,\poisonp)$.

        \pcasep{\labelcref{rule:pload-ss}}
        Then, $\inst = \Iload{\rega}{\vara}{\regb}{\pcsuc}$,
        $\stated_1 \trans{\Eload{0}}{\Dstep} \statef_1$,
        and $\statec_1 \trans{\Eload{\adrpp}}{\Dstep} \statee_1$,
        and $\poisonof{\regb} = \weakpoisoned$.
        Further $\statec_2 \trans{\Eload{\adrpp}}{\Dstep} \statee_2$.
        Due to $\poisonof{\regb} = \weakpoisoned$,
        $\rasgniof{\nstated_2}{\regbp} = 0 \in \sizeof{\vara}$
        justifies $\nstated_2 \ntrans{\Eload{0}}{\Dstep} \nstatef_2$.
        That suffices for $(\nstatec_2, \nstated_2, \poison) \sprodtrans{\Eload{\adrpp}}{\Dstep}{\Eload{0}}{\Dstep} (\nstatee_2,\nstatef_2,\poisonp)$.

        \pcasep{\labelcref{rule:pload-us}}
        Then, $\inst = \Iload{\rega}{\vara}{\regb}{\pcsuc}$,
        $\stated_1 \trans{\Eload{0}}{\Dstep} \statef_1$,
        and $\statec_1 \trans{\Eload{\adrpp}}{\Dload{\varb}{\adrppp}} \statee_1$,
        and $\poisonof{\regb} = \weakpoisoned$.
        Further $\statec_2 \trans{\Eload{\adrpp}}{\Dload{\varb}{\adrppp}} \statee_2$.
        Due to $\poisonof{\regb} = \weakpoisoned$,
        $\rasgniof{\nstated_2}{\regbp} = 0 \in \sizeof{\vara}$
        justifies $\nstated_2 \ntrans{\Eload{0}}{\Dstep} \nstatef_2$.
        That suffices for $(\nstatec_2, \nstated_2, \poison) \sprodtrans{\Eload{\adrpp}}{\Dload{\varb}{\adrppp}}{\Eload{0}}{\Dstep} (\nstatee_2,\nstatef_2,\poisonp)$.

        \pcasep{\labelcref{rule:pload-hsafe}}
        Then, $\inst = \Iload{\rega}{\vara}{\regb}{\pcsuc}$,
        $\stated_1 \trans{\Eload{\adr}}{\Dstep} \statef_1$,
        and $\statec_1 \trans{\Eload{\adr}}{\Dstep} \statee_1$,
        and $\poisonof{\regb} = \healthy$.
        Further $\statec_2 \trans{\Eload{\adr}}{\Dstep} \statee_2$.
        Due to $\poisonof{\regb} = \healthy$,
        $\rasgniof{\nstated_1}{\regbp} = \rasgniof{\nstatec_1}{\regb} = \adr = \rasgniof{\nstatec_2}{\regb} = \rasgniof{\nstated_2}{\regbp} \in \sizeof{\vara}$
        justifies $\nstated_2 \ntrans{\Eload{\adr}}{\Dstep} \nstatef_2$.
        That suffices for $(\nstatec_2, \nstated_2, \poison) \sprodtrans{\Eload{\adr}}{\Dstep}{\Eload{\adr}}{\Dstep} (\nstatee_2,\nstatef_2,\poisonp)$.

        \pcasep{\labelcref{rule:pload-hunsafe}}
        Then, $\inst = \Iload{\rega}{\vara}{\regb}{\pcsuc}$,
        $\stated_1 \trans{\Eload{\adr}}{\Dload{\varb}{\adrp}} \statef_1$,
        and $\statec_1 \trans{\Eload{\adr}}{\Dload{\varb}{\adrp}} \statee_1$,
        and $\poisonof{\regb} = \healthy$.
        Further $\statec_2 \trans{\Eload{\adr}}{\Dload{\varb}{\adrp}} \statee_2$.
        Due to $\poisonof{\regb} = \healthy$,
        $\rasgniof{\nstated_1}{\regbp} = \rasgniof{\nstatec_1}{\regb} = \adr = \rasgniof{\nstatec_2}{\regb} = \rasgniof{\nstated_2}{\regbp} \notin \sizeof{\vara}$
        justifies $\nstated_2 \ntrans{\Eload{\adr}}{\Dload{\varb}{\adrp}} \nstatef_2$.
        That suffices for $(\nstatec_2, \nstated_2, \poison) \sprodtrans{\Eload{\adr}}{\Dload{\varb}{\adrp}}{\Eload{\adr}}{\Dload{\varb}{\adrp}} (\nstatee_2,\nstatef_2,\poisonp)$.

        \pcasep{\labelcref{rule:pstore-ss}}
        Then, $\inst = \Istore{\vara}{\regb}{\regc}{\pcsuc}$,
        $\stated_1 \trans{\Estore{0}}{\Dstep} \statef_1$,
        and $\statec_1 \trans{\Estore{\adrpp}}{\Dstep} \statee_1$,
        and $\poisonof{\regb} = \weakpoisoned$.
        Further $\statec_2 \trans{\Estore{\adrpp}}{\Dstep} \statee_2$.
        Due to $\poisonof{\regb} = \weakpoisoned$,
        $\rasgniof{\nstated_2}{\regbp} = 0 \in \sizeof{\vara}$
        justifies $\nstated_2 \ntrans{\Estore{0}}{\Dstep} \nstatef_2$.
        That suffices for $(\nstatec_2, \nstated_2, \poison) \sprodtrans{\Estore{\adrpp}}{\Dstep}{\Estore{0}}{\Dstep} (\nstatee_2,\nstatef_2,\poisonp)$.

        \pcasep{\labelcref{rule:pstore-us}}
        Then, $\inst = \Istore{\vara}{\regb}{\regc}{\pcsuc}$,
        $\stated_1 \trans{\Estore{0}}{\Dstep} \statef_1$,
        and $\statec_1 \trans{\Estore{\adrpp}}{\Dstore{\varb}{\adrppp}} \statee_1$,
        and $\poisonof{\regb} = \weakpoisoned$.
        Further $\statec_2 \trans{\Estore{\adrpp}}{\Dstore{\varb}{\adrppp}} \statee_2$.
        Due to $\poisonof{\regb} = \weakpoisoned$,
        $\rasgniof{\nstated_2}{\regbp} = 0 \in \sizeof{\vara}$
        justifies $\nstated_2 \ntrans{\Estore{0}}{\Dstep} \nstatef_2$.
        That suffices for $(\nstatec_2, \nstated_2, \poison) \sprodtrans{\Estore{\adrpp}}{\Dstore{\varb}{\adrppp}}{\Estore{0}}{\Dstep} (\nstatee_2,\nstatef_2,\poisonp)$.

        \pcasep{\labelcref{rule:pstore-hunsafe}}
        Then, $\inst = \Istore{\vara}{\regb}{\regc}{\pcsuc}$,
        $\stated_1 \trans{\Estore{\adr}}{\Dstore{\varb}{\adrp}} \statef_1$,
        and $\statec_1 \trans{\Estore{\adr}}{\Dstore{\varb}{\adrp}} \statee_1$,
        and $\poisonof{\regb} = \healthy$.
        Further $\statec_2 \trans{\Estore{\adr}}{\Dstore{\varb}{\adrp}} \statee_2$.
        Due to $\poisonof{\regb} = \healthy$,
        $\rasgniof{\stated_1}{\regbp} = \rasgniof{\statec_1}{\regb} = \adr = \rasgniof{\statec_2}{\regb} = \rasgniof{\stated_2}{\regbp} \notin \sizeof{\vara}$
        justifies $\nstated_2 \ntrans{\Estore{\adr}}{\Dstore{\varb}{\adrp}} \nstatef_2$.
        That suffices for $(\nstatec_2, \nstated_2, \poison) \sprodtrans{\Estore{\adr}}{\Dstore{\varb}{\adrp}}{\Estore{\adr}}{\Dstore{\varb}{\adrp}} (\nstatee_2,\nstatef_2,\poisonp)$.

        \pcasep{\labelcref{rule:hbranch}}
        Then, $\inst = \Iif{\regb}{\pcsuc_{\btrue}}{\pcsuc_{\bfalse}}$,
        $\stated_1 \trans{\Eif{\bvalue}}{\Dif} \statef_1$,
        and $\statec_1 \trans{\Eif{\bvalue}}{\Dif} \statee_1$,
        and $\poisonof{\regb} = \healthy$.
        Further $\statec_2 \trans{\Eif{\bvalue}}{\Dif} \statee_2$.
        Due to $\poisonof{\regb} = \healthy$,
        $(\rasgniof{\stated_1}{\regbp} \sameas 0) = (\rasgniof{\statec_1}{\regb} \sameas 0) = \bvalue = (\rasgniof{\statec_2}{\regb} \sameas 0) = (\rasgniof{\stated_2}{\regbp} \sameas 0)$
        justifies $\nstated_2 \ntrans{\Eif{\bvalue}}{\Dif} \nstatef_2$.
        That suffices for $(\nstatec_2, \nstated_2, \poison) \sprodtrans{\Eif{\bvalue}}{\Dif}{\Eif{\bvalue}}{\Dif} (\nstatee_2,\nstatef_2,\poisonp)$.

        \pcasep{\labelcref{rule:hspec}}
        Then, $\inst = \Iif{\regb}{\pcsuc_{\btrue}}{\pcsuc_{\bfalse}}$,
        $\nstated_1 \ntrans{\Eif{\bvalue}}{\Dspec} \nstated_1\lseq\statef_1$,
        and $\nstatec_1 \ntrans{\Eif{\bvalue}}{\Dspec} \nstatec_1\lseq\statee_1$,
        and $\poisonof{\regb} = \healthy$.
        Further $\nstatec_2 \ntrans{\Eif{\bvalue}}{\Dspec} \nstatec_2\lseq\statee_2$.
        Due to $\poisonof{\regb} = \healthy$,
        $(\rasgniof{\nstated_1}{\regbp} \sameas 0) = (\rasgniof{\nstatec_1}{\regb} \sameas 0) = \bvalue = (\rasgniof{\nstatec_2}{\regb} \sameas 0) = (\rasgniof{\nstated_2}{\regbp} \sameas 0)$
        justifies $\nstated_2 \ntrans{\Eif{\bvalue}}{\Dspec} \nstated_2\lseq\statef_2 = \nstatef_2$.
        That suffices for $(\nstatec_2, \nstated_2, \poison) \sprodtrans{\Eif{\bvalue}}{\Dspec}{\Eif{\bvalue}}{\Dspec} (\nstatee_2,\nstatef_2,\poisonp)$.

        \pcasep{\labelcref{rule:psfence}}
        Then, $\inst = \Isfence{\pcsuc}$,
        $\stated_1 \ntrans{\Enone}{\Dstep} \statef_1$,
        and $\statec_1 \ntrans{\Enone}{\Dstep} \statee_1$,
        and $\sizeof{\nstated_1} = \sizeof{\nstatec_1} = \sizeof{\nstatec_2} = \sizeof{\nstated_2} = 1$.
        Further $\statec_2 \ntrans{\Enone}{\Dstep} \statee_2$.
        $\sizeof{\nstated_2} = 1$ justifies $\stated_2 \ntrans{\Enone}{\Dstep} \statef_2$.
        That suffices for $(\nstatec_2, \nstated_2, \poison) \sprodtrans{\Enone}{\Dstep}{\Enone}{\Dstep} (\nstatee_2,\nstatef_2,\poisonp)$.

        \pcasep{\labelcref{rule:pslh}}
        Then, $\inst = \Islh{\rega}{\pcsuc}$,
        $\nstated_1 \ntrans{\Enone}{\Dstep} \nstatef_1$,
        and $\nstatec_1 \ntrans{\Enone}{\Dstep} \nstatee_1$.
        Further $\nstatec_2 \ntrans{\Enone}{\Dstep} \nstatee_2$.
        Semantics yield an appropriate $\nstatef_2$ with $\nstated_2 \ntrans{\Enone}{\Dstep} \nstatef_2$.
        That suffices for $(\nstatec_2, \nstated_2, \poison) \sprodtrans{\Enone}{\Dstep}{\Enone}{\Dstep} (\nstatee_2,\nstatef_2,\poisonp)$.

        \pcasep{\labelcref{rule:prb}}
        Then, $\nstated_1 \ntrans{\Erb}{\Drb} \nstatedp_1 = \nstatef_1$,
        and $\nstatec_1 \ntrans{\Erb}{\Drb} \nstatecp_1 = \nstatee_1$.
        Further $\nstatec_2 \ntrans{\Erb}{\Drb} \nstatecp_2 = \nstatee_2$.
        Semantics yield $\nstated_2 \ntrans{\Erb}{\Drb} \nstatedp_2 = \nstatef_2$.
        That suffices for $(\nstatec_2, \nstated_2, \pseq\lseq\poison) \sprodtrans{\Erb}{\Drb}{\Erb}{\Drb} (\nstatee_2,\nstatef_2,\pseq)$.

        \pcasep{\labelcref{rule:pfill,,rule:pspill,,rule:pmove,,rule:pshuf-sfence,,rule:pshuf-slh}} Not applicable due to definition of $\simrel$, $\nstated_1 \simrel \nstatec_1$ makes it impossible for $\nstated_1$ to be at a shuffle program counter. \qedhere
    \end{pcases}
\end{proof}

\subsection*{\Cref{lem:repeat-shuffle}}

\begin{proof}
    We do the induction step by case distinction on $\nstatef_1 \ntrans{\leak}{\Dstep} \nstateh_1$.
    Let $\inst = \instof{\nstatef_1} = \instof{\nstatef_2}$.

    \pcase{\labelcref{rule:move,,rule:nslh}}
    Then, $\leak = \Enone$ and there is $\nstatef_2 \ntrans{\Enone}{\Dstep} \nstateh_2$ by definition of semantics.

    \pcase{\labelcref{rule:fill}}
    Then, $\leak = \Eload{\stackl}$ and there is $\nstatef_2 \ntrans{\Eload{\stackl}}{\Dstep} \nstateh_2$ by definition of semantics.

    \pcase{\labelcref{rule:spill}}
    Then, $\leak = \Estore{\stackl}$ and there is $\nstatef_2 \ntrans{\Estore{\stackl}}{\Dstep} \nstateh_2$ by definition of semantics.

    \pcase{\labelcref{rule:nsfence}}
    Then, $\sizeof{\nstatef_1} = 1 = \sizeof{\nstatef_2}$.
    By definition, $\nstatef_2 \ntrans{\Enone}{\Dstep} \nstateh_2$.
\end{proof}

\end{document}